\documentclass[acmsmall,screen,authorversion,nonacm]{acmart} 
\acmJournal{PACMPL}
\acmVolume{1}
\acmNumber{CONF} 
\acmArticle{1}
\acmYear{2018}
\acmMonth{1}
\acmDOI{} 
\startPage{1}

\setcopyright{cc}
\setcctype[4.0]{by}

\usepackage{booktabs,tablefootnote}
\usepackage{subcaption}
\usepackage{enumitem}
\usepackage{letltxmacro}
\usepackage[colorinlistoftodos,bordercolor=orange,backgroundcolor=orange!20,linecolor=orange,textsize=footnotesize]{todonotes}
\usepackage{pifont}
\usepackage[abbreviations,xlatin]{foreign}
\usepackage[super]{nth}

\usepackage{hyperref}
\usepackage[capitalise, noabbrev]{cleveref}

\AddToHook{cmd/appendix/before}{%
    \crefalias{section}{appendix}%
    \crefalias{subsection}{appendix}
}

\usepackage{listings}

\definecolor{dkblue}{rgb}{0,0.1,0.5}
\definecolor{lightblue}{rgb}{0,0.5,0.5}
\definecolor{dkgreen}{rgb}{0,0.4,0}
\definecolor{dk2green}{rgb}{0.4,0,0}
\definecolor{dkviolet}{rgb}{0.6,0,0.8}
\definecolor{dkpink}{rgb}{0.2,0,0.6}
\definecolor{commentcolor}{rgb}{0.2,0.4,0.4}
\lstdefinelanguage{SSR} {
mathescape=true,
escapeinside={[*}{*]},
texcl=false,
morekeywords=[1]{
From, Section, Module, End, Require, Import, Export, Defensive, Function,
Variable, Variables, Parameter, Parameters, Axiom, Hypothesis, Hypotheses,
Context,
Notation, Local, Global, Tactic, Reserved, Scope, Open, Close, Bind, Delimit,
Definition, Let, Ltac, Fixpoint, CoFixpoint, Add, Morphism, Relation,
Implicit, Arguments, Types, Set, Unset, Contextual, Strict, Prenex, Implicits,
Inductive, CoInductive, Record, Variant, Structure, Canonical, Coercion,
Theorem, Lemma, Corollary, Proposition, Fact, Remark, Example,
Proof, Goal, Save, Qed, Defined, Hint, Resolve, Rewrite, View,
Search, Show, Print, Printing, All, Graph, Primitive, Projections, inside,
outside, Locate, Maximal,
Parametricity, Translation},
morekeywords=[2]{forall, exists, exists2, fun, fix, cofix, struct,
      match, with, end, as, in, return, let, if, is, then, else,
      for, of, nosimpl},
morekeywords=[3]{Type, Prop},
morekeywords=[4]{
         pose, set, move, case, elim, clear,
            hnf, intro, intros, generalize, rename, pattern, after,
	    destruct, induction, using, refine, inversion, injection,
         rewrite, congr, unlock, lazy, compute, vm_compute, native_compute, ring, field,
            replace, fold, unfold, change, cutrewrite, simpl,
         have, gen, generally, suff, wlog, suffices, without, loss, nat_norm,
            assert, cut, trivial, revert, bool_congr, nat_congr, abstract,
	 symmetry, transitivity, auto, split, left, right, autorewrite},
morekeywords=[5]{
         by, done, exact, reflexivity, tauto, romega, omega,
         assumption, solve, contradiction, discriminate, congruence},
morekeywords=[6]{do, last, first, try, idtac, repeat},
morekeywords=[7]{Fail},
 literate=*
	{->}{{$\rightarrow\,$}}2
	{<-}{{$\leftarrow\,$}}2
 	{>->}{{$\rightarrowtail$}}2
 	{<=}{{$\leq$}}2
 	{>=}{{$\geq$}}2
 	{<>}{{$\neq$}}2
 	{/\\}{{$\wedge$}}2
 	{\\/}{{$\vee$}}2
	{\\o}{{$\circ$}}1
 	{<->}{{$\leftrightarrow\;$}}3
 	{<=>}{{$\Leftrightarrow\;$}}3
 	{:nat}{{$~\in\mathbb{N}$}}3
	{fforall\ }{{$\forall_f\,$}}1
	{forall\ }{{$\forall\,$}}2
 	{negb}{{$\neg$}}1
 	{spp}{{:*:\,}}1
 	{\\in}{{$\in$}}1
 	{\\notin}{{$\not\in$}}1
 	{/\\}{$\land\,$}1
 	{:*:}{{$*$}}2
	{=>}{{$\Rightarrow$}}2
	{==>}{\texttt{==>}}3
	{:=}{{$\coloneqq$}}2
 	{!=}{{$\neq$}\,}2
 	{^-1}{{$^{-1}$}}1
 	{elt'}{elt'}1
	{isn't }{{{\ttfamily\color{dkgreen} isn't }}}1,
showstringspaces=false,
morecomment=[s]{(*}{*)},
tabsize=3,
extendedchars=true,
sensitive=true,
identifierstyle={\ttfamily\color{black}},
keywordstyle=[1]{\ttfamily\color{dkviolet}},
keywordstyle=[2]{\ttfamily\color{dkgreen}},
keywordstyle=[3]{\ttfamily\color{lightblue}},
keywordstyle=[4]{\ttfamily\color{dkblue}},
keywordstyle=[5]{\ttfamily\color{red}},
keywordstyle=[6]{\ttfamily\color{dkpink}},
keywordstyle=[7]{\ttfamily\color{red}},
stringstyle=\ttfamily,
commentstyle={\ttfamily\color{commentcolor}},
}

\lstdefinestyle{plain}{
  basicstyle=\ttfamily,
  keywordstyle=,
  identifierstyle=,
  commentstyle=,
  stringstyle=,
  emphstyle=,
  backgroundcolor=,
  language=,
  frame=,
  framesep=0pt,
  rulecolor=\color{black},
  numbers=none,
  numberstyle=,
  xleftmargin=0pt,
  xrightmargin=0pt,
  basewidth=0.495em,
  lineskip=-.4ex,
  keepspaces,
  identifierstyle={\ttfamily\color{black}},
  keywordstyle={\ttfamily\color{dkviolet}},
  stringstyle=\ttfamily,
  commentstyle={\ttfamily\color{black!60!white}},
}

\lstdefinelanguage[]{ocaml}[Objective]{caml}{
  literate=* {:=}{{$\coloneqq$}}2
}

\lstnewenvironment{code}[1][]{
  \lstset{
  breaklines=false,
  #1}
}{}

\lstnewenvironment{coqcode}[1][]{
  \lstset{,
  language=SSR,
  breaklines=false
  #1}
}{}

\lstnewenvironment{camlcode}[1][]{
  \lstset{
  language={ocaml},
  breaklines=false,
  mathescape=true,
  #1}
}{}

\newcommand*{\SavedLstInline}{}
\LetLtxMacro\SavedLstInline\lstinline
\DeclareRobustCommand*{\lstinline}{%
  \ifmmode
    \let\SavedBGroup\bgroup
    \def\bgroup{%
      \let\bgroup\SavedBGroup
      \hbox\bgroup
    }%
  \fi
  \SavedLstInline
}

\newcommand\coq%
{\lstinline[language=SSR, breaklines=true, breakatwhitespace=true]}

\newcommand\caml%
{\lstinline[language={ocaml}, breaklines=true, breakatwhitespace=true, mathescape=true]}

\newcommand\haskell%
{\lstinline[breaklines=true, breakatwhitespace=true]}

\makeatletter

\newcommand*{\lstitem}[1][]{%
  \setbox0\hbox\bgroup
    \patchcmd{\lst@InlineM}{\@empty}{\@empty\egroup\item[\usebox0]\leavevmode\ignorespaces}{}{}%
    \lstinline[#1]%
}
\makeatother

\usepackage{mathtools}

\newcommand{\bigO}{\mathcal{O}}
\newcommand\concat{\ensuremath{\mathbin{+\mkern-6mu+}}}
\newcommand{\merge}{\mathbin{\land\hspace{-.45em}\land}}
\newcommand{\mergerev}{\mathbin{\lor\hspace{-.45em}\lor}}
\newcommand{\permeq}{\mathrel{=_\texttt{perm}}}
\newcommand\xmark{\text{\ding{55}}\xspace}
\newcommand\cmark{\text{\ding{51}}\xspace}
\newcommand{\sort}{\texttt{sort}}
\newcommand{\asort}{\texttt{asort}}
\newcommand{\Tsort}{\mathrm{T}_{\sort{}}}
\newcommand{\Tasort}{\mathrm{T}_{\asort{}}}
\newcommand{\TasortExt}{\mathrm{T}'_{\asort{}}}
\newcommand{\types}{\ensuremath{\mathcal{U}}}

\newcommand{\paramt}[1]{\ensuremath{\llbracket\,#1\,\rrbracket}}

\AtEndPreamble{%
  \theoremstyle{acmdefinition}
  \newtheorem{remark}[theorem]{Remark}
  \newtheorem{assumption}[theorem]{Assumption}}

\AddToHook{env/assumption/begin}{\crefalias{theorem}{assumption}}
\AddToHook{env/corollary/begin}{\crefalias{theorem}{corollary}}
\AddToHook{env/definition/begin}{\crefalias{theorem}{definition}}
\AddToHook{env/lemma/begin}{\crefalias{theorem}{lemma}}
\AddToHook{env/remark/begin}{\crefalias{theorem}{remark}}

\usepackage{xspace}

\newcommand{\Coq}{{\sffamily\scshape Rocq}\xspace}

\newcommand{\MC}{{\sffamily\scshape MathComp}\xspace}
\newcommand{\Paramcoq}{{\sffamily\scshape Paramcoq}\xspace}

\newcommand{\Isabelle}{{\sffamily\scshape Isabelle}\xspace}
\newcommand{\IsabelleHOL}{{\sffamily\scshape Isabelle/HOL}\xspace}
\newcommand{\Dafny}{{\sffamily\scshape Dafny}\xspace}
\newcommand{\KeY}{{\sffamily\scshape KeY}\xspace}

\newcommand{\OCaml}{{\sffamily OCaml}\xspace}

\newcommand{\Java}{{\sffamily Java}\xspace}
\newcommand{\GHC}{{\sffamily GHC}\xspace}

\newcommand{\wrt}{w.r.t.\@\xspace}

\begin{document}

\title{A bargain for mergesorts}
\subtitle{How to prove your mergesort correct and stable, almost for free}

\author{Cyril Cohen}
\orcid{0000-0003-3540-1050}
\affiliation{\institution{Inria, CNRS, ENS de Lyon, UCBL, LIP, UMR 5668} \country{France}}
\email{Cyril.Cohen@inria.fr}

\author{Kazuhiko Sakaguchi}
\orcid{0000-0003-1855-5189}
\affiliation{\institution{CNRS, ENS de Lyon, UCBL, LIP, UMR 5668} \country{France}}
\email{kazuhiko.sakaguchi@cnrs.fr}

\begin{abstract}
 We present a novel characterization of stable mergesort functions using relational parametricity,
 and show that it implies the functional correctness of mergesort.
 As a result, one can prove the correctness of several variations of mergesort (\eg, top-down,
 bottom-up, tail-recursive, non-tail-recursive, smooth, and non-smooth mergesorts) by proving the
 characteristic property for each variation.
 Thanks to our characterization and the parametricity translation, we deduced the correctness
 results, including stability, of various implementations of mergesort for lists, including highly
 optimized ones, in the \Coq Prover (formerly the {\sffamily\scshape Coq} Proof Assistant).
\end{abstract}



\maketitle

\section{Introduction}
\label{sec:introduction}

Sorting is ubiquitous in computing. Over the decades, several algorithms, optimization techniques,
and heuristics to solve it efficiently have been developed.
Among these, mergesort achieves both optimal $\bigO(n \log n)$ worst-case time complexity and
stability---the property that a sort algorithm always preserves the order of equivalent
elements---and is thus often preferred, particularly for sorting a linked list.
On the other hand, defining a sort function and proving its correctness is often a good example to
demonstrate some techniques (\eg, how to use advanced recursion and induction) in
interactive theorem provers such as \Coq~\cite{Leroy:mergesort}\cite[Section 7.1]{DBLP:books/daglib/0035083}
and \Isabelle~\cite{DBLP:journals/jar/Sternagel13,functional_algorithms_verified}.
However, these case studies consider only one or a few simple implementations of mergesort and do
not scale well to intricate implementations, such as smooth (\cref{sec:smooth-mergesort})
tail-recursive (\cref{sec:tailrec-mergesort}) mergesort, since a methodology to share a large part
of the functional correctness proofs between several variants of a sort algorithm has not been
studied.
In this paper, we address this issue by introducing a characterization of mergesort functions, which
is easy to prove and implies several correctness results of mergesort.

The intuition behind our characterization of mergesort is as follows: by replacing all the
occurrences of the merge function with concatenation, any stable mergesort function can be turned
into the identity function. If it is not the identity function but may permutate some elements, the
mergesort function is unstable (\cref{sec:characterization}).
In order to make sure that this replacement is done in the intended way, we first abstract out the
mergesort function to take a type parameter representing sorted lists and operators on them, \eg,
merge, singleton, and empty.
We define the \emph{characteristic property} of mergesort functions as the existence of such a
corresponding abstract mergesort function satisfying the following conditions
(\cref{sec:characterization,sec:new-characterization}):
\begin{enumerate}
 \item \label{item:asort_mergeE}it is the mergesort function when instantiated with ``merge'';
       that is, the abstract mergesort instantiates back to the concrete mergesort,
 \item \label{item:asort_catE}it is the identity function when instantiated with concatenation;
       that is, recursively dividing the input list and concatenating them, instead of merging
       them, gives us the input list, and
 \item \label{item:asort_R}it is relationally parametric~\cite{DBLP:conf/ifip/Reynolds83}
       (\cref{sec:prelim-param}).
\end{enumerate}
Our correctness proof of mergesort is twofold.
Firstly, we prove the characteristic property for each mergesort function we wish to verify
(\cref{sec:implementation-to-characterization,sec:new-characterization-tailrec,sec:new-characterization-smooth}).
Among the three conditions above, Equations \eqref{item:asort_mergeE} and \eqref{item:asort_catE}
can be proved by functional induction and equational reasoning using basic facts about lists, for
all the variants of mergesort presented in the paper, and \eqref{item:asort_R} follows from the
abstraction theorem~\cite{DBLP:conf/ifip/Reynolds83} for all ground implementations. Furthermore,
\eqref{item:asort_mergeE} holds by definition in the simplest cases.
Secondly, we prove correctness results for any mergesort function, whose sole assumption is the
characteristic property
(\cref{sec:characterization-to-correctness,sec:new-characterization-to-correctness}).
The characteristic property implies an induction principle over \emph{traces}---binary trees
reflecting the underlying divide-and-conquer structure of mergesort---to reason about the relation
between the input and output of mergesort (\cref{lemma:sort_ind,lemma:sort_ind'}), and the
naturality of sort (\cref{lemma:sort_map}).
These two consequences are sufficient to deduce several correctness results of mergesort, including
stability.

While we first present a simplified version of characterization and proofs that work only for
non-tail-recursive mergesort (\cref{sec:nontailrec}), we later extend them (\cref{sec:new-proofs})
to the following optimization techniques for mergesort.
\begin{itemize}
 \item Tail-recursive mergesort (\cref{sec:tailrec-mergesort}) in call-by-value evaluation, \eg,
       \caml{List.stable_sort} of the \OCaml standard library~\cite{OCaml}, has the advantage that
       it does not use up stack space and thus is efficient.
 \item Smooth mergesort (\cref{sec:smooth-mergesort}) reuses sorted slices in the input in the
       sorting process and is efficient on almost sorted input.
       For example, \haskell{Data.List.sort} of the \GHC (Glasgow Haskell
       Compiler) libraries~\cite{GlasgowHaskell} is a smooth non-tail-recursive mergesort.
\end{itemize}
Note that non-tail-recursive mergesort still has the advantage over tail-recursive mergesort that it
allows us to compute a first few elements of the output incrementally in call-by-need evaluation.
The problem of obtaining the $k$ smallest items of a list of length $n$ is called \emph{partial
sorting}, and its online version that allows us to stop the sorting process and resume it later to
increase $k$ is called \emph{incremental sorting}~\cite{DBLP:conf/alenex/ParedesN06}.
In fact, non-tail-recursive mergesort achieves the optimal time complexity $\bigO(n + k \log k)$ of
these problems.
Thus, the optimal mergesort function depends on the evaluation strategy and whether one wants to use
it as a partial or incremental sort.
From this perspective, the advantage of our verification approach is that it allows us to verify
both optimal implementations of mergesort modularly.

We formalized our functional correctness proofs in the \Coq Prover~\cite{rocqrefman}.
In \cref{sec:formalization}, we discuss two technical aspects of our formalization.
Firstly, we review a technique to convince \Coq that bottom-up non-tail-recursive
mergesort is terminating (\cref{sec:nontailrec-mergesort-in-coq}), which does not require explicit
termination proofs (\ie, no use of well-founded recursion), but is done by making the mergesort
function structurally recursive~\cite{DBLP:conf/types/Gimenez94} without the use of artificial
decreasing argument (fuel)~\cite{Gonthier:2009}.
Avoiding the use of well-founded recursion has two practical advantages:
1) it makes execution of mergesort inside the \Coq kernel easier and more efficient, and
2) it makes application of the parametricity translation~\cite{DBLP:conf/fossacs/BernardyL11,
DBLP:journals/jfp/BernardyJP12, DBLP:conf/csl/KellerL12} of the \Paramcoq plugin~\cite{paramcoq} to
the abstract mergesort function straightforward.
Furthermore, this technique of making mergesort structurally recursive allows us to implement
bottom-up tail-recursive mergesort (\cref{sec:tailrec-mergesort-in-coq}), and both bottom-up
non-tail-recursive and tail-recursive mergesorts can be easily made smooth
(\cref{appx:nontailrec-mergesort-in-coq,appx:tailrec-mergesort-in-coq}) in contrast to top-down
mergesorts, \eg, \cref{sec:tailrec-mergesort} and \caml{List.stable_sort} of \OCaml.
Up to our knowledge, \cref{appx:tailrec-mergesort-in-coq} provides the first implementation of
smooth tail-recursive mergesort for lists in functional programming
(\cref{table:classify-mergesorts}).
Secondly, we discuss the design and organization of the library (\cref{sec:interface}), particularly,
the interface for mergesort functions (\cref{sec:the-interface}) which allows us to state our
correctness lemmas polymorphically for any stable mergesort functions
(\cref{sec:overloaded-correctness-lemmas}).
To construct an instance of this interface for a concrete mergesort function, we prove
\eqref{item:asort_mergeE} and \eqref{item:asort_catE} by induction and equational reasoning, and use
\Paramcoq to automatically prove \eqref{item:asort_R}
(\cref{sec:populating-the-interface}).

Our approach currently has two limitations (\cref{sec:limitations}).
Firstly, the proof of \eqref{item:asort_mergeE} cannot be done by definition (by computation in
\Coq) for elaborate variations of mergesort, such as tail-recursive (\cref{sec:tailrec-mergesort}),
smooth (\cref{sec:smooth-mergesort}), and structurally-recursive
(\cref{sec:nontailrec-mergesort-in-coq,sec:tailrec-mergesort-in-coq}) mergesorts.
This is because these mergesorts inspect sorted lists by operations not supplied to the abstract
mergesort function, and we have to abstract out these mergesorts in a non-trivial way.
Since we are primarily interested in structurally-recursive mergesorts in our formalization, the
proof of \eqref{item:asort_mergeE} has to be done by hand for most variations.
Secondly, our approach does not support the use of $n$-way merge function that merges $n > 2$ lists
and cannot be expressed as a simple combination of the 2-way merge, \eg, \haskell{Data.List.sort} of
\GHC 9.12.1 and later which uses 3-way and 4-way merge functions.
In this paper, ``\haskell{Data.List.sort} of \GHC'' generally refers to its old version.

\begin{table}[t]
 \newcommand{\yes}{\textcolor{ACMGreen}{\cmark}\xspace}
 \newcommand{\maybe}{\textcolor{ACMPurple}{\ding{46}}\xspace}
 \newcommand{\no}{\textcolor{ACMRed}{\xmark}\xspace}
 \caption{Classification of mergesort functions and their optimization techniques, by whether they
 are top-down or bottom-up, tail-recursive or not, and smooth or not. \yes means that the
 classification applies to \emph{one of} the given implementations, and \no means that the
 classification does not apply.
 Incompatible classifications, \eg, top-down and bottom-up, may apply to the same row,
 since one row may correspond to a few implementations (the number indicated in the ``\#
 implem.'' column).
 Since only bottom-up mergesorts can be made smooth, ``non-smooth'' and ``smooth'' appear as
 the sub-columns of the ``bottom-up'' column.}
 \label{table:classify-mergesorts}
 \footnotesize
 \begin{tabular}{p{4cm} c c c c c c}
  \hline
  & \# implem. & top-down & \multicolumn{2}{c}{bottom-up} & non-tail-rec. & tail-rec. \\
  & & & non-smooth & smooth & & \\
  \hline
  \cref{sec:mergesort-approaches}
  & 2 & \yes & \yes & \no  & \yes & \no  \\
  \cref{sec:tailrec-mergesort}
  & 1 & \yes & \no  & \no  & \no  & \yes \\
  \cref{sec:smooth-mergesort}
  & 1 & \no  & \no  & \yes & \yes & \no  \\
  \cref{sec:nontailrec-mergesort-in-coq}
  & 1 & \no  & \yes & \no  & \yes & \no  \\
  \cref{sec:tailrec-mergesort-in-coq}
  & 1 & \no  & \yes & \no  & \no  & \yes \\
  \cref{appx:nontailrec-mergesort-in-coq} (incl.\@ \cref{sec:nontailrec-mergesort-in-coq})
  & 4 & \no  & \yes & \yes & \yes & \no  \\
  \cref{appx:tailrec-mergesort-in-coq} (incl.\@ \cref{sec:tailrec-mergesort-in-coq})
  & 4 & \no  & \yes & \yes & \no  & \yes \\
  \OCaml (\texttt{List.stable\_sort})
  & 1 & \yes & \no  & \no  & \no  & \yes \\
  \GHC (\texttt{Data.List.sort}) and \citet{DBLP:journals/jar/Sternagel13,DBLP:journals/tocl/LeinoL15}
  & 1 & \no  & \no  & \yes & \yes & \no  \\
  \citet{functional_algorithms_verified} (incl.\@ \citet{DBLP:journals/jar/Sternagel13})
  & 3 & \yes & \yes & \yes & \yes & \no  \\
  \citet{DBLP:books/daglib/0035083}
  & 1 & \yes & \no  & \no  & \yes & \no  \\
  \citet{Leroy:mergesort}
  & 1 & \no  & \yes & \no  & \yes & \no  \\
  \hline
 \end{tabular}
 \let\yes\undefined
 \let\no\undefined
 \let\maybe\undefined
\end{table}
To summarize, \cref{table:classify-mergesorts} classifies the mergesort functions and their
optimization techniques presented in this paper and related work.

\paragraph{Disclaimer}
While we extensively use the algorithmic complexities and performance of mergesort as the
\emph{motivation} of the present work (especially \cref{sec:optimizations}), the contribution of the
paper is the functional correctness proofs, and we decided to neither informally nor formally show
any complexity results of mergesort, beyond describing the well-known fact that the worst-case time
complexity of mergesort is $\bigO(n \log n)$ (\cref{sec:mergesort-approaches}).
Our rationale is that
1) such an addition would obscure the contribution of the paper,
2) our characterization does not rule out sorting algorithms slower than $\bigO(n \log n)$, \eg,
insertion sort (\cref{def:insertion-sort}), and
3) the only case where we improve the worst-case time complexity beyond $\bigO(n \log n)$ is
non-tail-recursive mergesort in call-by-need evaluation as optimal incremental sorting.

\paragraph{Outline}
The paper is organized as follows.
In \cref{sec:preliminaries}, we introduce the \OCaml functions used in the paper for manipulating
lists (\cref{sec:prelim-list}) and parametricity (\cref{sec:prelim-param}).
In \cref{sec:nontailrec}, we present a simplified version of our characterization
(\cref{sec:characterization}) and correctness proofs (\cref{sec:characterization-to-correctness})
that work only for non-tail-recursive mergesort.
As examples considered in this section, we review top-down and bottom-up mergesort whose underlying
traces have different shapes (\cref{sec:mergesort-approaches}), and describe how to prove the
characteristic property on them (\cref{sec:implementation-to-characterization}).
In \cref{sec:optimizations}, we review two optimization techniques for mergesort, namely,
tail-recursive mergesort (\cref{sec:tailrec-mergesort}) and smooth mergesort
(\cref{sec:smooth-mergesort}), and extend the characterization and proof technique presented in
\cref{sec:nontailrec} to these optimization techniques (\cref{sec:new-proofs}).
In \cref{sec:formalization}, we discuss technical aspects of our formalization. We show how to make
mergesort structurally recursive (\cref{sec:termination}) and discuss the design and organization of
the library (\cref{sec:interface}).
In \cref{sec:limitations}, we discuss the limitations of our approach.
In \cref{sec:related-work}, we discusses related work before concluding the paper in
\cref{sec:conclusion}.

Because of the structure of our correctness proofs, the paper roughly consists of two parts:
sections concerning concrete mergesorts, including optimization techniques and the proofs that they
satisfy our characteristic property (\cref{sec:mergesort-approaches,%
sec:implementation-to-characterization,sec:tailrec-mergesort,sec:smooth-mergesort,%
sec:new-characterization-tailrec,sec:new-characterization-smooth,sec:termination}),
and sections concerning the characteristic property and how it solely implies several correctness
results (\cref{sec:characterization,sec:characterization-to-correctness,sec:new-characterization,%
sec:new-characterization-to-correctness,sec:interface}).
In order to present, extend, and formalize our proof technique incrementally, these two kinds of
sections appear alternately in the paper. The former sections are marked with \dag{} to guide the
readers.

\paragraph{Appendices}

\Cref{appx:basic-definitions-and-facts} provides a list of basic definitions and lemmas in the \MC
library~\cite{mathcomp} used for the proofs in
\cref{sec:implementation-to-characterization,sec:characterization-to-correctness,%
sec:new-characterization-to-correctness,appx:stablesort-theory}.
While we refer definitions and lemmas from this appendix in proofs for the sake of preciseness,
these references come with enough extra information, \eg, \cref{lemma:coq:catA} (associativity of
$\concat$), to be safely ignored.
\Cref{appx:stablesort-theory} provides the list of all lemmas about stable sort functions solely
derived from the characterization, their formal statements in \Coq, and their informal proofs.
In \cref{appx:stability-statements}, we compare the statements of our stability results with ones in
the literature~\cite{CLRS4th, Leroy:mergesort, DBLP:journals/jar/Sternagel13,
  DBLP:journals/tocl/LeinoL15} and argue that the former is slightly more general than the latter.
In \cref{appx:mergesort-in-coq}, we present the full \Coq implementations of structurally-recursive
non-tail-recursive and tail-recursive mergesorts including smooth ones, since \cref{sec:termination}
presents only the non-smooth ones in the \OCaml syntax.

\section{Preliminaries}
\label{sec:preliminaries}

Throughout the paper except appendix, mergesort functions are presented in the \OCaml syntax.
This preliminary section introduces the \OCaml functions used in the paper for manipulating lists
(\cref{sec:prelim-list}) and binary relational parametricity (\cref{sec:prelim-param}).

Without loss of generality, we consider only sorting functions that
take a ``less than or equal'' relation of type $T \to T \to \mathrm{bool}$, which we denote by
$\leq$ or \caml{(<=)}, and sort items in ascending order.

\subsection{Functons from the \caml{List} library}
\label{sec:prelim-list}

The mergesort functions presented in the paper use the following \OCaml functions for manipulating
lists.
These functions are defined in the \caml{List} module of the standald library, except for
\caml{List.split_n} defined in the core standard library overlay~\cite{janestreet-core}.
\begin{description}[style=nextline]
 \lstitem[language={[Objective]caml}]{val length : 'a list -> int}
   \caml{List.length l} returns the length (number of elements) of \caml{l}.
 \lstitem[language={[Objective]caml}]{val rev : 'a list -> 'a list}
   List reversal.
 \lstitem[language={[Objective]caml}]{val append : 'a list -> 'a list -> 'a list}
   Concatenates two lists.
 \lstitem[language={[Objective]caml}]{val rev_append : 'a list -> 'a list -> 'a list}
   \caml{List.rev_append l1 l2} reverses \caml{l1} and concatenates it with \caml{l2}.
 \lstitem[language={[Objective]caml}]{val flatten : 'a list list -> 'a list}
   \caml{List.flatten} returns the list obtained by flattening (folding by concatenation) the given
   list of lists.
 \lstitem[language={[Objective]caml}]{val split_n : 'a list -> int -> 'a list * 'a list}
   \caml{List.split_n l n} splits \caml{l} into two lists of the first \caml{n} elements and the
   remaining, and returns the pair of them.
 \lstitem[language={[Objective]caml}]{val map : ('a -> 'b) -> 'a list -> 'b list}
   \caml{List.map f [a1; ...; an]} returns the list \caml{[f a1; ...; f an]}.
 \lstitem[language={[Objective]caml}]{val filter : ('a -> bool) -> 'a list -> 'a list}
   \caml{List.filter p l} returns the list collecting all the elements of \caml{l} that satisfy the
   predicate \caml{p}, and preserves the order of the elements in the input.
\end{description}
Hereinafter, we omit the prefix \caml{List} for the above functions, assuming that the \caml{List}
module is open.
Among these, we assume that \caml{length}, \caml{rev}, and \caml{rev_append} are tail recursive, and
\caml{append}, \caml{flatten}, \caml{split_n}, \caml{map} and \caml{filter} are not, as in their
standard implementations.
The basic definitions and facts in \Coq, including the \Coq counterpart of the above functions, used
for our correctness proofs are listed in \cref{appx:basic-definitions-and-facts}.

We use the following notations in the part of the paper concerning proofs:
\begin{align*}
 \texttt{map}_f
 &\coloneq \caml{map f},
 &[f \, x \mid x \leftarrow \mathit{xs}]
 &\coloneq \texttt{map}_f \, \mathit{xs}, \\
 \texttt{filter}_p
 &\coloneq \caml{filter p},
 &[x \mid x \leftarrow \mathit{xs}, p \, x]
 &\coloneq \texttt{filter}_p \, \mathit{xs}, \\
 \mathit{xs} \concat \mathit{ys}
 &\coloneq \caml{append} \, \mathit{xs} \, \mathit{ys},
 &[f \, x \mid x \leftarrow \mathit{xs}, p \, x]
 &\coloneq \texttt{map}_f \, (\texttt{filter}_p \, \mathit{xs}).
\end{align*}

\subsection{Parametricity}
\label{sec:prelim-param}

We assume the existence of an interpretation function $\paramt{\cdot}$ from types in $\types{}$ (resp.~type families ranging in~$\types$) of
our calculus to (resp.~families of) relations on terms, which is called a binary parametricity translation. In this paper, we generally abbreviate ``binary parametricity'' to ``parametricity''.
Arrow types $S \to T$ are interpreted as morphisms for relations $\paramt{S}$ and $\paramt{T}$, \ie,
given $f_1, f_2 : S \to T$, we have $(f_1, f_2) \in \paramt{S \to T}$ iff
$\forall x_1 \, x_2, (x_1, x_2) \in \paramt{S} \Rightarrow (f_1 \, x_1, f_2 \, x_2) \in \paramt{T}$.
Finally, polymorphic types are interpreted as follows: given
$f_1, f_2 : (\forall S : \types, T)$, we have $(f_1, f_2) \in \paramt{\forall S : \types, T}$ iff
$\forall \left(S_1 \, S_2 : \types\right) ({\sim_S} \subseteq S_1 \times S_2), (f_1 \, S_1, f_2 \, S_2) \in \paramt{T}$,
where $\paramt{T}$ interprets $S$ as $\paramt{S} \coloneq {\sim_S}$.
For a full account on parametric interpretations and parametricity translations in this style, we
refer to \citet{DBLP:conf/ifip/Reynolds83} and \citet{DBLP:conf/fpca/Wadler89}.

\begin{definition}[Parametricity]
 \label{def:parametricity}
 We say a term $t$ of type $T$ is parametric if $(t, t) \in \paramt{T}$.
\end{definition}

We also assume that the interpretation of ground types (\eg, $\paramt{\mathrm{bool}}$ and
$\paramt{\mathrm{nat}}$) is equality, and that
$\paramt{\mathrm{list} \, T} \coloneq \paramt{\mathrm{list}}_{\paramt{T}}$
is the pointwise lifting of relation $\paramt{T}$ to lists, defined as follows:
$\left([x_1, \dots, x_n], [y_1, \dots, y_m]\right) \in \paramt{\mathrm{list}}_\sim$ iff the two
lists have the same length $n = m$ and they are pairwise related, \ie, $x_i \sim y_i$ for any $1 \leq i \leq n$.



In \cref{sec:nontailrec,sec:optimizations}, we assume we work in a fragment of System F, extended with recursors, which satisfies the abstraction theorem.
\begin{assumption}[Abstraction theorem]
 \label{assumption:abstraction}
 Closed terms in the empty context are parametric.
\end{assumption}
\noindent Relations obtained by the interpretation $\paramt{\cdot}$ cannot be expressed in such a
calculus, and thus, we consider them as expressed in the meta language.
However, in \cref{sec:interface}, we formalize our arguments in \Coq where relations are
interpreted as functions $S \to T \to \types{}$, which justifies the abuse of $\forall$ notation for
polymorphic types in the calculus, \eg, $\forall T : \types, S \to V$, and universally quantified
formulas in the meta language, \eg, $\forall T : \types, \phi \Rightarrow \psi$.


\section{Non-tail-recursive mergesorts}
\label{sec:nontailrec}

In this section, we present a simplified version of the characteristic property
(\cref{sec:characterization}) and the correctness proofs
(\cref{sec:implementation-to-characterization,sec:characterization-to-correctness}) that work only
for non-tail-recursive mergesort.
As variations of mergesort covered by this section, we consider top-down and bottom-up
non-tail-recursive mergesorts (\cref{sec:mergesort-approaches}).
The proofs are twofold. Firstly, we prove the characteristic property for each mergesort function
(\cref{sec:implementation-to-characterization}). Secondly, we prove that any mergesort function
satisfying the characteristic property is correct (\cref{sec:characterization-to-correctness}).
We will later extend our approach (\cref{sec:new-characterization}) to support tail-recursive
mergesorts (\cref{sec:tailrec-mergesort}) and smooth mergesorts (\cref{sec:smooth-mergesort}).

\subsection{Top-down and bottom-up approaches$^\dag$}
\label{sec:mergesort-approaches}

Mergesort is a sort algorithm that works by merging sorted sequences.
Merging two input lists \caml{xs} and \caml{ys} sorted \wrt (with regard to) a total preorder \caml{(<=)}, \ie, a total
and transitive binary relation, obtains another sorted list which contains each element of
the input lists exactly one time, as in \cref{fig:nontailrec-merge}.
In order to sort an arbitrary input list, mergesort has to divide the input into sorted slices, \eg,
singleton lists, and recursively merge them. There are two approaches to do so: \emph{top-down} and
\emph{bottom-up}.
The top-down approach implemented in \cref{fig:top-down-mergesort}
(1)~divides the input into two lists of roughly the same length,
(2)~recursively sorts each half, and
(3)~merges two sorted halves.
The input is eventually divided into singleton lists, and recursion stops there.
The bottom-up approach implemented in \cref{fig:bottom-up-mergesort}
(1)~divides the input into singleton lists,
(2)~obtains sorted lists of length 2, then 4, then 8, and so on, by merging each adjacent pair, and
(3)~eventually obtains one sorted list.

\begin{figure}[t]
\begin{subfigure}[t]{.5\textwidth}
\begin{camlcode}
let rec merge (<=) xs ys =
  match xs, ys with
  | [], ys -> ys
  | xs, [] -> xs
  | x :: xs', y :: ys' ->
    if x <= y then
      x :: merge (<=) xs' ys
    else
      y :: merge (<=) xs ys'
\end{camlcode}
\caption{Non-tail-recursive merge.}
\label{fig:nontailrec-merge}
\end{subfigure}%
\begin{subfigure}[t]{.5\textwidth}
\begin{camlcode}
let rec revmerge (<=) xs ys accu =
  match xs, ys with
  | [], ys -> rev_append ys accu
  | xs, [] -> rev_append xs accu
  | x :: xs', y :: ys' ->
    if x <= y then
      revmerge (<=) xs' ys (x :: accu)
    else
      revmerge (<=) xs ys' (y :: accu)
\end{camlcode}
\caption{Tail-recursive merge.}
\label{fig:tailrec-merge}
\end{subfigure}
\caption{Non-tail-recursive and tail-recursive merge functions.}
\label{fig:merge}

\begin{subfigure}[t]{.5\textwidth}
\begin{camlcode}[escapeinside=\#\#]
let rec sort (<=) = function
  | [] -> []
  | [x] -> [x]
  | xs ->
    let k = length xs / 2 in
    let (xs1, xs2) = split_n xs k in
    merge (<=)
      (sort (<=) xs1)
      (sort (<=) xs2)

#\phantom{\underline{()}}#
\end{camlcode}
\caption{Top-down mergesort.}
\label{fig:top-down-mergesort}
\end{subfigure}%
\begin{subfigure}[t]{.5\textwidth}
\begin{camlcode}
let sort (<=) xs =
  let rec merge_pairs = function
    | a :: b :: xs ->
      merge (<=) a b :: merge_pairs xs
    | xs -> xs in
  let rec merge_all = function
    | [] -> []
    | [x] -> x
    | xs -> merge_all (merge_pairs xs)
  in
  merge_all (map (fun x -> [x]) xs)
\end{camlcode}
\caption{Bottom-up mergesort.}
\label{fig:bottom-up-mergesort}
\end{subfigure}
\caption{Naive implementations of top-down and bottom-up mergesort algorithms.}
\label{fig:naive-mergesort-algorithms}
\end{figure}

\begin{figure}[t]
 \tikzset{
  level/.style = {sibling distance={7cm/pow(2,#1)}, level distance=.5cm},
  every node/.style = {fill=black!20, circle, inner sep=-2pt, minimum width=1em},
  merge/.style = {font={$\merge$}}}
 \begin{subfigure}{.5\textwidth}
  \centering
  \begin{tikzpicture}
   \node[merge] {}
     child{
       node[merge] {}
       child{
         node[merge] {} child{ node {7} } child{ node {6} }
       }
       child{
         node[merge] {}
         child{ node {2} }
         child{
           node[merge] {} child{ node {4} } child{ node {3} }
         }
       }
     }
     child{
       node[merge] {}
       child{
         node[merge] {} child{ node {8} } child{ node {1} }
       }
       child{
         node[merge] {}
         child{ node {5} }
         child{
           node[merge] {} child{ node {0} } child{ node {9} }
         }
       }
     };
  \end{tikzpicture}
 \caption{A trace of top-down mergesort.}
 \label{fig:mergesort-top-down-trace}
 \end{subfigure}%
 \begin{subfigure}{.5\textwidth}
  \centering
  \begin{tikzpicture}
   \node[merge] {}
     child{
       node[merge] {}
       child{
         node[merge] {}
         child{
           node[merge] {} child{ node {7} } child{ node {6} }
         }
         child{
           node[merge] {} child{ node {2} } child{ node {4} }
         }
       }
       child{
         node[merge] {}
         child{
           node[merge] {} child{ node {3} } child{ node {8} }
         }
         child{
           node[merge] {} child{ node {1} } child{ node {5} }
         }
       }
     }
     child{
       node[merge] {} child{ node {0} } child{ node {9} }
     };
  \end{tikzpicture}
 \caption{A trace of bottom-up mergesort.}
 \label{fig:mergesort-bottom-up-trace}
 \end{subfigure}
 \caption{The traces of the sorting processes for the input list \caml{[7; 6; 2; 4; 3; 8; 1;
 5; 0; 9]} in the naive top-down \subref{fig:mergesort-top-down-trace} and bottom-up
 \subref{fig:mergesort-bottom-up-trace} mergesort algorithms. Each number placed as a leaf denotes
 a singleton list consisting of it, and $\merge$ placed as an internal node denotes the merging of
 its children. The difference of their shapes reflects the difference of associativity of merge.}
 \label{fig:mergesort-traces}
\end{figure}

In the top-down mergesort, two lists being merged have the same length, or the latter has one extra
element. In the bottom-up mergesort, the length of the former list is a power of 2 and is not less
than the length of the latter list.
Therefore, when the length of the input is not a power of two, these two mergesort functions split
the input differently, which can be illustrated using \emph{traces} as in
\cref{fig:mergesort-traces}. A trace is a binary tree that reflects the divide-and-conquer structure
of computation of mergesort, whose leaves and internal nodes denote a singleton list and merge,
respectively.
A trace of the top-down mergesort is always an AVL tree; that is, the heights of the two children of
any node differ by at most one. In contrast, the left child of any node is a perfect binary tree in
a trace of the bottom-up mergesort.
Regardless of the approach, the height of any trace must be logarithmic in the length $n$ of the
input to achieve the optimal time complexity, \eg, the height is always $\lceil \log_2 n \rceil$ in
our implementations in \cref{fig:naive-mergesort-algorithms}.
Since the merge function costs a linear time in the sum of lengths of input, merge operations of
each level in the trace cost a linear time $\bigO(n)$ in total. Thus, we can conclude that mergesort has a
quasilinear time complexity $\bigO(n \log n)$.
More rigorous complexity analyses of mergesort that use a recurrence relation can be found in some
textbooks on algorithms, \eg, \cite[Section 2.3.2]{CLRS4th}.

\subsection{Characterization of stable non-tail-recursive mergesort functions}
\label{sec:characterization}

In this section, we present the characterization of stable non-tail-recursive mergesort functions.
A sort algorithm is \emph{stable} if equivalent elements, such as $x$ and $y$ satisfying
$x \equiv y \coloneq x \leq y \land y \leq x$, always appear in the same order in the input and
output.
To maintain the stability of a mergesort function, one first has to notice that the ordering of the
last two arguments of \caml{merge} (\caml{xs} and \caml{ys} in \cref{fig:nontailrec-merge}) matters.
When two elements being compared in \caml{merge} are equivalent, the element from the first list
\caml{xs} appears earlier than the other in the output.
Thus, these two lists being merged should be obtained by sorting two adjacent (former and latter)
slices of the input, respectively.
Taking this observation into account, the intuition behind the characterization is that the
mergesort function can be turned into the identity function by replacing \caml{merge} with
concatenation. In other words, we should always be able to get the input by flattening traces, \eg,
\cref{fig:mergesort-traces}.

In order to ensure this replacement is done in the intended way and to state the characterization
formally, we first abstract out the mergesort function \sort{} of type:
\begin{alignat*}{3}
 \Tsort \coloneq{}
 & \forall (T : \types{}),
 && \underbrace{(T \to T \to \mathrm{bool})}_{\text{comparison function}} \to{}
 && \underbrace{\mathrm{list} \, T\vphantom{()}}_{\text{input}} \to
    \underbrace{\mathrm{list} \, T\vphantom{()}}_{\text{output}} \\
 \intertext{to \asort{} of type:}
 \Tasort \coloneq{} & \forall (T \, R : \types{}),
 && \underbrace{(R \to R \to R)}_{\text{merge}\vphantom{b}} \to
    \underbrace{(T \to R)}_{\text{singleton}} \to
    \underbrace{R\vphantom{()}}_{\text{empty}} \to{}
 && \underbrace{\mathrm{list} \, T\vphantom{()}}_{\text{input}} \to
    \underbrace{R\vphantom{()}}_{\text{output}}.
\end{alignat*}
The abstract sort function \asort{} is abstracted over the type $R$ of sorted lists and the
basic operations on them, namely, the merge, singleton, and empty constructions.
Therefore, we expect to get \sort{} by instantiating \asort{} with $\merge_\leq$
($\coloneq \texttt{merge} \, (\leq)$):
\begin{equation}
\forall (T : \types{}) \, ({\leq} : T \to T \to \mathrm{bool}) \, (\mathit{xs} : \mathrm{list} \, T),
 \asort{} \, (\merge_\leq) \, [\cdot] \, [] \, \mathit{xs} =
 \sort{}_\leq \, \mathit{xs} \label{eq:asort_mergeE}
\end{equation}
where $[\cdot]$ and $[]$ are the singleton list function $(\lambda (x : T), [x])$ and the empty
list, respectively.
The replacement of merge ($\merge_\leq$) with concatenation ($\concat$) can be done by instantiating
\asort{} with concatenation. We expect to get the identity function by this instantiation:
\begin{equation}
\forall (T : \types{}) \, (\mathit{xs} : \mathrm{list} \, T),
 \asort{} \, (\concat) \, [\cdot] \, [] \, \mathit{xs} = \mathit{xs}. \label{eq:asort_catE}
\end{equation}
\stepcounter{equation}

Although both \cref{eq:asort_mergeE,eq:asort_catE} instantiate the type parameter $R$ of the
abstract sort function \asort{} with $\mathrm{list} \, T$, abstracting it out as a type
parameter is crucial in ensuring that \asort{} builds the output by using the three
operators on $R$ inductively and uniformly, \ie, in such a way that
\asort{} is parametric as defined in \cref{sec:prelim-param}:
$(\asort{}, \asort{}) \in \paramt{\Tasort}$, which holds iff:
\begin{flalign*}
 \forall & (T_1 \, T_2 : \types{}) \, ({\sim_T} \subseteq T_1 \times T_2), & \\
 \forall & (R_1 \, R_2 : \types{}) \, ({\sim_R} \subseteq R_1 \times R_2), & \\
 \forall & (\mathit{merge}_1 : R_1 \to R_1 \to R_1) \, (\mathit{merge}_2 : R_2 \to R_2 \to R_2), & \text{(merge)} \\
         & (\forall (\mathit{xs}_1 : R_1) \, (\mathit{xs}_2 : R_2), \mathit{xs}_1 \sim_R \mathit{xs}_2 \Rightarrow \\
         & \phantom{(}
            \forall (\mathit{ys}_1 : R_1) \, (\mathit{ys}_2 : R_2), \mathit{ys}_1 \sim_R \mathit{ys}_2 \Rightarrow
	      (\mathit{merge}_1 \, \mathit{xs}_1 \, \mathit{ys}_1) \sim_R (\mathit{merge}_2 \, \mathit{xs}_2 \, \mathit{ys}_2)) \Rightarrow \\
 \forall & ([\cdot]_1 : T_1 \to R_1) \, ([\cdot]_2 : T_2 \to R_2), & \text{(singleton)} \\
         & (\forall (x_1 : T_1) \, (x_2 : T_2), x_1 \sim_T x_2 \Rightarrow [x_1]_1 \sim_R [x_2]_2) \Rightarrow \\
 \forall & ([]_1 : R_1) \, ([]_2 : R_2), []_1 \sim_R []_2 \Rightarrow & \text{(empty)} \\
 \forall & (\mathit{xs}_1 : \mathrm{list} \, T_1) \, (\mathit{xs}_2 : \mathrm{list} \, T_2),
           (\mathit{xs}_1, \mathit{xs}_2) \in \paramt{\mathrm{list}}_{\sim_T} \Rightarrow \\
         & (\asort{} \, \mathit{merge}_1 \, [\cdot]_1 \, []_1 \, \mathit{xs}_1) \sim_R
           (\asort{} \, \mathit{merge}_2 \, [\cdot]_2 \, []_2 \, \mathit{xs}_2) &
\end{flalign*}

To sum up, we define the characteristic property of stable mergesort functions as the existence of
an abstract mergesort function \asort{} that is parametric and satisfies
\cref{eq:asort_mergeE,eq:asort_catE}.

\subsection{Proofs of the characteristic property$^\dag$}
\label{sec:implementation-to-characterization}

In this section, we briefly show that the two mergesort functions presented in
\cref{fig:naive-mergesort-algorithms} satisfy the characteristic property presented in
\cref{sec:characterization}.
Before doing so, we present an insertion sort (\cref{def:insertion-sort}) as the simplest instance
of the characterization (\cref{lemma:insertion-sort-stable}).

\begin{definition}[Insertion sort]
\label{def:insertion-sort}
Noticing that one-element insertion to a sorted list is a special case of merge, an insertion sort
can be defined as follows.
\begin{camlcode}
let rec sort (<=) = function [] -> [] | x :: xs -> merge (<=) [x] (sort (<=) xs)
\end{camlcode}
\end{definition}

\begin{lemma}
\label{lemma:insertion-sort-stable}
The insertion sort (\cref{{def:insertion-sort}}) satisfies the characteristic property.
\end{lemma}

\begin{proof}
We first abstract the three basic operators on sorted list out of \cref{{def:insertion-sort}}.
\begin{camlcode}
let rec asort merge singleton empty = function
  [] -> empty | x :: xs -> merge (singleton x) (asort merge singleton empty xs)
\end{camlcode}
\Cref{eq:asort_mergeE} holds by definition. The parametricity of \caml{asort} follows from the
abstraction theorem (\cref{assumption:abstraction}).
Therefore, we are left to prove that \cref{eq:asort_catE}, \ie, $\phi \, \mathit{xs} = \mathit{xs}$
where $\phi \coloneq \caml{asort} \, (\concat) \, [\cdot] \, []$, by structural induction on
$\mathit{xs}$.
If $\mathit{xs} = []$, it is trivial. Otherwise $\mathit{xs} = x :: \mathit{xs'}$ and
\begin{align*}
  \phi \, (x :: \mathit{xs'})
  &= [x] \concat \phi \, \mathit{xs'} & (\text{Definition of \caml{asort}}) \\
  &= x :: \phi \, \mathit{xs'}        & (\text{\cref{def:coq:cat} ($\concat$)}) \\
  &= x :: \mathit{xs'}                & (\text{I.H.}) & \qedhere
\end{align*}
\end{proof}

To show that top-down and bottom-up mergesort functions (\cref{fig:naive-mergesort-algorithms})
satisfy the characteristic property, we abstract them out as in \cref{fig:naive-abstract-mergesort}.
Again, we are left to prove that $\caml{asort} \, (\concat) \, [\cdot] \, []$ is the identity
function, for each variation.
The \Coq counterpart of this section is \cref{sec:populating-the-interface}.

\begin{figure}[t]
\begin{subfigure}[t]{.5\textwidth}
\begin{camlcode}[escapeinside=\#\#]
let asort merge singleton empty =
  let rec asort_rec = function
    | [] -> #\underline{empty}#
    | [x] -> #\underline{singleton}# x
    | xs ->
      let k = length xs / 2 in
      let (xs1, xs2) = split_n xs k in
      #\underline{merge}#
        (asort_rec xs1)
        (asort_rec xs2)
  in asort_rec
\end{camlcode}
\caption{Top-down abstract mergesort.}
\label{fig:top-down-abstract-mergesort}
\end{subfigure}%
\begin{subfigure}[t]{.5\textwidth}
\begin{camlcode}[escapeinside=\#\#]
let asort merge singleton empty xs =
  let rec merge_pairs = function
    | a :: b :: xs ->
      #\underline{merge}# a b :: merge_pairs xs
    | xs -> xs in
  let rec merge_all = function
    | [] -> #\underline{empty}#
    | [x] -> x
    | xs -> merge_all (merge_pairs xs)
  in
  merge_all (map #\underline{singleton}# xs)
\end{camlcode}
\caption{Bottom-up abstract mergesort.}
\label{fig:bottom-up-abstract-mergesort}
\end{subfigure}
 \caption{Two abstract mergesort functions derived from the top-down and bottom-up mergesort
 functions (\cref{fig:naive-mergesort-algorithms}). The use of abstract operators on sorted lists
 \caml{merge}, \caml{singleton}, and \caml{empty} are underlined.}
\label{fig:naive-abstract-mergesort}
\end{figure}

\begin{lemma}[\Cref{eq:asort_catE} for the top-down mergesort in \cref{fig:top-down-mergesort}]
  \caml{asort} in \cref{fig:top-down-abstract-mergesort} satisfies
  $\caml{asort} \, (\concat) \, [\cdot] \, [] \, \mathit{xs} = \mathit{xs}$ for any $\mathit{xs}$.
\end{lemma}

\begin{proof}
 Let $\phi$ be $\caml{asort} \, (\concat) \, [\cdot] \, []$.
 The proof is by strong mathematical induction on the length $n$ of $\mathit{xs}$.
 If $n < 2$, $\mathit{xs}$ is either an empty or singleton list, and the equation holds by
 definition.
 Otherwise, $\mathit{xs}$ is split into the list of the first
 $k \coloneq \left\lfloor \frac{n}{2} \right\rfloor$ elements $\mathit{xs}_1$ and the rest
 $\mathit{xs}_2$ (of length $n - k$). Since $0 < k < n$ and thus $n - k < n$,
 $\phi \, \mathit{xs}_1 = \mathit{xs}_1$ and $\phi \, \mathit{xs}_2 = \mathit{xs}_2$ follow
 from the induction hypothesis.
 Therefore, $\phi \, \mathit{xs}
 = \phi \, \mathit{xs}_1 \concat \phi \, \mathit{xs}_2
 = \mathit{xs}_1 \concat \mathit{xs}_2$,
 which is equal to $\mathit{xs}$ (\cref{lemma:coq:cat_take_drop}).
\end{proof}

\begin{lemma}[\Cref{eq:asort_catE} for the bottom-up mergesort in \cref{fig:bottom-up-mergesort}]
  \caml{asort} in \cref{fig:bottom-up-abstract-mergesort} satisfies
$\caml{asort} \, (\concat) \, [\cdot] \, [] \, \mathit{xs} = \mathit{xs}$ for any $\mathit{xs}$.
\end{lemma}

\begin{proof}
\par First, we prove
\begin{equation}
  \caml{flatten} \, (\caml{merge_pairs} \, \mathit{xs}) = \caml{flatten} \, \mathit{xs}  \label{eq:flatten_merge_pairs}
\end{equation}
for any list of lists $\mathit{xs}$ by structural induction on $\mathit{xs}$.
 If $\mathit{xs}$ has length $n < 2$ the equation holds by definition. Otherwise $\mathit{xs} = a :: b :: \mathit{xs'}$ and
 \begin{align*}
     & \caml{flatten} \, (\caml{merge_pairs} \, \mathit{xs}) \\
  ={}& \caml{flatten} \, ((a \concat b) :: \caml{merge_pairs} \, \mathit{xs'})
     & (\text{Definition \cref{fig:bottom-up-abstract-mergesort}}) \\
  ={}& (a \concat b) \concat \caml{flatten} \, (\caml{merge_pairs} \, \mathit{xs'})
     & (\text{\cref{def:coq:flatten} (\caml{flatten})}) \\
  ={}& (a \concat b) \concat \caml{flatten} \, \mathit{xs'}
     & (\text{I.H.}) \\
  ={}& a \concat (b \concat \caml{flatten} \, \mathit{xs'})
     & (\text{\cref{lemma:coq:catA} (associativity of $\concat$)}) \\
  ={}& \caml{flatten} \, \mathit{xs}.
     & (\text{\cref{def:coq:flatten} (\caml{flatten})})
 \end{align*}

Second, we prove
\begin{equation}
\caml{merge_all} \, \mathit{xs} = \caml{flatten} \, \mathit{xs}\label{eq:flatten_all}
\end{equation}
for any list of
lists $\mathit{xs}$ by strong mathematical induction on the length $n$ of $\mathit{xs}$.
 If $n < 2$, $\mathit{xs}$ is either an empty or singleton list, and the equation holds by
 definition.
 Otherwise,
 \begin{align*}
  \caml{merge_all} \, \mathit{xs}
  &= \caml{merge_all} \, (\caml{merge_pairs} \, \mathit{xs}) & (\text{Definition \cref{fig:bottom-up-abstract-mergesort}}) \\
  \intertext{Noticing that the length of $\caml{merge_pairs} \, \mathit{xs}$ is
  $\left\lceil\frac{n}{2}\right\rceil < n$, we can apply the induction hypothesis:}
  &= \caml{flatten} \, (\caml{merge_pairs} \, \mathit{xs}) & \text{(I.H.)}\\
  &= \caml{flatten} \, \mathit{xs}. & (\text{\cref{eq:flatten_merge_pairs}})
 \end{align*}

Finally, we show
 \begin{align*}
  \caml{asort} \, (\concat) \, [\cdot] \, [] \, \mathit{xs}
  &= \caml{merge_all} \, [[x] \mid x \leftarrow \mathit{xs}] & (\text{Definition \cref{fig:bottom-up-abstract-mergesort}}) \\
  &= \caml{flatten} \, [[x] \mid x \leftarrow \mathit{xs}] & (\text{\cref{eq:flatten_all}})\\
  &= \mathit{xs}. & (\text{Structural induction on $\mathit{xs}$}) & \qedhere
 \end{align*}
\end{proof}

\subsection{Correctness proofs}
\label{sec:characterization-to-correctness}

In this section, we deduce several correctness results of mergesort solely from the characteristic
property (\cref{sec:characterization}).
We universally quantify the mergesort function \sort{} in this section, and thus, all the
correctness results below apply to any mergesort function satisfying the characteristic property,
including the top-down and bottom-up mergesorts in \cref{fig:naive-mergesort-algorithms}.
The use of the parametricity of the abstract sorting function \asort{} in our
correctness proofs is twofold: deducing an induction principle over traces
(\cref{sec:induction-over-traces}), and deducing the naturality of \sort{}
(\cref{sec:naturality}).

\subsubsection{Induction over traces}
\label{sec:induction-over-traces}

\begin{lemma}[An induction principle over traces of \sort{}]
 \label{lemma:sort_ind}
 Suppose $\leq$ and $\sim$ are binary relations on $T$ and $\mathrm{list} \, T$, respectively, and
 $\mathit{xs}$ is a list of type $\mathrm{list} \, T$.
 Then, $\mathit{xs} \sim \sort{}_\leq \, \mathit{xs}$ holds whenever the following three
 induction cases hold:
 \begin{itemize}
  \item for any lists $\mathit{xs}$, $\mathit{xs'}$, $\mathit{ys}$, and $\mathit{ys'}$ of type
	$\mathrm{list} \, T$,
	$(\mathit{xs} \concat \mathit{ys}) \sim (\mathit{xs'} \merge_\leq \mathit{ys'})$ holds
	whenever $\mathit{xs} \sim \mathit{xs'}$ and $\mathit{ys} \sim \mathit{ys'}$ hold,
  \item for any $x$ of type $T$, $[x] \sim [x]$ holds, and
  \item $[] \sim []$ holds.
 \end{itemize}
\end{lemma}
\begin{proof}
 Thanks to \cref{eq:asort_mergeE,eq:asort_catE},
 $\mathit{xs} \sim \sort{}_\leq \, \mathit{xs}$ holds if and only if
 \[
  \asort{} \, (\concat) \, (\lambda (x : T), [x]) \, [] \, \mathit{xs} \sim
  \asort{} \, (\merge_\leq) \, (\lambda (x : T), [x]) \, [] \, \mathit{xs}.
 \]
 We apply the parametricity of \asort{} by instantiating $\sim_T$ with the equality over $T$ and
 $\sim_R$ with $\sim$.
 The premise $(\mathit{xs}_1, \mathit{xs}_2) \in \paramt{\mathrm{list}}_{\sim_T}$ holds
 because both $\mathit{xs}_1$ and $\mathit{xs}_2$ are instantiated with $\mathit{xs}$ and
 $\paramt{\mathrm{list}}_{\sim_T}$ is just the equality over $\mathrm{list} \, T$.
 The other three premises exactly correspond to the three induction cases.
\end{proof}

\begin{remark}
 \label{remark:induction}
 \Cref{lemma:sort_ind} is a variant of the general fact that parametricity implies induction
 principles on Church-encoded datatypes~\cite{DBLP:journals/pacmpl/AltenkirchCKS24,
 DBLP:journals/pacmpl/KaposiKA19, DBLP:conf/itp/Tassi19, Wadler:rectypes}, except that we use the
 binary version of parametricity to provide a simultaneous induction principle relating the input and output of mergesort.
 If we instantiate the abstract mergesort with a concrete type of elements $T$ and an input of type
 $\mathrm{list} \, T$, we get a term of type
 \[ \forall (R : \types{}), (R \to R \to R) \to (T \to R) \to R \to R, \]
 which represents Church-encoded binary trees with singleton and empty constructs as their
 leaves.
\end{remark}

As applications of \cref{lemma:sort_ind}, we prove the permutation property of \sort{}
(\cref{lemma:perm_sort,corollary:mem_sort}) and sortedness and stability results of \sort{}
(\cref{lemma:pairwise_sort,lemma:sort_pairwise_stable,corollary:sort_stable}).

\begin{lemma}
 \label{lemma:perm_sort}
 For any relation $\leq$ on type $T$ and $\mathit{xs}$ of type $\mathrm{list} \, T$,
 $\sort{}_\leq \, \mathit{xs} \permeq \mathit{xs}$ holds; that is,
 $\sort{}_\leq \, \mathit{xs}$ is a permutation of $\mathit{xs}$ (\cf
 \cref{def:coq:perm_eq} for a precise definition of $\permeq$).
\end{lemma}

\begin{proof}
 We prove it by induction on $\sort{}_\leq \, \mathit{xs}$ (\cref{lemma:sort_ind}).
 Since the last two cases are obvious, suffice it to show that
 $\mathit{xs'} \merge_\leq \mathit{ys'} \permeq \mathit{xs} \concat \mathit{ys}$
 whenever $\mathit{xs'} \permeq \mathit{xs}$ and $\mathit{ys'} \permeq \mathit{ys}$.
 \begin{align*}
   \mathit{xs'} \merge_\leq \mathit{ys'}
   &\permeq \mathit{xs'} \concat \mathit{ys'}
   & (\text{\cref{lemma:coq:perm_merge}})
   \intertext{Since $\concat$ is congruent with respect to $\permeq$ (\text{\cref{lemma:coq:perm_cat}}),}
   &\permeq \mathit{xs} \concat \mathit{ys}. && \qedhere
 \end{align*}
\end{proof}

\begin{corollary}
 \label{corollary:mem_sort}
 For any relation $\leq$ on type $T$ and $\mathit{xs}$ of type $\mathrm{list} \, T$,
 $\sort{}_\leq \, \mathit{xs}$ has the same set of elements as $\mathit{xs}$, \ie,
 $x \in \sort{}_\leq \, \mathit{xs}$ iff $x \in \mathit{xs}$ for any $x \in T$.
\end{corollary}

We define two versions of sortedness of lists to state
\cref{lemma:pairwise_sort,lemma:sort_pairwise_stable} in their general form.

\begin{definition}[Sortedness]
 \label{def:sortedness}
 Suppose $\mathrel{R}$ is a relation on type $T$.
 A list $\mathit{xs} \coloneq [x_0, \dots, x_n]$ of type $\mathrm{list} \, T$ is said to be:
 \begin{itemize}
  \item \emph{sorted} \wrt $\mathrel{R}$ if the relation $\mathrel{R}$ holds for each adjacent pair,
	\ie, $x_0 \mathrel{R} x_1 \land \dots \land x_{n - 1} \mathrel{R} x_n$, and
  \item \emph{pairwise sorted} \wrt $\mathrel{R}$ if the relation $\mathrel{R}$ holds for any $x_i$
	and $x_j$ such that $i < j \leq n$, \ie,
	\[
	 x_0 \mathrel{R} x_1 \land \dots \land x_0 \mathrel{R} x_n \land
	 x_1 \mathrel{R} x_2 \land \dots \land x_1 \mathrel{R} x_n \land \dots \land
	 x_{n - 1} \mathrel{R} x_n.
	\]
 \end{itemize}
\end{definition}

\begin{lemma}
\label{lemma:pairwise_sort}
For any $s$ of type $\mathrm{list} \, T$ pairwise sorted \wrt ${\leq} \subseteq T \times T$,
$\sort{}_\leq \, s = s$ holds.
\end{lemma}

\begin{proof}
 Since we use this lemma only for the proof of \cref{corollary:sort_standard_stable}, we omit the
 proof, which is done by induction on $\sort{}_\leq \, s$ (\cref{lemma:sort_ind}).
 See \cref{lemma:coq:pairwise_sort} for the complete proof.
\end{proof}

\begin{theorem}[Sortedness and stability of \sort{}]
 \label{lemma:sort_pairwise_stable}
 Suppose $\leq_1$ and $\leq_2$ are binary relations on type $T$, $\leq_1$ is total, and
 $\mathit{xs}$ is a list of type $\mathrm{list} \, T$.
 Then, $\sort{}_{\leq_1} \, \mathit{xs}$ is sorted \wrt the following lexicographic order:
 \[
  x \leq_{\mathrm{lex}} y \coloneq x \leq_1 y \land (y \not\leq_1 x \lor x \leq_2 y)
 \]
 whenever $\mathit{xs}$ is pairwise sorted \wrt $\leq_2$.
\end{theorem}

\begin{proof}
 We prove a generalized proposition:
 \[
 (\sort{}_{\leq_1} \, \mathit{xs} \subseteq \mathit{xs}) \land
 (\text{$\mathit{xs}$ is pairwise sorted \wrt $\leq_2$} \Rightarrow
  \text{$\sort{}_{\leq_1} \, \mathit{xs}$ is sorted \wrt $\leq_{\mathrm{lex}}$})
 \]
 by induction on $\sort{}_{\leq_1} \, \mathit{xs}$ (\cref{lemma:sort_ind}).
 Since $\subseteq$ is reflexive and a list whose length is less than 2 is always sorted, the last
 two cases are obvious.%
 \pagebreak 

 For the first component of the conjunction in the first induction case, suffice it to show that
 \[
  (\mathit{xs'} \subseteq \mathit{xs}) \land (\mathit{ys'} \subseteq \mathit{ys}) \Rightarrow
  (\mathit{xs'} \merge_{\leq_1} \mathit{ys'} \subseteq \mathit{xs} \concat \mathit{ys})
 \]
 which is obvious since $\mathit{xs'} \merge_{\leq_1} \mathit{ys'}$ is a permutation of
 $\mathit{xs'} \concat \mathit{ys'}$ (\cref{lemma:coq:perm_merge}).

 For the second component of the conjunction, suffice it to show that
 $\mathit{xs'} \merge_{\leq_1} \mathit{ys'}$ is sorted \wrt $\leq_{\mathrm{lex}}$
 whenever:
 \begin{enumerate}[label=(\roman*)]
  \item \label{item:xs_ys_subset} $\mathit{xs'}$ (resp.~$\mathit{ys'}$) is a subset of $\mathit{xs}$
    (resp.~$\mathit{ys}$),
  \item \label{item:xs_ys_sorted} $\mathit{xs'}$ (resp.~$\mathit{ys'}$) is sorted \wrt
	$\leq_{\mathrm{lex}}$ if $\mathit{xs}$ (resp.~$\mathit{ys}$) is pairwise sorted \wrt $\leq_2$, and
  \item \label{item:concat_pairwise} $\mathit{xs} \concat \mathit{ys}$ is pairwise sorted \wrt $\leq_2$.
 \end{enumerate}
 Among these, \ref{item:concat_pairwise} is equivalent to the following conjunction
 (\cref{lemma:coq:pairwise_cat}):
 \begin{enumerate}[label=(\roman*), resume]
  \item \label{item:xs_ys_pairwise} both $\mathit{xs}$ and $\mathit{ys}$ are pairwise sorted \wrt
	$\leq_2$, and
  \item \label{item:xs_le2_ys} $x \leq_2 y$ holds for any $x \in \mathit{xs}$ and $y \in \mathit{ys}$.
 \end{enumerate}
 Hypotheses \ref{item:xs_ys_sorted} and \ref{item:xs_ys_pairwise} imply that both $\mathit{xs'}$
 and $\mathit{ys'}$ are sorted \wrt $\leq_{\mathrm{lex}}$.
 Hypotheses \ref{item:xs_ys_subset} and \ref{item:xs_le2_ys} imply that $x \leq_2 y$ holds for
 any $x \in \mathit{xs'}$ and $y \in \mathit{ys'}$ (\cref{lemma:coq:sub_all}).
 These two facts suffice to show that $\mathit{xs'} \merge_{\leq_1} \mathit{ys'}$ is sorted \wrt
 $\leq_{\mathrm{lex}}$ (\cref{lemma:coq:merge_stable_sorted}).
\end{proof}

The following corollary also holds since the sortedness and the pairwise sortedness are equivalent
for any transitive relation (\cref{lemma:coq:sorted_pairwise}).

\begin{corollary}
 \label{corollary:sort_stable}
 Suppose $\leq_1$ and $\leq_2$ are binary relations on type $T$, $\leq_1$ is total, $\leq_2$ is
 transitive, and $\mathit{xs}$ is a list of type $\mathrm{list} \, T$.
 Then, $\sort{}_{\leq_1} \, \mathit{xs}$ is sorted \wrt the lexicographic order
 $\leq_{\mathrm{lex}}$ of $\leq_1$ and $\leq_2$ whenever $\mathit{xs}$ is sorted \wrt $\leq_2$.
\end{corollary}

\subsubsection{Naturality}
\label{sec:naturality}

The induction principle over traces (\cref{lemma:sort_ind}) does not allow us to relate two sorting
processes that behave parametrically.
Instead, we obtain the parametricity (\cref{lemma:param_sort}) and the naturality (\cref{lemma:sort_map}) of
\sort{} from the parametricity of \asort{}.

\begin{lemma}[The parametricity of \sort{}]
 \label{lemma:param_sort}
 Any \sort{} function satisfying the characteristic property is parametric
 (\cref{sec:prelim-param}), \ie, $(\sort, \sort) \in \paramt{\Tsort}$, which holds iff:
 \begin{align*}
  \forall & (T_1 \, T_2 : \types{}) \, ({\sim_T} \subseteq T_1 \times T_2), \\
  \forall & ({\leq_1} : T_1 \to T_1 \to \mathrm{bool}) \, ({\leq_2} : T_2 \to T_2 \to \mathrm{bool}), \\
  & (\forall (x_1 : T_1) \, (x_2 : T_2), x_1 \sim_T x_2 \Rightarrow
     \forall (y_1 : T_1) \, (y_2 : T_2), y_1 \sim_T y_2 \Rightarrow
     (x_1 \leq_1 y_1) = (x_2 \leq_2 y_2)) \Rightarrow \\
  \forall & (\mathit{xs}_1 : \mathrm{list} \, T_1) \, (\mathit{xs}_2 : \mathrm{list} \, T_2),
            (\mathit{xs}_1, \mathit{xs}_2) \in \paramt{\mathrm{list}}_{\sim_T} \Rightarrow
    (\sort{}_{\leq_1} \, \mathit{xs}_1, \sort{}_{\leq_2} \, \mathit{xs}_2) \in
    \paramt{\mathrm{list}}_{\sim_T}.
 \end{align*}
\end{lemma}

\begin{proof}
 $\sort{}_{\leq_i}$ is extensionally equal to
 $\asort{} \, (\merge_{\leq_i}) \, (\lambda (x : T), [x]) \, []$
 for each $i \in \{1, 2\}$ thanks to the characterization.
 Since \asort{} and its arguments are parametric, \sort{} is parametric as well.
\end{proof}

\begin{lemma}[The naturality of \sort{}]
 \label{lemma:sort_map}
 Suppose $\leq_T$ is a relation on type $T$, $f$ is a function from $T'$ to $T$, and $\mathit{xs}$ is
 a list of type $\mathrm{list} \, T'$. Then, the following equation holds:
 \[
 \sort{}_{\leq_T} \, [f \, x \mid x \leftarrow \mathit{xs}] =
 [f \, x \mid x \leftarrow \sort{}_{\leq_{T'}} \, \mathit{xs}]
 \]
 where $x \leq_{T'} y \coloneq f \, x \leq_T f \, y$.
\end{lemma}

\begin{proof}
 We instantiate \cref{lemma:param_sort} with
 $T_1 \coloneq T$,
 $T_2 \coloneq T'$,
 $x \sim_T y \coloneq (x = f \, y)$,
 ${\leq_1} \coloneq {\leq_T}$, and
 ${\leq_2} \coloneq {\leq_{T'}}$.
 Then, the first premise is equivalent to the definition of $\leq_{T'}$.
 The second premise $(\mathit{xs}_1, \mathit{xs}_2) \in \paramt{\mathrm{list}}_{\sim_T}$
 is equivalent to $\mathit{xs}_1 = [f \, x \mid x \leftarrow \mathit{xs}_2]$.
 By substituting this equation to the conclusion, we get
 $\sort{}_{\leq_T} \, [f \, x \mid x \leftarrow \mathit{xs}_2] =
 [f \, x \mid x \leftarrow \sort{}_{\leq_{T'}} \, \mathit{xs}_2]$.
\end{proof}

\begin{remark}
 \label{remark:naturality}
 \Cref{lemma:sort_map} is known as a free theorem~\cite[Section 3.3]{DBLP:conf/fpca/Wadler89} for
 type $\forall (T : \types{}), (T \to T \to \mathrm{bool}) \to \mathrm{list} \, T \to
 \mathrm{list} \, T$.
 It is also a case where parametricity implies naturality~\cite{Reddy:1997}.
\end{remark}

\begin{theorem}
 \label{lemma:filter_sort}
 For any total preorder $\leq$ on $T$ and predicate $p$ on $T$, $\texttt{filter}_p$ commutes with
 $\sort{}_\leq$ under function composition; that is, the following equation holds for any
 $\mathit{xs}$ of type $\mathrm{list} \, T$:
 \[
  \texttt{filter}_p \, (\sort{}_\leq \, \mathit{xs}) =
  \sort{}_\leq \, (\texttt{filter}_p \, \mathit{xs}).
 \]
\end{theorem}

\begin{proof}
 We will rely on the fact that two lists are equal whenever they are sorted \wrt a transitive and
 irreflexive relation and contain the same set of elements (\cref{lemma:coq:irr_sorted_eq}).

 Since $\leq$ is total and hence reflexive, we instead consider a relation on natural numbers
 (indices):
 \[
  i \leq_I j \coloneq
  \texttt{nth} \, x_0 \, \mathit{xs} \, i \leq \texttt{nth} \, x_0 \, \mathit{xs} \, j
 \]
 where $\texttt{nth} \, x_0 \, \mathit{xs} \, i$ is the $i^\mathrm{th}$ element in the list
 $\mathit{xs}$ (\cref{def:coq:nth}) and its default value $x_0$ is an arbitrary element from the
 list $\mathit{xs}$.
 We can turn this relation into an irreflexive relation by composing it lexicographically with the
 strict order on natural numbers $<_{\mathbb{N}}$, resulting in the relation
 \[
  i <_I j \coloneq i \leq_I j \land (j \not\leq_I i \lor i <_{\mathbb{N}} j).
 \]

 Now we replace $\mathit{xs}$ everywhere with
 $\texttt{map}_{\texttt{nth} \, x_0 \, \mathit{xs}} \, \mathit{is}$
 where $\mathit{is} \coloneq [0, \dots, \lvert\mathit{xs}\rvert - 1]$ (\cref{lemma:coq:mkseq_nth}).
 It thus remains to prove
 \[
  \texttt{filter}_p \,
    (\sort{}_\leq \, (\texttt{map}_{\texttt{nth} \, x_0 \, \mathit{xs}} \, \mathit{is}))
  = \sort{}_\leq \,
      (\texttt{filter}_p \, (\texttt{map}_{\texttt{nth} \, x_0 \, \mathit{xs}} \, \mathit{is})).
 \]
 Using the naturality of \sort{} (\cref{lemma:sort_map}) and \texttt{filter}
 (\cref{lemma:coq:filter_map}), \ie,
 $\texttt{map}_f \, (\texttt{filter}_{p \circ f} \, \mathit{xs})
  = \texttt{filter}_p \, (\texttt{map}_f \, \mathit{xs})$,
 it remains to prove
 \[
  \texttt{map}_{\texttt{nth} \, x_0 \, \mathit{xs}} \,
    (\texttt{filter}_{p_I} \, (\sort{}_{\leq_I} \, \mathit{is}))
 = \texttt{map}_{\texttt{nth} \, x_0 \, \mathit{xs}} \,
     (\sort{}_{\leq_I} \, (\texttt{filter}_{p_I} \, \mathit{is}))
 \]
 where $p_I \coloneq p \circ \texttt{nth} \, x_0 \, \mathit{xs}$.

 Now, we apply the congruence rule with respect to \texttt{map} and the fact mentioned in the
 beginning of this proof (\cref{lemma:coq:irr_sorted_eq}) by checking that both sides of the above
 equation are sorted \wrt $<_I$ and they contain the same set of elements.
 The latter condition follows from \cref{corollary:mem_sort} and the fact that
 $x \in \texttt{filter}_p \, \mathit{xs}$ iff $p \, x \land x \in \mathit{xs}$ for any $p$,
 $\mathit{xs}$, and $x$ (\cref{lemma:coq:mem_filter}).
 Finally, both lists are sorted \wrt $<_I$ because \sort{} is stable
 (\cref{corollary:sort_stable}), $\mathit{is} \coloneq [0, \dots, \lvert\mathit{xs}\rvert - 1]$ is
 sorted \wrt $<_{\mathbb{N}}$ (\cref{lemma:coq:iota_ltn_sorted}), and \coq{filter} turns a sorted list
 into a sorted list whenever the relation is transitive (\cref{lemma:coq:sorted_filter}).
\end{proof}


\subsubsection{A remark on the formulation of stability}
\label{sec:sort_standard_stable}

The literature~\cite{CLRS4th, Leroy:mergesort, DBLP:journals/jar/Sternagel13,
DBLP:journals/tocl/LeinoL15} often formulates the stability of a sort function as follows.

\begin{corollary}[The standard formulation of the stability]
 \label{corollary:sort_standard_stable}
 For any total preorder $\leq$ on $T$, the equivalent elements always appear in the same order in
 the input and output of sorting; that is, the following equation holds for any $x$ of type $T$ and
 $s$ of type $\mathrm{list} \, T$:
 \[
  [y \leftarrow \sort{}_\leq \, s \mid x \equiv y] = [y \leftarrow s \mid x \equiv y].
 \]
\end{corollary}

\begin{proof}
 \begin{align*}
   [y \leftarrow \sort{}_\leq \, s \mid x \equiv y]
   &= \sort{}_\leq \, [y \leftarrow s \mid x \equiv y]
   & (\text{\cref{lemma:filter_sort}})
   \intertext{$[y \leftarrow s \mid x \equiv y]$, whose elements are all equivalent, is pairwise
     sorted \wrt $\leq$. Thus,}
   &= [y \leftarrow s \mid x \equiv y]
   & (\text{\cref{lemma:pairwise_sort}})
   & \qedhere
 \end{align*}
\end{proof}

\pagebreak 

We argue that our stability results (\cref{lemma:sort_pairwise_stable,lemma:filter_sort}) are
more general than \cref{corollary:sort_standard_stable} for the following reasons:
\begin{itemize}
\item While \cref{corollary:sort_standard_stable} restricts the predicate to $(x \equiv \cdot)$,
  \cref{lemma:filter_sort} applies to any predicate and opened up a natural way to prove other
  useful stability results
  (\cref{lemma:coq:sorted_mask_sort,lemma:coq:subseq_sort,lemma:coq:sorted_subseq_sort,{lemma:coq:mem2_sort}}).
\item In the above proofs, \cref{corollary:sort_standard_stable} is an easy consequence of
  \cref{lemma:filter_sort}, which follows from \cref{lemma:sort_pairwise_stable}.
\item While we proved the converse implications, \ie, \cref{corollary:sort_standard_stable} implies
  \cref{lemma:sort_pairwise_stable,lemma:filter_sort} under some assumptions on \sort, their proofs
  are non-trivial, and \cref{lemma:sort_pairwise_stable} derived from
  \cref{corollary:sort_standard_stable} requires $\leq_1$ to be transitive (see
  \cref{appx:stability-statements}).
\end{itemize}

\section{Optimizations}
\label{sec:optimizations}

In this section, we review some optimization techniques for mergesort, and see how our proof
technique presented in \cref{sec:nontailrec} extends to them (\cref{sec:new-proofs}).
In \cref{sec:tailrec-mergesort}, we review tail-recursive mergesort in call-by-value evaluation,
which does not use up stack space and thus is efficient.
In \cref{sec:smooth-mergesort}, we review smooth mergesort, that reuses sorted slices in the input
in the sorting process.
In \cref{sec:new-characterization}, we extend the characterization presented in
\cref{sec:characterization} to support the optimized mergesorts.
In \cref{sec:new-characterization-tailrec,sec:new-characterization-smooth}, we describe how the
tail-recursive and smooth mergesorts satisfy the extended characteristic property, respectively.
In \cref{sec:new-characterization-to-correctness}, we demonstrate that the extended characteristic
property implies the same correctness results as \cref{sec:characterization-to-correctness}.

\subsection{Tail-recursive mergesort$^\dag$}
\label{sec:tailrec-mergesort}

Although the naive mergesort algorithms presented in \cref{sec:mergesort-approaches} achieves
optimal $\bigO(n \log n)$ time complexity, the \caml{merge} function (\cref{fig:nontailrec-merge})
in call-by-value evaluation consumes a linear amount of stack space and crashes on longer inputs,
since it is not tail recursive.
The commonly used technique to make it tail recursive is to add an \emph{accumulator} argument
\caml{accu} that is initially the empty list and accumulates the result, as in \caml{revmerge} in
\cref{fig:tailrec-merge}.\footnote{\label{fn:rev_append-is-tailrec}Although \caml{rev_append xs ys}
is equal to \caml{append (rev xs) ys} (\cref{lemma:coq:catrevE}), the former is tail recursive but
the latter is not. Therefore, the \caml{revmerge} function has to use the former to avoid using up
stack space.}
The tail-recursive merge function accumulates the first elements of the input lists as the last
element of the output list, and produces its output in reverse order. That is to say, the following
equation holds for any binary relation \caml{(<=)} and any lists \caml{xs} and \caml{ys}:
\[
\caml{revmerge (<=) xs ys []} = \caml{rev (merge (<=) xs ys)}.
\]

Two lists sorted in descending order can be merged without reversing them
using the converse relation \caml{(>=) := (fun x y -> y <= x)}.
In fact, the following equation holds for any total preorder
\caml{(<=)}, and any lists \caml{s1}, \caml{s2}, \caml{s3}, and \caml{s4} sorted \wrt
\caml{(<=)}:\footnote{\label{fn:revmerge-arguments-order}Note that swapping the arguments of the
outer \caml{revmerge} matters for maintaining the stability of the algorithm.}
\begin{align*}
      & \caml{merge (<=) (merge (<=) s1 s2) (merge (<=) s3 s4)} \\
 = {} & \caml{revmerge (>=) (revmerge (<=) s3 s4 []) (revmerge (<=) s1 s2 []) []}.
\end{align*}

There are other non-tail-recursive functions in the naive mergesort algorithms.
In \cref{fig:top-down-mergesort}, the splitting function \caml{split_n} is not tail recursive and
costs a linear time in its second argument \caml{k}.
In \cref{fig:bottom-up-mergesort}, the \caml{merge_pairs} and \caml{map} functions are not tail
recursive.
Note that the non tail-recursiveness of the top-down \caml{sort} function (or \caml{sort_rec} below)
is not an actual issue since the depth of its recursive calls is logarithmic in the length of the
input.
Therefore, based on the top-down approach, all the major inefficiency issues explained here can be
addressed as follows.
\begin{camlcode}
let sort (<=) xs =
  let (>=) = (fun x y -> y <= x) in
  let rec sort_rec xs b n =
    match n, xs with
    | 1, x :: xs' -> [x], xs'
    | _, _ ->
      let n1 = n / 2 in
      let s1, xs'  = sort_rec xs (not b) n1 in
      let s2, xs'' = sort_rec xs' (not b) (n - n1) in
      (if b then revmerge (>=) s2 s1 [] else revmerge (<=) s1 s2 []), xs''
  in
  if xs = [] then [] else fst (sort_rec xs true (length xs))
\end{camlcode}
The auxiliary recursive function \caml{sort_rec} takes three arguments: a list \caml{xs}, a Boolean
value \caml{b}, a positive integer \caml{n} that must be less than or equal to the length of
\caml{xs}.
It returns the pair of the sorted list of the first \caml{n} elements of \caml{xs} and the rest of
the input. The sorted list (first component) is in ascending order if \caml{b} is \caml{true},
otherwise in descending order.

\caml{List.stable_sort} of the \OCaml standard library follows the same approach as above.
Additionally, its auxiliary function corresponding to \caml{sort_rec} above is defined as two
mutually-recursive functions corresponding to the cases that \caml{b} is \caml{true} or
\caml{false}, respectively. It stops the recursion when $n \leq 3$, which is an effective
micro-optimization.
We will present a bottom-up tail-recursive mergesort in \cref{sec:tailrec-mergesort-in-coq} and make
it smooth (\cref{sec:smooth-mergesort}) in \cref{appx:tailrec-mergesort-in-coq}.


\subsection{Smooth mergesort$^\dag$}
\label{sec:smooth-mergesort}

A mergesort algorithm that takes advantage of sorted slices in the input is called
\emph{natural} mergesort~\cite[Algorithm N in Section 5.2.4]{DBLP:books/aw/Knuth73}.
In this paper, we instead call it \emph{smooth} mergesort~\cite{DBLP:journals/scp/Dijkstra82}%
\cite[Subsection ``Bottom-up merge sort'' in Section 3.21]{DBLP:books/daglib/0084777} to avoid
confusion with naturality (\cref{lemma:sort_map}).
Bottom-up non-tail-recursive mergesort (\cref{fig:bottom-up-mergesort}) can easily be made smooth by
dividing the input into weakly increasing or strictly decreasing slices instead of singleton lists.
Such a slice of the input that is already sorted is called a \emph{run}.
Note that we cannot use a non-strictly decreasing slice, because its reversal does not preserve the
order of equivalent elements in the slice, and thus, it breaks the stability of the algorithm.
An example of smooth bottom-up non-tail-recursive mergesort follows:
\begin{camlcode}[escapeinside=\#\#]
let sort (<=) xs =
  let rec merge_pairs = ... in (* These functions remain    *)
  let rec merge_all = ... in   (* unchanged from Figure #\ref{fig:bottom-up-mergesort}#. *)
  let rec sequences = function
    | a :: b :: xs -> if a <= b then ascending b [a] xs else descending b [a] xs
    | [a] -> [[a]]
    | [] -> []
  and ascending a accu = function
    | b :: xs when a <= b -> ascending b (a :: accu) xs
    | xs -> rev (a :: accu) :: sequences xs
  and descending a accu = function
    | b :: xs when not (a <= b) -> descending b (a :: accu) xs
    | xs -> (a :: accu) :: sequences xs
  in
  merge_all (sequences xs)
\end{camlcode}
where \caml{sequences} splits the input into sorted slices, and \caml{ascending} and
\caml{descending} process increasing and decreasing runs, respectively.

In fact, \GHC's mergesort function \haskell{Data.List.sort} is smooth bottom-up non-tail-recursive
mergesort. Its slightly modified versions have been formally verified in
\IsabelleHOL~\cite{DBLP:journals/jar/Sternagel13} and \Dafny~\cite{DBLP:journals/tocl/LeinoL15},
which we compare to our formalization in \cref{sec:related-work}.

\subsection{Extended characterization and correctness proofs}
\label{sec:new-proofs}

\subsubsection{Extended characterization of stable mergesort functions}
\label{sec:new-characterization}

In this section, we extend the characterization presented in \cref{sec:characterization} to
support tail-recursive and smooth mergesorts.
We first add more operators on $T$ and $R$ to the type of abstract sort functions \asort{}, as
follows:
\begin{align*}
 \TasortExt \coloneq{} \forall (T \, R : \types{}),
 &\underbrace{(T \to T \to \mathrm{bool})}_{\text{relation $\leq$}} \to
  \underbrace{(R \to R \to R)}_{\text{merge by $\leq$}} \to
  \underbrace{(R \to R \to R)}_{\text{merge by $\geq$}} \to \\
 &\underbrace{(T \to R)}_{\text{singleton}} \to
  \underbrace{R\vphantom{()}}_{\text{empty}} \to
  \underbrace{\mathrm{list} \, T\vphantom{()}}_{\text{input}} \to
  \underbrace{R\vphantom{()}}_{\text{output}}.
\end{align*}
The first argument of type $T \to T \to \mathrm{bool}$ is there to give \asort{} direct access
to the relation $\leq$ without going through merge, which we will exploit to support smooth
mergesorts.
Since tail-recursive mergesorts merge sorted sequences both by $\leq$ and $\geq$, the second and
third arguments of type $R \to R \to R$ now represent merge by $\leq$ and $\geq$, respectively.
However, $R$ still represents the type of lists sorted \wrt $\leq$. In order to merge them with
$\geq$, we introduce the following operator $\mergerev$:
\[
 \mathit{xs} \mergerev_\leq \mathit{ys} \coloneq
 \texttt{rev} \, (\texttt{rev} \, \mathit{ys} \mathbin{\merge_\geq} \texttt{rev} \, \mathit{xs}).
\]
Therefore, we replace \cref{eq:asort_mergeE,eq:asort_catE} with the following equations,
respectively:
\begin{gather}
\forall (T : \types{}) \, ({\leq} : T \to T \to \mathrm{bool}) \, (\mathit{xs} : \mathrm{list} \, T),
 \asort{} \, (\leq) \, (\merge_\leq) \, (\mergerev_\leq) \, [\cdot] \, [] \, \mathit{xs} =
 \sort{}_\leq \, \mathit{xs},
 \label{eq:asort_mergeE'} \\
\forall (T : \types{}) \, ({\leq} : T \to T \to \mathrm{bool}) \, (\mathit{xs} : \mathrm{list} \, T),
 \asort{} \, (\leq) \, (\concat) \, (\concat) \, [\cdot] \, [] \, \mathit{xs} = \mathit{xs}.
 \label{eq:asort_catE'}
\end{gather}
We define the \emph{extended characteristic property} of stable mergesort functions as the existence
of an abstract mergesort function \asort{} that is parametric (\cref{sec:prelim-param}), \ie,
$(\asort{},\asort{}) \in \paramt{\TasortExt}$, and satisfies \cref{eq:asort_mergeE',eq:asort_catE'}.

\subsubsection{Tail-recursive mergesort$^\dag$}
\label{sec:new-characterization-tailrec}

The abstract mergesort function for the tail-recursive mergesort function
(\cref{sec:tailrec-mergesort}) can be obtained just by abstracting out the tail-recursive merge
function with \caml{(<=)} and \caml{(>=)} to the two abstract merge functions as follows:
\begin{camlcode}[escapeinside=\#\#]
let asort (<=) merge merge' singleton empty xs =
  let rec sort_rec xs b n =
    match n, xs with
    | 1, x :: xs' -> #\underline{singleton}# x, xs'
    | _, _ ->
      let n1 = n / 2 in
      let s1, xs' = sort_rec xs (not b) n1 in
      let s2, xs'' = sort_rec xs' (not b) (n - n1) in
      (if b then #\underline{merge'}# s1 s2 else #\underline{merge}# s1 s2), xs''
  in
  if xs = [] then #\underline{empty}# else fst (sort_rec xs true (length xs))
\end{camlcode}
By instantiating the above \caml{asort} as in \cref{eq:asort_mergeE'}, we replace
\caml{revmerge (<=) s1 s2 []} and \caml{revmerge (>=) s2 s1 []} in \caml{sort} in
\cref{sec:tailrec-mergesort} with \caml{merge (<=) s1 s2} and
\caml{rev (merge (>=) (rev s2) (rev s1))}, respectively.
While the sorted lists that appear in execution of \caml{sort} in \cref{sec:tailrec-mergesort} are a
mix of increasing and decreasing lists, this replacement turns all of them in increasing order.
Therefore, the proof of \cref{eq:asort_mergeE'} cannot be done just by definition, and involves some
equational reasoning about merge and reversal of lists (see \cref{sec:limitations}).

\subsubsection{Smooth mergesort$^\dag$}
\label{sec:new-characterization-smooth}

The major obstacle in defining the abstract mergesort function for the smooth mergesort
(\cref{sec:smooth-mergesort}) is that it uses the cons (\caml{::}) which does not directly
correspond to any of the four abstract operators \caml{merge}, \caml{merge'}, \caml{singleton}, and
\caml{empty} on sorted lists.
For example, the recursive function \caml{descending} uses \caml{a :: accu} knowing that
\caml{a} is strictly smaller than the head of \caml{accu} and \caml{accu} is strictly increasing,
and thus, ensures \caml{a :: accu} is strictly increasing.
Therefore, we can simulate this behavior with \caml{merge accu (singleton a)}.
Similarly, we can simulate the behavior of \caml{a :: accu} in \caml{ascending}, where \caml{a} is
greater than or equal to the head of \caml{accu} and \caml{accu} is weakly decreasing, with
\caml{merge' accu (singleton a)}.
Again, the proof of \cref{eq:asort_mergeE'} cannot be done just by definition, and involves
these arguments and reasoning about reversal of lists (see \cref{sec:limitations}).

\subsubsection{Correctness proofs}
\label{sec:new-characterization-to-correctness}

In this section, we adapt the correctness proofs presented in
\cref{sec:characterization-to-correctness} to the extended characteristic property.
We adapt the induction principle over traces (\cref{lemma:sort_ind}) as follows.
\begin{lemma}[An induction principle over traces of \sort{}]
 \label{lemma:sort_ind'}
 Suppose $\leq$ and $\sim$ are binary relations on $T$ and $\mathrm{list} \, T$, respectively, and
 $\mathit{xs}$ is a list of type $\mathrm{list} \, T$.
 Then, $\mathit{xs} \sim \sort{}_\leq \, \mathit{xs}$ holds whenever the following four
 induction cases hold:
 \begin{itemize}
  \item for any lists $\mathit{xs}$, $\mathit{xs'}$, $\mathit{ys}$, and $\mathit{ys'}$ of type
	$\mathrm{list} \, T$,
	$(\mathit{xs} \concat \mathit{ys}) \sim (\mathit{xs'} \merge_\leq \mathit{ys'})$ holds
	whenever $\mathit{xs} \sim \mathit{xs'}$ and $\mathit{ys} \sim \mathit{ys'}$ hold,
  \item for any lists $\mathit{xs}$, $\mathit{xs'}$, $\mathit{ys}$, and $\mathit{ys'}$ of type
	$\mathrm{list} \, T$,
	$(\mathit{xs} \concat \mathit{ys}) \sim
	 \texttt{rev} \, (\texttt{rev} \, \mathit{ys'} \merge_\geq \texttt{rev} \, \mathit{xs'})$
	holds whenever $\mathit{xs} \sim \mathit{xs'}$ and $\mathit{ys} \sim \mathit{ys'}$ hold,
  \item for any $x$ of type $T$, $[x] \sim [x]$ holds, and
  \item $[] \sim []$ holds.
 \end{itemize}
\end{lemma}

The only difference between the old and new induction principles is the addition of the second
induction case.
Therefore, the proofs of \cref{lemma:perm_sort,lemma:sort_pairwise_stable} can be easily adapted to
the new induction principle by adding the corresponding case, as follows.
Again, we refer the readers to \cref{lemma:coq:pairwise_sort} for the complete proof of
\cref{lemma:pairwise_sort}.
\begin{proof}[Extension to the proof of \cref{lemma:perm_sort}]
 We prove $\sort{}_\leq \, \mathit{xs} \permeq \mathit{xs}$ by induction on
 $\sort{}_\leq \, \mathit{xs}$ (\cref{lemma:sort_ind'}).
 Suffice it to show that $\mathit{xs'} \permeq \mathit{xs}$ and $\mathit{ys'} \permeq \mathit{ys}$
 imply:
 \begin{align*}
   \texttt{rev} \, (\texttt{rev} \, \mathit{ys'} \mathbin{\merge_\geq} \texttt{rev} \, \mathit{xs'})
   &\permeq \texttt{rev} \, \mathit{ys'} \mathbin{\merge_\geq} \texttt{rev} \, \mathit{xs'}
   & (\text{\cref{lemma:coq:perm_rev}}) \\
   &\permeq \texttt{rev} \, \mathit{ys'} \concat \texttt{rev} \, \mathit{xs'}
   & (\text{\cref{lemma:coq:perm_merge}}) \\
   &=       \texttt{rev} \, (\mathit{xs'} \concat \mathit{ys'})
   & (\text{\cref{lemma:coq:rev_cat}}) \\
   &\permeq \mathit{xs'} \concat \mathit{ys'}
   & (\text{\cref{lemma:coq:perm_rev}}) \\
   &\permeq \mathit{xs} \concat \mathit{ys}.
   & (\text{\cref{lemma:coq:perm_cat}})
 \end{align*}
 The other induction cases are done in the first proof of \cref{lemma:perm_sort}.
\end{proof}
\begin{proof}[Extension to the proof of \cref{lemma:sort_pairwise_stable}]
 We prove a generalized proposition:
 \[
 (\sort{}_{\leq_1} \, \mathit{xs} \subseteq \mathit{xs}) \land
 (\text{$\mathit{xs}$ is pairwise sorted \wrt $\leq_2$} \Rightarrow
  \text{$\sort{}_{\leq_1} \, \mathit{xs}$ is sorted \wrt $\leq_{\mathrm{lex}}$})
 \]
 by induction on $\sort{}_{\leq_1} \, \mathit{xs}$ (\cref{lemma:sort_ind'}), where
 $x \leq_{\mathrm{lex}} y \coloneq x \leq_1 y \land (y \not\leq_1 x \lor x \leq_2 y)$.

 For the first component of the conjunction, suffice it to show
 \[
  (\mathit{xs'} \subseteq \mathit{xs}) \land (\mathit{ys'} \subseteq \mathit{ys}) \Rightarrow
  (\texttt{rev} \, (\texttt{rev} \, \mathit{ys'} \mathbin{\merge_{\geq_1}} \texttt{rev} \, \mathit{xs'})
   \subseteq \mathit{xs} \concat \mathit{ys})
 \]
 which is obvious since
 $\texttt{rev} \, (\texttt{rev} \, \mathit{ys'} \mathbin{\merge_{\geq_1}} \texttt{rev} \, \mathit{xs'})$
 is a permutation of $\mathit{xs'} \concat \mathit{ys'}$.

 For the second component of the conjunction, suffice it to show that
 $\texttt{rev} \, (\texttt{rev} \, \mathit{ys'} \mathbin{\merge_{\geq_1}} \texttt{rev} \, \mathit{xs'})$
 is sorted \wrt $\leq_{\mathrm{lex}}$, or equivalently, its reversal
 $\texttt{rev} \, \mathit{ys'} \mathbin{\merge_{\geq_1}} \texttt{rev} \, \mathit{xs'}$ is sorted \wrt
 the converse of $\leq_{\mathrm{lex}}$ (\cref{lemma:coq:rev_sorted}):
 \[
  x \geq_{\mathrm{lex}} y \coloneq x \geq_1 y \land (y \not\geq_1 x \lor x \geq_2 y),
 \]
 whenever:
 \begin{enumerate}[label=(\roman*)]
  \item \label{item:xs_ys_subset'} $\mathit{xs'}$ (resp.~$\mathit{ys'}$) is a subset of $\mathit{xs}$
	(resp.~$\mathit{ys}$),
  \item \label{item:xs_ys_sorted'} $\mathit{xs'}$ (resp.~$\mathit{ys'}$) is sorted \wrt
	$\leq_{\mathrm{lex}}$ if $\mathit{xs}$ (resp.~$\mathit{ys}$) is pairwise sorted \wrt $\leq_2$,
	and
  \item \label{item:concat_pairwise'} $\mathit{xs} \concat \mathit{ys}$ is pairwise sorted \wrt
	$\leq_2$.
 \end{enumerate}
 Among these, \ref{item:concat_pairwise'} is equivalent to the following conjunction
 (\cref{lemma:coq:pairwise_cat}):
 \begin{enumerate}[label=(\roman*), resume]
  \item \label{item:xs_ys_pairwise'} both $\mathit{xs}$ and $\mathit{ys}$ are pairwise sorted \wrt
	$\leq_2$, and
  \item \label{item:xs_le2_ys'} $x \leq_2 y$ holds for any $x \in \mathit{xs}$ and $y \in \mathit{ys}$.
 \end{enumerate}
 Hypotheses \ref{item:xs_ys_sorted'} and \ref{item:xs_ys_pairwise'} imply that both
 $\mathit{xs'}$ and $\mathit{ys'}$ are sorted \wrt $\leq_{\mathrm{lex}}$, or equivalently,
 $\texttt{rev} \, \mathit{xs'}$ and $\texttt{rev} \, \mathit{ys'}$ are sorted \wrt
 $\geq_{\mathrm{lex}}$ (\cref{lemma:coq:rev_sorted}).
 Hypotheses \ref{item:xs_ys_subset'} and \ref{item:xs_le2_ys'} imply that $y \geq_2 x$ holds for
 any $y \in \texttt{rev} \, \mathit{ys'}$ and $x \in \texttt{rev} \, \mathit{xs'}$
 (\cref{lemma:coq:sub_all,lemma:coq:allrel_rev2}).
 These two facts suffice to show that
 $\texttt{rev} \, \mathit{ys'} \mathbin{\merge_{\geq_1}} \texttt{rev} \, \mathit{xs'}$ is sorted \wrt
 $\geq_{\mathrm{lex}}$ (\cref{lemma:coq:merge_stable_sorted}).

 The other induction cases are done in the first proof of \cref{lemma:sort_pairwise_stable}.
\end{proof}

The naturality of \caml{sort} (\cref{lemma:sort_map}), which remains the same, can be deduced from
the parametricity of \asort{} as well.
Therefore, the rest of the correctness proofs, which have been verified in \Coq, remains the same.

\section{Formalization in \Coq}
\label{sec:formalization}

In this section, we discuss two technical aspects of our formalization of mergesort functions and
their correctness proofs in \Coq.
We first review a technique to make bottom-up mergesorts structurally
recursive~\cite{Gonthier:2009}\footnote{Gonthier introduced this technique to the Mathematical
Components (\MC) library~\cite{mathcomp} in 2008. The same technique has been used in the \Coq
standard library~\cite{rocqrefman:library,herbelin:f698148}.}, so that their termination becomes
trivial for \Coq (\cref{sec:termination}).
Furthermore, this technique makes the balanced binary tree construction of bottom-up mergesort
tail-recursive (\cref{sec:nontailrec-mergesort-in-coq}), and thus allows us to implement bottom-up
tail-recursive mergesort (\cref{sec:tailrec-mergesort-in-coq}).
We second discuss the design and organization of the library, particularly, the interface for
mergesort functions which allows us to state our correctness lemmas polymorphically for any stable
mergesort function, and how to populate this interface with concrete mergesort functions
(\cref{sec:interface}).

While \cref{sec:termination} continues to use the \OCaml syntax to present new mergesort functions,
\cref{appx:mergesort-in-coq} presents our actual \Coq implementations of structurally-recursive
mergesort functions, including some optimized implementations such as smooth variants
(\cref{sec:smooth-mergesort}).

\subsection{Structurally-recursive bottom-up mergesorts$^\dag$}
\label{sec:termination}

\subsubsection{The syntactic guard condition}
\label{sec:guard-condition}

A fixpoint function $f$ in \Coq~\cite{rocqrefman:inductive} has the form of
\coq|(fix $f_\mathrm{rec}$ ($\vec{x} : \vec{A}$) {struct $x_k$} : $B$ := $M$)| where
$f_\mathrm{rec}$ is the local name of the fixpoint function bound in $M$,
$(\vec{x} : \vec{A})$ is the list of arguments,
$x_k$ is the $k^\mathrm{th}$ element of $\vec{x}$ and the recursive (decreasing) argument,
and $M$ is the function body of type $B$.
This fixpoint function $f$ has type \coq{(forall ($\vec{x} : \vec{A}$), $B$)}.
To ensure the termination of $f$, all recursive calls of $f_\mathrm{rec}$ in $M$ must be
\emph{guarded by destructors}~\cite{DBLP:conf/types/Gimenez94} in \Coq. \pagebreak That is to say, $M$ must do
recursive calls to $f_\mathrm{rec}$ only on strict subterms of $x_k$.
In practice, the annotation of decreasing argument \coq|{struct $x_k$}| may be left implicit, and
\Coq can infer it automatically.
Hereafter in this section, we explain how to make bottom-up mergesorts structurally recursive, but
in the \OCaml syntax with annotations of decreasing argument as comments, \eg,
\caml{(* struct xs *)}.

As the first example of termination checking, we use the non-tail-recursive merge function.
Its definition in \cref{fig:nontailrec-merge} already does not satisfy the syntactic guard
condition, because the first and second recursive calls of \caml{merge} are decreasing only on the
first and the second lists, respectively, and the syntactic guard condition does not take a
termination argument involving multiple parameters (\eg, lexicographic termination) into account.

\begin{figure}[t]
\begin{camlcode}
let rec merge (<=) xs ys (* struct xs *) =
  match xs with
  | [] -> ys
  | x :: xs' ->
    let rec merge' ys (* struct ys *) =
      match ys with
      | [] -> xs
      | y :: ys' -> if x <= y then x :: merge (<=) xs' ys else y :: merge' ys'
    in
    merge' ys

let sort (<=) =
  let rec push xs stack (* struct stack *) =
    match stack with
    | [] :: stack | ([] as stack) -> xs :: stack
    | ys :: stack -> [] :: push (merge (<=) ys xs) stack
  in
  let rec pop xs stack (* struct stack *) =
    match stack with
    | [] -> xs
    | ys :: stack -> pop (merge (<=) ys xs) stack
  in
  let rec sort_rec stack xs (* struct xs *) =
    match xs with
    | [] -> pop [] stack
    | x :: xs -> sort_rec (push [x] stack) xs
  in
  sort_rec []
\end{camlcode}
\caption{Structurally-recursive non-tail-recursive merge and mergesort in \OCaml.}
\label{fig:struct-nontailrec-mergesort}
\end{figure}

One way to work around this restriction is to use nested fixpoint as in \caml{merge} in
\cref{fig:struct-nontailrec-mergesort}.
The outer recursive function \caml{merge} performs recursion on \caml{xs}, and the first case
where \caml{x <= y} holds calls it with \coq{xs'}, which is obtained by destructing \coq{xs} and
hence a strict subterm of \coq{xs}.
Similarly, the inner recursive function \coq{merge'} performs recursion on \coq{ys}, and the second
case where \caml{x <= y} does not hold calls it with \coq{ys'}, which is a strict subterm of
\coq{ys}.
This way, the termination checker can confirm that the \coq{merge} function terminates for any
input.


\subsubsection{Non-tail-recursive mergesort}
\label{sec:nontailrec-mergesort-in-coq}

In \cref{fig:bottom-up-mergesort}, \caml{merge_all} does not satisfy the syntactic guard condition
since \caml{merge_pairs xs} is not a strict subterm of \caml{xs}.
To work around this issue, we manage pending mergings using an explicit stack~\cite{Gonthier:2009}.
We represent the stack of pending mergings as a list of sorted lists whose head (\nth{0}) and tail
elements respectively correspond to the top and bottom elements of the stack.
The sorting process proceeds by repetitively pushing items to be sorted to the stack, as in
\caml{push} in \cref{fig:struct-nontailrec-mergesort}.
An item $\mathit{xs}$ pushed to the top of the stack is called a \emph{merging of level 0}, and a
result of merging two mergings of level $n$ is called a \emph{merging of level $n + 1$}.
The $n^\mathrm{th}$ element of the stack must be a merging of level $n$ if any; otherwise, the
empty list.
Therefore, if the top of the stack (or the stack itself) is empty, pushing an item $\mathit{xs}$ is
done by replacing the top element with $\mathit{xs}$; otherwise---if the stack $S$ has the form of
$\mathit{xs}_0 :: S'$ where $\mathit{xs}_0$ is nonempty---, it is done by replacing the top element
$\mathit{xs}_0$ with the empty list and pushing $\mathit{xs}_0 \merge_\leq \mathit{xs}$%
\footnote{Note that the first item of the input list will be pushed first and placed as a part of
the bottom element of the stack. The ordering of the arguments $\mathit{xs}_0$ and $\mathit{xs}$
here matters for maintaining the stability of the algorithm.} to $S'$.
In terms of trace, this procedure can be seen as a technique to construct balanced binary trees, and
the $n^\mathrm{th}$ element of the stack corresponds to a perfect binary tree of height $n$, as
in \cref{fig:explicit-stack}.
The sorting process can be completed by pushing all the items in the input to the stack and then
folding the stack by \caml{merge}, as in \caml{sort_rec} and \caml{pop} in
\cref{fig:struct-nontailrec-mergesort}, respectively.

\begin{figure}[t]
 \tikzset{
  level/.style = {sibling distance={3.2cm/pow(2,#1)}, level distance=.5cm},
  every node/.style = {fill=black!20, circle, inner sep=-2pt, minimum width=1em},
  merge/.style = {font={$\merge$}}}
 \begin{subfigure}[t]{.5\textwidth}
  \[
   []
  \]
  \caption{The initial state of stack, which is empty.}
  \label{fig:explicit-stack-1}
 \end{subfigure}%
 \begin{subfigure}[t]{.5\textwidth}
  \[
   [\begin{tikzpicture}[baseline=(base.base)]
     \node (base) {1};
    \end{tikzpicture}]
  \]
  \caption{The stack after pushing $1$ to \subref{fig:explicit-stack-1}.}
  \label{fig:explicit-stack-2}
 \end{subfigure}
 \begin{subfigure}[t]{.5\textwidth}
  \[
   [\xmark\,;\,
    \begin{tikzpicture}[baseline=(base.base)]
     \node[merge] {}
       child{ node (base) {1} }
       child{ node {2} };
    \end{tikzpicture}]
  \]
  \caption{The stack after pushing $2$ to \subref{fig:explicit-stack-2}.
  The new item $2$ has been merged with the top element $1$ to obtain $1 \merge 2$.}
  \label{fig:explicit-stack-3}
 \end{subfigure}%
 \begin{subfigure}[t]{.5\textwidth}
  \[
   [\begin{tikzpicture}[baseline=(base.base)]
     \node (base) {3};
    \end{tikzpicture}\,;\,
    \begin{tikzpicture}[baseline=(base.base)]
     \node[merge] {}
       child{ node (base) {1} }
       child{ node {2} };
    \end{tikzpicture}]
  \]
  \caption{The stack after pushing $3$ to \subref{fig:explicit-stack-3}.}
  \label{fig:explicit-stack-4}
 \end{subfigure}
 \begin{subfigure}[t]{.5\textwidth}
  \[
   [\xmark\,;\,
    \xmark\,;\,
    \begin{tikzpicture}[level/.append style = {sibling distance={3.6cm/pow(2,#1)}}, baseline=(base.base)]
     \node[merge] {}
       child{ node[merge] {} child{ node (base) {1} } child{ node {2} } }
       child{ node[merge] {} child{ node {3} } child{ node {4} } };
    \end{tikzpicture}]
  \]
  \caption{The stack after pushing $4$ to \subref{fig:explicit-stack-4}.
  The new item $4$ has been merged with the top element $3$ and then with the second
  element $1 \merge 2$, to obtain $(1 \merge 2) \merge (3 \merge 4)$.}
  \label{fig:explicit-stack-5}
 \end{subfigure}%
 \begin{subfigure}[t]{.5\textwidth}
  \[
   [\begin{tikzpicture}[baseline=(base.base)]
     \node (base) {5};
    \end{tikzpicture}\,;\,
    \xmark\,;\,
    \begin{tikzpicture}[level/.append style = {sibling distance={3.6cm/pow(2,#1)}}, baseline=(base.base)]
     \node[merge] {}
       child{ node[merge] {} child{ node (base) {1} } child{ node {2} } }
       child{ node[merge] {} child{ node {3} } child{ node {4} } };
    \end{tikzpicture}]
  \]
  \caption{The stack after pushing $5$ to \subref{fig:explicit-stack-5}.}
  \label{fig:explicit-stack-6}
 \end{subfigure}
 \begin{subfigure}[t]{.5\textwidth}
  \[
   [\xmark\,;\,
    \begin{tikzpicture}[baseline=(base.base)]
     \node[merge] {}
       child{ node (base) {5} }
       child{ node {6} };
    \end{tikzpicture}\,;\,
    \begin{tikzpicture}[level/.append style = {sibling distance={3.6cm/pow(2,#1)}}, baseline=(base.base)]
     \node[merge] {}
       child{ node[merge] {} child{ node (base) {1} } child{ node {2} } }
       child{ node[merge] {} child{ node {3} } child{ node {4} } };
    \end{tikzpicture}]
  \]
  \caption{The stack after pushing $6$ to \subref{fig:explicit-stack-6}.
  The new item $6$ has been merged with the top element $5$ to obtain $5 \merge 6$.}
  \label{fig:explicit-stack-7}
 \end{subfigure}%
 \begin{subfigure}[t]{.5\textwidth}
  \[
   [\begin{tikzpicture}[baseline=(base.base)]
     \node (base) {7};
    \end{tikzpicture}\,;\,
    \begin{tikzpicture}[baseline=(base.base)]
     \node[merge] {}
       child{ node (base) {5} }
       child{ node {6} };
    \end{tikzpicture}\,;\,
    \begin{tikzpicture}[level/.append style = {sibling distance={3.6cm/pow(2,#1)}}, baseline=(base.base)]
     \node[merge] {}
       child{ node[merge] {} child{ node (base) {1} } child{ node {2} } }
       child{ node[merge] {} child{ node {3} } child{ node {4} } };
    \end{tikzpicture}]
  \]
  \caption{The stack after pushing $7$ to \subref{fig:explicit-stack-7}.}
  \label{fig:explicit-stack-8}
 \end{subfigure}
 \caption{An example of state transitions of the explicit stack of pending mergings in the sorting
 process explained in \cref{sec:nontailrec-mergesort-in-coq}. The cross mark \xmark denotes an
 empty element in the stack. In general, pushing a new item $\mathit{xs}$ to the stack is done by
 (1) replacing the longest consecutive subsequence of nonempty elements
     $\mathit{xs}_0, \dots, \mathit{xs}_n$ from the top of the stack, where $\mathit{xs}_0$ is the
     top element, with empty elements, and then
 (2) replacing the next empty element $\mathit{xs}_{n + 1}$ in the stack with
 $\mathit{xs}_n \merge (\dots \merge (\mathit{xs}_0 \merge \mathit{xs})\dots)$.}
 \label{fig:explicit-stack}
\end{figure}

\citeauthor{Gonthier:2009}'s mergesort (\cref{fig:struct-nontailrec-mergesort}), particularly its
construction of balanced binary trees, can be seen as a variation of smooth bottom-up
non-tail-recursive mergesort presented by \citet{o1982smooth} and reviewed by
\citet[Subsection ``Bottom-up merge sort'' in Section 3.21]{DBLP:books/daglib/0084777}.
\citeauthor{o1982smooth}'s mergesort does not keep empty lists in the stack, but recovers the
information encoded by empty lists by counting the number of pushes performed on the stack, because
one may know whether the stack contains a merging of level $n$ or not by testing the $n^\mathrm{th}$
binary digit of the counter.
To put it another way, the stack in our implementation can be seen as a binary natural number in the
least significant bit first form representing the counter, by replacing the empty and non-empty
elements with 0 and 1, respectively.
Regardless of which approach we choose, bottom-up mergesort using the explicit stack of pending
mergings does not use up stack space except for \caml{merge}, because the depth of recursive calls
of \caml{push} is logarithmic in the length of the input, and \caml{pop} and \caml{sort_rec} are
tail recursive.
In \cref{sec:tailrec-mergesort-in-coq}, we implement bottom-up tail-recursive mergesort by
relying on this observation.

Structurally-recursive non-tail-recursive mergesort can be made smooth (\cref{sec:smooth-mergesort})
by pushing sorted slices to the stack instead of singleton lists.
See \cref{appx:nontailrec-mergesort-in-coq} for its \Coq implementation.

Since the emptiness test for sorted lists is not supplied to the abstract mergesort function as
defined in \cref{sec:new-characterization}, we change the type of the stack in the abstract
mergesort from \caml{list R} to \caml{list (option R)} and use the \caml{None} constructor of the
\caml{option} type to represent empty lists in the stack.
Therefore, the proof of \cref{eq:asort_mergeE'} for structurally-recursive mergesorts cannot be done
just by definition (\cref{sec:limitations}).

\subsubsection{Tail-recursive mergesort}
\label{sec:tailrec-mergesort-in-coq}

\begin{figure}[t]
\begin{camlcode}
let rec revmerge (<=) xs ys accu (* struct xs *) =
  match xs with
  | [] -> rev_append ys accu
  | x :: xs' ->
    let rec revmerge' ys accu (* struct ys *) =
      match ys with
      | [] -> xs
      | y :: ys' ->
        if x <= y then
          revmerge (<=) xs' ys (x :: accu)
        else
          revmerge' ys' (y :: accu)
    in
    revmerge' ys accu

let sort (<=) =
  let (>=) = (fun x y -> y <= x) in
  let rec push xs stack (* struct stack *) =
    match stack with
    | [] :: stack | ([] as stack) -> xs :: stack
    | ys :: [] :: stack | ys :: ([] as stack) ->
      [] :: revmerge (<=) ys xs [] :: stack
    | ys :: zs :: stack ->
      [] :: [] :: push (revmerge (>=) (revmerge (<=) ys xs []) zs []) stack
  in
  let rec pop mode xs stack (* struct stack *) =
    match stack, mode with
    | [], true -> rev xs
    | [], false -> xs
    | [] :: [] :: stack, _ -> pop mode xs stack
    | [] :: stack, _ -> pop (not mode) (rev xs) stack
    | ys :: stack, true -> pop false (revmerge (>=) xs ys []) stack
    | ys :: stack, false -> pop true (revmerge (<=) ys xs []) stack
  in
  let rec sort_rec stack xs (* struct xs *) =
    match xs with
    | [] -> pop false [] stack
    | x :: xs -> sort_rec (push [x] stack) xs
  in
  sort_rec []
\end{camlcode}
\caption{Structurally-recursive tail-recursive merge and mergesort in \OCaml.}
\label{fig:struct-tailrec-mergesort}
\end{figure}

The nested fixpoint technique in defining a recursive function that have more than one recursive
argument (\cref{sec:guard-condition}) can be easily adapted to the tail-recursive merge function
(\cref{fig:tailrec-merge}), as in \caml{revmerge} in \cref{fig:struct-tailrec-mergesort}.
As noted in \cref{sec:tailrec-mergesort}, this function produces its output in reverse order.

If we take reversals done by \caml{revmerge} into account in our construction scheme of balanced
binary trees (\cref{sec:nontailrec-mergesort-in-coq}), the \caml{push} function can be adapted as in
\cref{fig:struct-tailrec-mergesort}.
In this new \caml{push} function, pending mergings of ascending and descending orders should appear
alternately in the stack, and its top element and $\mathit{xs}$ should always be ascending ones.
Its one recursion step processes two elements of the stack to maintain this recursion
invariant, and pushing an item $\mathit{xs}$ to the stack proceeds as follows.
\begin{itemize}
 \item If the top of the stack (or the stack itself) is empty, it replaces the top element with
       $\mathit{xs}$.
 \item If the top of the stack $\mathit{ys}$ is nonempty and the next element $\mathit{zs}$ is empty
       (or the length of the stack is 1), it replaces $\mathit{zs}$ with
       $\texttt{rev} \, (\mathit{ys} \merge_\leq \mathrm{xs})$ and $\mathit{ys}$ with an empty
       element.
 \item Otherwise---if the stack has the form of $\mathit{ys} :: \mathit{zs} :: S'$ where both
       $\mathit{ys}$ and $\mathit{zs}$ are nonempty---, it has to push the result of merging
       $\mathit{xs}$, $\mathit{ys}$, and $\mathit{zs}$ to $S'$, where $\mathit{xs}$ and
       $\mathit{ys}$ are ascending but $\mathit{zs}$ is descending.
       Therefore, it pushes
       $\texttt{rev} \, (\mathit{zs} \merge_\geq \texttt{rev} \, (\mathit{ys} \merge_\leq \mathit{xs}))$
       to $S'$ and replaces $\mathit{xs}$ and $\mathit{ys}$ with empty elements.
\end{itemize}

As in \cref{sec:nontailrec-mergesort-in-coq}, the sorting process can be completed by pushing all
the items in the input to the stack and then folding the stack by \caml{revmerge} (\caml{sort_rec}
and \caml{pop} in \cref{fig:struct-tailrec-mergesort}, respectively).
In a recursive call of \caml{pop}, the head of the stack can either be ascending or descending in
contrast to \caml{push}.
Its additional argument \caml{mode} of type \caml{bool} is \caml{true} iff \caml{xs} and the head of
the stack are descending ones.
If an element of the stack to be processed is empty, \caml{pop} reverses \caml{xs} (in the \nth{4}
case), because \caml{xs} and the head of the stack have to be in the same order.
In order to avoid reversing \caml{xs} twice when the stack has a pair of two adjacent empty
elements, the \nth{3} case skips such empty elements just for performance reasons.

Note again that the recursive functions presented in this section are tail recursive except for
\caml{push}, whose depth of recursive calls is logarithmic in the length of the input.
Therefore, the \coq{sort} function above does not use up stack space.
Furthermore, it can be made smooth in the same way as the non-tail-recursive counterpart
(\cref{sec:nontailrec-mergesort-in-coq}).
See \cref{appx:tailrec-mergesort-in-coq} for its \Coq implementation.

\subsection{Interface for stable mergesorts}
\label{sec:interface}

In this section, we present the interface (\cref{sec:the-interface}) for mergesort functions
bundling the extended characteristic property
(\cref{sec:characterization,sec:new-characterization}), so that we can state our correctness lemmas
polymorphically for any stable mergesort function (\cref{sec:overloaded-correctness-lemmas}), and
explain how to populate (\cref{sec:populating-the-interface}) this interface with concrete mergesort
functions, in \Coq.

\subsubsection{The interface}
\label{sec:the-interface}

We first define the following constant \coq{asort_ty} for the type of abstract sort functions.
\begin{coqcode}
Definition asort_ty :=
  forall (T R : Type),
    (T -> T -> bool) -> (R -> R -> R) -> (R -> R -> R) -> (T -> R) -> R ->
    list T -> R.
\end{coqcode}
To automatically generate parametricity statements and proofs, we use \Paramcoq~\cite{paramcoq} that implements a version of the parametricity
translation $\paramt{\cdot}$~\cite{DBLP:conf/fossacs/BernardyL11, DBLP:journals/jfp/BernardyJP12, DBLP:conf/csl/KellerL12}
for \Coq{}, extended to terms of the calculus.
The abstraction theorem (\cref{assumption:abstraction}) for this parametricity translation can be
restated as follows.
\begin{theorem}[Abstraction theorem]
 \label{thm:abstraction}
 If $\vdash t : A$, then $\vdash \paramt{t} : \paramt{A} \, t \, t$.
\end{theorem}
\noindent Now, $\paramt{\cdot}$ translates a \Coq term $t : A$ to
a term of type $\paramt{A} \, t \, t$, which is a \Coq internalization of the membership $(t, t) \in \paramt{A}$ from~\cref{sec:prelim-param}, where the relation $\paramt{A}$ is internalized as a function $A \to A \to \types$.
Although we assumed the abstraction theorem in \cref{sec:nontailrec,sec:optimizations}
(\cref{assumption:abstraction}), and \cref{thm:abstraction} holds only for a fragment of the
underlying calculus of \Coq, we stress that our functional correctness proofs are axiom free thanks
to \Paramcoq, which produces parametricity proofs for all the \asort{} functions in our
formalization.

The following \coq{Parametricity} command from \Paramcoq declares the parametricity relation
$\paramt{\coq{asort_ty}}$ as a new constant \coq{asort_ty_R}.
\begin{coqcode}
Parametricity asort_ty.
\end{coqcode}

We second define the interface for stable mergesort functions, that is, a dependent record type
bundling a mergesort function and the characteristic property on it.
\begin{coqcode}
Structure stableSort := StableSort {
  apply : forall T : Type, (T -> T -> bool) -> list T -> list T;
  asort : asort_ty;
  asort_R : asort_ty_R asort asort;
  asort_mergeE : forall (T : Type) (leT : T -> T -> bool) (xs : list T),
    let geT x y := leT y x in
    let mergerev xs ys := rev (merge geT (rev ys) (rev xs)) in
    asort leT (merge leT) mergerev (fun x => [:: x]) [::] xs = apply leT xs;
  asort_catE : forall (T : Type) (leT : T -> T -> bool) (xs : list T),
    asort leT cat cat (fun x => [:: x]) [::] xs = xs;
}.
\end{coqcode}
where \coq{apply} is the mergesort function in question, \coq{asort} is the abstract version of
\coq{apply}, \coq{asort_R} is the parametricity of \coq{asort}, \coq{asort_mergeE}
and \coq{asort_catE} are respectively the two equational properties \eqref{eq:asort_mergeE'} and
\eqref{eq:asort_catE'} on \coq{asort}, \coq{stableSort} is the record type bundling these five
fields, and \coq{StableSort} is the constructor of \coq{stableSort}.

\subsubsection{Populating the interface}
\label{sec:populating-the-interface}

Suppose \coq{sort1} is a mergesort function and \coq{asort1} is its abstract version.
We state \cref{eq:asort_mergeE',eq:asort_catE'} as follows:
\begin{coqcode}
Fact asort1_mergeE (T : Type) (leT : T -> T -> bool) (xs : list T) :
  let geT x y := leT y x in
  let mergerev xs ys := rev (merge geT (rev ys) (rev xs)) in
  asort1 leT (merge leT) mergerev (fun x => [:: x]) [::] xs = sort1 leT xs.

Fact asort1_catE (T : Type) (leT : T -> T -> bool) (xs : list T) :
  asort1 leT cat cat (fun x => [:: x]) [::] xs = xs.
\end{coqcode}
and prove them by functional induction on \coq{asort1} and equational reasoning, as we saw in
\cref{sec:implementation-to-characterization}.
While \cref{eq:asort_mergeE'} can be proved by computation (by the \coq{reflexivity} tactic) in the
simplest cases, it is not the case for any structurally-recursive mergesort
(\cref{sec:termination}).
We discuss this limitation further in \cref{sec:limitations}.

Proving that \coq{asort1} is parametric can be done just by the parametricity translation
(\cref{sec:the-interface}):
\begin{coqcode}
Parametricity asort1. (* generates a parametricity proof asort1_R. *)
\end{coqcode}

Provided we have all these ingredients, we can construct an instance of the \coq{stableSort} record
type as follows.
\begin{coqcode}
Definition sort1_stable :=
  StableSort sort1 asort1 asort1_R asort1_mergeE asort1_catE.
\end{coqcode}

\subsubsection{Correctness lemmas}
\label{sec:overloaded-correctness-lemmas}

The correctness lemmas of mergesort (\cref{sec:characterization-to-correctness}) can be stated
polymorphically for any \coq{stableSort} instance.
For example, the commutation of \coq{filter} and mergesort functions (\cref{lemma:filter_sort}) can
be stated as follows:
\begin{coqcode}
Lemma filter_sort (sort : stableSort) (T : Type) (leT : T -> T -> bool) :
  total leT -> transitive leT ->
  forall (p : T -> bool) (xs : list T),
    filter p ($\color{gray}\texttt{apply}$ sort T leT xs) = $\color{gray}\texttt{apply}$ sort T leT (filter p xs).
\end{coqcode}
where $\color{gray}\texttt{apply}$ can be omitted by declaring it as an implicit
coercion~\cite{rocqrefman:coercion}, so that a \coq{stableSort} instance itself can be seen as a sort
function in the user-facing syntax:
\begin{coqcode}
Coercion apply : stableSort >-> Funclass.
\end{coqcode}

Also, we provide for several lemmas a version where hypotheses are localized by a predicate.
For example, the relation \coq{leT} in the above lemma \coq{filter_sort} (\cref{lemma:filter_sort})
only needs to be total and transitive on the domain delimited by a predicate \coq{P} provided all
elements of \coq{s} are in \coq{P}, as follows:
\begin{corollary}\label{corollary:filter_sort_in}
 The following equation holds
 \[
  \texttt{filter}_p \, (\sort{}_\leq \, \mathit{xs}) =
  \sort{}_\leq \, (\texttt{filter}_p \, \mathit{xs})
 \]
 whenever $\leq$ is total and transitive on a predicate $P \subseteq T$, meaning that
 $x \leq y \lor y \leq x$ and $x \leq y \land y \leq z \Rightarrow x \leq z$ hold for any
 $x, y, z \in P$, and all elements of $\mathit{xs}$ satisfy $P$.

 This corollary corresponds to the following lemma in \Coq:
\begin{coqcode}
Lemma filter_sort_in
  (sort : stableSort) (T : Type) (P : T -> bool) (leT : T -> T -> bool) :
  {in P &, total leT} -> {in P & &, transitive leT} ->
  forall (p : T -> bool) (xs : list T), all P xs ->
    filter p (sort T leT xs) = sort T leT (filter p xs).
\end{coqcode}
 where \coq|{in P &, Q}| reduces to \coq{forall x : T, P x -> forall y : T, P y -> R x y} given that
 \coq{Q} reduces to \coq{forall x y : T, R x y} and \coq{P} has type \coq{T -> bool},
 \coq|{in P & &, Q}| is its ternary version, and \coq{all P s} means that \coq{P} holds for any
 element of \coq{xs}.
\end{corollary}

\begin{proof}
 In dependent type theory, one may define a subtype \coq{sig P} collecting inhabitants of \coq{T}
 satisfying \coq{P}.
 Using the canonical \coq{val : sig P -> T} function, the relation \coq{leT} can be turned into a
 relation \coq{leP x y := leT (val x) (val y)} on \coq{sig P}, which is a total preorder thanks to
 the assumptions on \coq{leT} (\cref{lemma:coq:in2_sig}).
 Since \coq{P} holds for any element of \coq{xs}, we can replace \coq{xs} everywhere with
 \coq{map val xs'}, where \coq{xs'} is a list of type \coq{list (sig P)}
 (\cref{lemma:coq:all_sigP}).
 Thanks to the naturality of \coq{sort} (\cref{lemma:sort_map}) and \texttt{filter}
 (\cref{lemma:coq:filter_map}), and the congruence rule with respect to
 \coq{map}, it remains to proves
\begin{coqcode}
  filter p' (sort (sig P) leP xs') = sort (sig P) leP (filter p' xs')
\end{coqcode}
 where \coq{p' x := p (val x)}.
 We can then conclude by applying \coq{filter_sort} (\cref{lemma:filter_sort}).
\end{proof}

\Cref{appx:stablesort-theory} provides the list of all lemmas about stable sort functions solely
derived from the characterization, their formal statements in \Coq that use the \coq{stableSort}
structure, and their informal proofs.


\section{Limitations}
\label{sec:limitations}

As we mentioned in \cref{sec:new-characterization-tailrec,sec:new-characterization-smooth,%
sec:nontailrec-mergesort-in-coq,sec:populating-the-interface}, the present proof technique has the
limitation that the proof of \cref{eq:asort_mergeE'} cannot be done by definition (by computation in
\Coq) for elaborate variations of mergesort.
This is because these mergesorts inspect sorted lists by operations not supplied to the abstract
mergesort function, the abstract mergesort function has to simulate it in a non-trivial way, and
thus, the proof of \cref{eq:asort_mergeE'} involves the simulation arguments.
The examples of such operations appeared in the paper are
tail-recursive merge (\cref{sec:tailrec-mergesort,sec:new-characterization-tailrec}),
cons in smooth mergesorts (\cref{sec:smooth-mergesort,sec:new-characterization-smooth}), and
emptiness test in structurally-recursive mergesorts (\cref{sec:termination}).
Since we are primarily interested in structurally-recursive mergesorts in our formalization, the
proof of \cref{eq:asort_mergeE'} had to be done by hand for most variations.
Nevertheless, we argue that the proofs remain relatively short compared to related work
(\cref{sec:related-work}).

We might be able to relax the limitation by supplying more operators to the abstract mergesort
function, or by using a type system that has more conversion rules.
For example, we can consider splitting the abstract type $R$ representing sorted lists into two
abstract types $R$ and $R'$ respectively representing sorted lists in ascending and descending
orders, and supplying tail-recursive merge functions of types $R \to R \to R'$ and $R' \to R' \to R$,
to relax the limitation for tail-recursive mergesorts.
However, this solution would come at a cost that the stack of pending mergings in
structurally-recursive tail-recursive mergesort (\cref{sec:tailrec-mergesort-in-coq}) cannot be
represented as a homogeneous list, since abstract sorted lists of types $R$ and $R'$ should appear
alternately in the stack.
Similarly, we can consider supplying an emptiness test function to relax the limitation for
structurally-recursive mergesorts.
Any solution of supplying more operators would come at a cost that the abstract correctness proofs
have to be adapted as we extend the characterization with the new operators.
Exploring such possibilities is left as future work.

Since \GHC 9.12.1~\cite{GlasgowHaskell:9-12-1}, \haskell{Data.List.sort} has been changed to use
3-way and 4-way merge functions \haskell{merge3 as bs cs} and
\haskell{merge4 as bs cs ds}~\cite{core-libraries-committee:236}, which are documented as
manually-fused versions of \haskell{merge (merge as bs) cs} and
\haskell{merge (merge as bs) (merge cs ds)}, respectively.
However, \haskell{merge4 [] as bs cs} and  \haskell{merge4 as [] bs cs} reduce to
\haskell{merge3 as bs cs}, and thus, it changes the way it nests the 2-way merge.
Therefore, the 4-way merge actually cannot be expressed as a simple combinations of the 2-way merge;
thus, applying our proof technique to this mergesort would require to either fix the 4-way merge
function by adding a right-associative version of 3-way merge or extend the characteristic property
with the 4-way merge.

\section{Related work}
\label{sec:related-work}

\Cref{sec:mergesort-approaches} implies that mergesort conceptually consists of a construction of
a balanced binary tree and folding of this binary tree with merge, which is an instance of the fact
that many sorting algorithms can be explained as folds of unfolds, or, dually, as unfolds of
folds~\cite{DBLP:conf/icfp/HinzeJHWM12,DBLP:conf/birthday/HinzeMW13}.
The traces we informally show in \cref{fig:mergesort-traces} are reminiscent of the tree datatype
of~\citet[Section 6]{DBLP:conf/birthday/HinzeMW13}.
Our \asort{} functions for non-tail-recursive mergesorts (\cref{sec:characterization}) and their
\texttt{makeTree} functions have the same type, except that binary trees are Church-encoded in the
former (\cref{remark:induction}), and expose the underlying binary tree construction of mergesort.

Using parametricity to test or verify sort functions is not a new idea.
For example, Knuth's 0-1-Principle~\cite[Theorem Z of Section 5.3.4]{DBLP:books/aw/Knuth73} can be
deduced from parametricity~\cite{Day99logicalabstractions}.
While such an idea has been advanced towards both program
testing~\cite{DBLP:conf/popl/Voigtlander08, DBLP:journals/pacmpl/HouW22} and formal
verification~\cite{DBLP:conf/types/BoveC04}, the present work is the first application of
parametricity to prove the stability of sorting functions to the best of our knowledge.

The only prior work on formally proving the stability of mergesort is by
\citet{DBLP:journals/tocl/LeinoL15,Leroy:mergesort,DBLP:journals/jar/Sternagel13}.
Among these, \citet{Leroy:mergesort} proved a non-smooth bottom-up non-tail-recursive mergesort
(similar to \cref{fig:bottom-up-mergesort}) correct in \Coq (in 360 LoC),
\citet{DBLP:journals/jar/Sternagel13} proved a slightly modified version of \GHC's mergesort correct
in \IsabelleHOL (in 177 LoC), and
\citet{DBLP:journals/tocl/LeinoL15} ported the latter result to \Dafny (in 633 LoC).
All of them proved the permutation property (\cref{lemma:perm_sort}), sortedness
(\cref{lemma:coq:sort_sorted}), and stability (\cref{corollary:sort_standard_stable}).
In addition, \citet{DBLP:journals/jar/Sternagel13} proved extensional equality of the mergesort and
an insertion sort (\cref{lemma:coq:eq_sort_insert}).
On the other hand, our formalization consists of:
\begin{itemize}
\item 121 LoC for implementations and functional correctness proofs of 4 merge functions (hence
  around 31 LoC per merge function on average),
\item 16 LoC for implementation and characteristic property proof of the insertion sort
  (\cref{def:insertion-sort,lemma:insertion-sort-stable}), and
\item 533 LoC for implementations and characteristic property proofs of 8 mergesort functions listed
  in \cref{appx:mergesort-in-coq} (hence around 67 LoC per sort function on average), and
\item 472 LoC for generic functional correctness proofs and the infrastructure for 43 lemmas listed
  in \cref{appx:stablesort-theory}, but only 198 LoC to cover our main stability results
  (\cref{lemma:sort_pairwise_stable,lemma:filter_sort}) and the results in related work
  (\cref{lemma:perm_sort,lemma:coq:sort_sorted,corollary:sort_standard_stable,lemma:coq:eq_sort_insert}).
\end{itemize}
Hence, for:
\begin{itemize}
\item a single mergesort function, we would have on average 312 LoC ($= 31 + 16 + 67 + 198$) to get
  its implementation, permutation, sortedness, stability, extensional equality to the insertion
  sort, including all the infrastructure, and
\item all eight mergesort functions, it amounts to an average of 109 LoC per sort function for
  theorems up to \cref{corollary:sort_standard_stable}, and 143 LoC for all 43 lemmas of our theory
  of sort functions.
\end{itemize}

Therefore, although comparing LoC across different systems, standard libraries, coding styles, \etc
only allows for very crude comparison, we conclude that the only prior work similarly concise to our
proof is the one by \citet{DBLP:journals/jar/Sternagel13}.
Whether our proof is shorter depends on how we compare them: if we extract from our proof only one implementation it is 135 LoC longer than his, while if we divide the total LoC by the number of mergesort variants, it is 68 LoC shorter.
Now, the key ingredient for achieving the conciseness of \citet{DBLP:journals/jar/Sternagel13}'s proofs appear to be skillful functional
induction schemes and a simple invariant for proving \cref{corollary:sort_standard_stable} for GHC's
mergesort~\cite[Section 4]{DBLP:journals/jar/Sternagel13}.
However, this approach does not fully scale to tail-recursive mergesorts and the proof of
\cref{lemma:sort_pairwise_stable}, which is slightly more general than
\cref{corollary:sort_standard_stable} (\cref{sec:sort_standard_stable,appx:stability-statements}),
and thus, does not allow us to easily prove a large part of our functional correctness results.
Our smooth non-tail-recursive mergesort (an extension of \cref{sec:nontailrec-mergesort-in-coq},
shown in \cref{appx:nontailrec-mergesort-in-coq}) does not follow \GHC's mergesort as close as in
\citet{DBLP:journals/jar/Sternagel13,DBLP:journals/tocl/LeinoL15}, since \haskell{mergeAll}
function (\caml{merge_all} in \cref{fig:bottom-up-mergesort}) as is cannot be defined in \Coq due to
the syntactic guard condition. Nevertheless, this is a limitation of \Coq rather than a limitation
of our approach.

\Citet{functional_algorithms_verified} presented several verified sort algorithms, \eg, non-smooth
top-down and bottom-up non-tail-recursive mergesorts, and smooth mergesort presented by
\citet{DBLP:journals/jar/Sternagel13}.
However, a large part of their functional correctness proofs are done again for each implementation,
and thus, their proofs are not modular as ours.
In contrast to the present work, \citet{functional_algorithms_verified} verifies asymptotic
complexity bounds of mergesorts using automatically derived \emph{running time functions} modelling
a call-by-value semantics. However, this approach lacks a formal guarantee that the running time
function represents the actual running time and would not allow us to formally prove that
non-tail-recursive mergesort in call-by-need evaluation is an optimal incremental
sorting~\cite{DBLP:conf/alenex/ParedesN06}.

\Citet{DBLP:journals/jar/GouwBBHRS19} verified a smooth mergesort for arrays
called \emph{TimSort} taken from the OpenJDK core library using \Java
verification tool \KeY.
During their attempt of verification, they discovered and fixed a bug that may cause an array index
out of bounds error, and proved that no such error occurs in the fixed version.
While the bug fix has a significant impact, they had not proven the functional correctness
(sortedness and permutation properties) of the fixed version.

As a concluding remark, we stress that our formalization has a few features that none of the related
work above achieved:
1) a common interface (characterization) for stable mergesort functions, general induction principle
for mergesort replacing functional induction, sharing of functional correctness proofs between
several variations of mergesort,
2) formally verified tail-recursive mergesort
(\cref{sec:tailrec-mergesort,sec:tailrec-mergesort-in-coq}), including the smooth one
(\cref{appx:tailrec-mergesort-in-coq}), and
3) extensive theory of stable sort functions (\cref{appx:stablesort-theory}) which cannot be easily
derived from the results in the related work
(\cref{sec:sort_standard_stable,appx:stability-statements}).

\section{Conclusion}
\label{sec:conclusion}

We characterized mergesort functions for lists using their abstract versions and parametricity
(\cref{sec:characterization,sec:new-characterization}).
By abstracting out the type of sorted lists as a type parameter, we forced the abstract mergesort
functions to use the only provided operators (such as the order relation, merge, singleton, and
empty) to construct sorted lists, thus we ruled out behaviors incorrect as sorting functions, and
the parametricity of the abstract mergesort function ensures such correct behavior.
By instantiating the abstract mergesort functions in two ways, we should be able to obtain both
input and output of the mergesort function (\cref{eq:asort_mergeE',eq:asort_catE'}).
By exploiting the fact that parametricity implies induction principles on Church-encoded datatypes
(\cref{remark:induction}) and a case where parametricity implies naturality
(\cref{remark:naturality}), we deduced an induction principle
(\cref{lemma:sort_ind,lemma:sort_ind'}) over traces (\cref{fig:mergesort-traces}) from
\cref{eq:asort_mergeE',eq:asort_catE'}, and the naturality of mergesort (\cref{lemma:sort_map}),
respectively.
These two properties were sufficient to deduce several correctness results of mergesort, including
stability (\cref{sec:characterization-to-correctness,sec:new-characterization-to-correctness}).

In order to verify a given mergesort function using our technique, one just has to define its
abstract version, and then, prove its parametricity, and \cref{eq:asort_mergeE',eq:asort_catE'}.
Among these properties, parametricity follows from the abstraction
theorem~\cite{DBLP:conf/ifip/Reynolds83, DBLP:conf/fossacs/BernardyL11,
DBLP:journals/jfp/BernardyJP12, DBLP:conf/csl/KellerL12}, which is implemented in \Paramcoq~\cite{paramcoq}.
The rest of the correctness proofs work generically for any mergesort satisfying our
characteristic property.
Therefore, the actual work that has to be carried out by hand is to prove
\cref{eq:asort_mergeE',eq:asort_catE'} by induction and equational reasoning,
which justifies the title of this paper, claiming our functional correctness proofs are almost for
free and at a bargain.

\section*{Data-Availability Statement}
An artifact of this work is available for reproduction on
Zenodo\anon[ (anonymized citation)]{~\cite{stablesort:1.0.2-vm}}
and includes source code and dependencies.
Source code is further available for reuse through the Git repository at
\anon[anonymized url]{\url{https://github.com/pi8027/stablesort}}.

\begin{acks}
 The authors gratefully thank anonymous reviewers of ICFP '22, '24,
 and '25, Yves Bertot, Kazuhiro Inaba, and Yannick Zakowski for their comments.
 The general structure of our correctness proofs of mergesorts
 presented in \cref{sec:characterization-to-correctness}, except for
 the use of parametricity, is largely based on the correctness proofs
 of a non-smooth bottom-up non-tail-recursive mergesort
 (\coq{path.sort}) in the \MC library, to which the authors made
 significant contributions (in pull requests
 \href{https://github.com/math-comp/math-comp/pull/328}{\#328},
 \href{https://github.com/math-comp/math-comp/pull/358}{\#358},
 \href{https://github.com/math-comp/math-comp/pull/601}{\#601},
 \href{https://github.com/math-comp/math-comp/pull/650}{\#650},
 \href{https://github.com/math-comp/math-comp/pull/680}{\#680},
 \href{https://github.com/math-comp/math-comp/pull/727}{\#727},
 \href{https://github.com/math-comp/math-comp/pull/1174}{\#1174}, and
 \href{https://github.com/math-comp/math-comp/pull/1186}{\#1186}).
 The authors would like to thank other \MC developers and contributors
 who contributed to the discussion: Yves Bertot, Christian Doczkal,
 Georges Gonthier, Assia Mahboubi, and Anton Trunov.
\end{acks}

\nocite{rocqrefman}
\bibliography{bibliography}

\pagebreak

\appendix

\section{Basic definitions and facts used for proofs}
\label{appx:basic-definitions-and-facts}

This appendix provides a list of some definitions and lemmas in the \MC library~\cite{mathcomp} used
for the proofs in
\cref{sec:implementation-to-characterization,sec:characterization-to-correctness,%
sec:new-characterization-to-correctness,appx:stablesort-theory}.
Most informal definitions and lemmas are followed by the corresponding formal definitions and
statements in \Coq.
Some of these formal definitions and statements are modified to avoid introducing new definitions
but still convertible with the original definitions and statements.

We use the following \coq{Implicit Types} declaration to interpret the formal definitions and
statements in
\cref{appx:basic-definitions-and-facts,appx:stablesort-theory,appx:stability-statements}.
\begin{coqcode}
Implicit Types (sort : stableSort) (T R S : Type).
\end{coqcode}

\subsection{Predicates and relations}

\begin{definition}[Unary predicates and binary relations]
\label{def:coq:pred}
\label{def:coq:rel}
A unary \emph{predicate} $p \subseteq T$ and a binary \emph{relation} $R \subseteq T \times T$ on a
type $T$ are functions of types $T \to \mathrm{bool}$ and $T \to T \to \mathrm{bool}$, respectively:
\begin{coqcode}
Definition pred T : Type := T -> bool.
Definition rel T : Type := T -> pred T.
\end{coqcode}
We sometimes generalize relations to ones between two types $T$ and $S$; that is, a relation
$R \subseteq T \times S$ between $T$ and $S$ is a function of type $T \to S \to \mathrm{bool}$, \eg,
\cref{def:coq:allrel}.
We also abuse this terminology to mean any subset of $T$ and $T \times T$, that is not necessarily
decidable and respectively corresponds to any function of type \coq{T -> Prop} and
\coq{T -> T -> Prop} in \Coq.
\end{definition}

\begin{definition}[Totality of binary relations]
\label{def:coq:total}
A binary relation $R$ on type $T$ is \emph{total} if $x \mathrel{R} y$ or $y \mathrel{R} x$ holds
for any $x, y \in T$.
\begin{coqcode}
Definition total T (R : rel T) : Prop := forall x y : T, R x y || R y x.
\end{coqcode}
\end{definition}

\begin{definition}[Transitive relations]
\label{def:coq:transitive}
A binary relation $R$ on type $T$ is \emph{transitive} if $x \mathrel{R} y$ and $y \mathrel{R} z$
imply $x \mathrel{R} z$ for any $x, y, z \in T$.
\begin{coqcode}
Definition transitive T (R : rel T) : Prop :=
  forall y x z : T, R x y -> R y z -> R x z.
\end{coqcode}
\end{definition}

\begin{definition}[Antisymmetric relations]
\label{def:coq:antisymmetric}
A binary relation $R$ on type $T$ is \emph{antisymmetric} if $x \mathrel{R} y$ and $y \mathrel{R} x$
imply $x = y$ for any $x, y \in T$.
\begin{coqcode}
Definition antisymmetric T (R : rel T) : Prop :=
  forall x y : T, R x y && R y x -> x = y.
\end{coqcode}
\end{definition}

\begin{definition}[Reflexive relations]
\label{def:coq:reflexive}
A binary relation $R$ on type $T$ is \emph{reflexive} if $x \mathrel{R} x$ holds for any $x \in T$.
\begin{coqcode}
Definition reflexive T (R : rel T) : Prop := forall x : T, R x x.
\end{coqcode}
\end{definition}

\begin{definition}[Irreflexive relations]
\label{def:coq:irreflexive}
A binary relation $R$ on type $T$ is \emph{irreflexive} if $x \mathrel{R} x$ does not for any
$x \in T$.
\begin{coqcode}
Definition irreflexive T (R : rel T) : Prop := forall x : T, R x x = false.
\end{coqcode}
\end{definition}

\begin{definition}[Preorders and orders]
\label{def:coq:preorders}
A binary relation $\leq$ on type $T$ is said to be:
\begin{itemize}
 \item a \emph{total preorder} if $\leq$ is transitive and total,
 \item a \emph{strict preorder} if $\leq$ is transitive and irreflexive, and
 \item a \emph{total order} if $\leq$ is transitive, antisymmetric, and total.
\end{itemize}
\end{definition}

\begin{definition}[Lexicographic relations]
\label{def:coq:lexord}
Given two binary relations $\leq_1$ and $\leq_2$ on type $T$, their lexicographic relation
$\leq_{(1, 2)}$ is defined as follows:
\[
 x \leq_{(1, 2)} y \coloneq x \leq_1 y \land (y \not\leq_1 x \lor x \leq_2 y).
\]
We also write it as $\leq_\mathrm{lex}$ when there is only one lexicographic relation in the context
and thus there is no ambiguity.
\begin{coqcode}
Definition lexord T (leT leT' : rel T) :=
  [rel x y | leT x y && (leT y x ==> leT' x y)].
\end{coqcode}
\end{definition}

\begin{lemma}
\label{lemma:coq:lexord_total}
\label{lemma:coq:lexord_trans}
The lexicographic relation $\leq_{(1, 2)}$ is total (resp.~transitive) whenever both $\leq_1$ and
$\leq_2$ are total (resp.~transitive).
\begin{coqcode}
Lemma lexord_total T (leT leT' : rel T) :
  total leT -> total leT' -> total (lexord leT leT').

Lemma lexord_trans T (leT leT' : rel T) :
  transitive leT -> transitive leT' -> transitive (lexord leT leT').
\end{coqcode}
\end{lemma}

\begin{lemma}
\label{lemma:coq:lexordA}
The lexicographic composition of binary relations is associative; that is, the left associative
composition $\leq_{((1, 2), 3)}$ is the same relation as the right associative one
$\leq_{(1, (2, 3))}$.
\begin{coqcode}
Lemma lexordA T (leT leT' leT'' : rel T) :
  lexord leT (lexord leT' leT'') =2 lexord (lexord leT leT') leT''.
\end{coqcode}
\end{lemma}

\subsection{Lists and list surgery operators}

\begin{definition}[Lists]
\label{def:coq:list}
A list of type $\mathrm{list} \, T$ is a finite sequence of values of type $T$; that is, the set of
lists is defined as the least fixed point of the following rules:
\begin{itemize}
 \item the empty sequence $[]$ (nil) is a list of type $\mathrm{list} \, T$ for any type $T$, and
 \item a cons cell $x :: s$ is a list of type $\mathrm{list} \, T$ if $x$ and $s$ have type $T$ and
       $\mathrm{list} \, T$, respectively.
\end{itemize}
\begin{coqcode}
Inductive list (A : Type) : Type := nil : list A | cons : A -> list A -> list A.
\end{coqcode}
We write list literals of the form $(x_1 :: x_2 :: \dots :: x_n :: [])$ as $[x_1, x_2, \dots, x_n]$,
or \coq{[:: x1; x2; ...; xn]} in \Coq.
\end{definition}

\begin{definition}[Concatenation]
\label{def:coq:cat}
Given two lists $s_1 \coloneq [x_1, \dots, x_n]$ and $s_2 \coloneq [y_1, \dots, y_m]$,
$\texttt{cat} \, s_1 \, s_2$, also written as $s_1 \concat s_2$, concatenates $s_1$ and $s_2$, \ie,
$[x_1, \dots, x_n, y_1, \dots, y_m]$.
This is a \Coq equivalent of \caml{List.append} in \OCaml.
\begin{coqcode}
Definition cat T : list T -> list T -> list T :=
  fix cat (s1 s2 : list T) {struct s1} : list T :=
    if s1 is x :: s1' then x :: cat s1' s2 else s2.

Notation "s1 ++ s2" := (cat s1 s2).
\end{coqcode}
\end{definition}

\begin{lemma}
\label{lemma:coq:catA}
The concatenation operator $\concat$ is associative.
\begin{coqcode}
Lemma catA T (x y z : list T) : x ++ (y ++ z) = (x ++ y) ++ z.
\end{coqcode}
\end{lemma}

\begin{definition}[Right fold]
\label{def:coq:foldr}
Given a function $f$, an initial value $z$, and a list $s \coloneq [x_1, x_2, \dots, x_n]$,
$\texttt{foldr} \, f \, z \, s$ is the right fold of $s$ with $f$ and $z$, \ie,
$f \, x_1 \, (f \, x_2 \, (\dots(f \, x_n \, z)\dots))$.
\begin{coqcode}
Definition foldr T R (f : T -> R -> R) (z0 : R) : list T -> R :=
  fix foldr (s : list T) {struct s} : R :=
    if s is x :: s' then f x (foldr s') else z0.
\end{coqcode}
\end{definition}

\begin{definition}[Flattening]
\label{def:coq:flatten}
Given a list of lists $\mathit{ss} = [s_1, s_2, \dots, s_n]$, $\texttt{flatten} \, \mathit{ss}$ is
the list obtained by folding $\mathit{ss}$ with concatenation, \ie,
$s_1 \concat s_2 \concat \dots \concat s_n$.
This is a \Coq equivalent of \caml{List.flatten} in \OCaml.
\begin{coqcode}
Definition flatten T : list (list T) -> list T := foldr (@cat T) [::].
\end{coqcode}
\end{definition}

\begin{definition}
\label{def:coq:catrev}
Given two lists $s_1 \coloneq [x_1, \dots, x_n]$ and $s_2 \coloneq [y_1, \dots, y_m]$,
$\texttt{catrev} \, s_1 \, s_2$ reverses $s_1$ and concatenates it with $s_2$, \ie,
$[x_n, \dots, x_1, y_1, \dots, y_m]$.
This is a \Coq equivalent of \caml{List.rev_append} in \OCaml.
\begin{coqcode}
Definition catrev T : list T -> list T -> list T :=
  fix catrev (s1 s2 : list T) {struct s1} : list T :=
    if s1 is x :: s1' then catrev s1' (x :: s2) else s2.
\end{coqcode}
\end{definition}

\begin{definition}[List reversal]
\label{def:coq:rev}
Given a list $s = [x_1, x_2, \dots, x_n]$, $\texttt{rev} \, s$ is the reversal of $s$, \ie,
$[x_n, \dots, x_2, x_1]$. This is a \Coq equivalent of \caml{List.rev} in \OCaml.
\begin{coqcode}
Definition rev T (s : list T) : list T := catrev s [::].
\end{coqcode}
\end{definition}

\begin{lemma}
\label{lemma:coq:catrevE}
For any lists $s$ and $t$, $\texttt{catrev} \, s \, t = \texttt{rev} \, s \concat t$ holds.
\begin{coqcode}
Lemma catrevE T (s t : list T) : catrev s t = rev s ++ t.
\end{coqcode}
\end{lemma}

\begin{lemma}
\label{lemma:coq:revK}
\texttt{rev} is involutive; that is, $\texttt{rev} \, (\texttt{rev} \, s) = s$ holds for any
list $s$.
\begin{coqcode}
Lemma revK T (s : list T) : rev (rev s) = s.
\end{coqcode}
\end{lemma}

\begin{lemma}
\label{lemma:coq:rev_cat}
For any lists $s$ and $t$,
$\texttt{rev} \, (s \concat t) = \texttt{rev} \, t \concat \texttt{rev} \, s$ holds.
\begin{coqcode}
Lemma rev_cat T (s t : list T) : rev (s ++ t) = rev t ++ rev s.
\end{coqcode}
\end{lemma}

\begin{definition}
\label{def:coq:take}
\label{def:coq:drop}
Given a natural number $n$ and a list $s$, $\texttt{take} \, n \, s$ and $\texttt{drop} \, n \, s$ are
respectively
\begin{itemize}
 \item the list collecting the first $n$ elements of $s$, or $s$ if the length of $s$ is less than $n$,
       and
 \item the list obtained by dropping the first $n$ elements of $s$, or $[]$ if the length of $s$ is
       less than $n$.
\end{itemize}
Their pair $(\texttt{take} \, n \, s, \texttt{drop} \, n \, s)$ is a \Coq equivalent of
\caml{List.split_n} in \OCaml.
\begin{coqcode}
Definition take T : nat -> list T -> list T :=
  fix take (n : nat) (s : list T) {struct s} : list T :=
    match s, n with
    | [::], _ | _, 0 => nil
    | x :: s', S n' => x :: take n' s'
    end.

Definition drop T : nat -> list T -> list T :=
  fix drop (n : nat) (s : list T) {struct s} : list T :=
    match s, n with
    | [::], _ | _ :: _, 0 => s
    | _ :: s', S n' => drop n' s'
    end.
\end{coqcode}
\end{definition}

\begin{lemma}
\label{lemma:coq:cat_take_drop}
For any natural number $n$ and list $s$,
$\texttt{take} \, n \, s \concat \texttt{drop} \, n \, s = s$ holds.
\begin{coqcode}
Lemma cat_take_drop (n : nat) T (s : list T) : take n s ++ drop n s = s.
\end{coqcode}
\end{lemma}

\subsection{Counting and size}

\begin{definition}
\label{def:coq:count}
Given a predicate $p \subseteq T$ and a list $s$ of type $\mathrm{list} \, T$, $\lvert s \rvert_a$
is the number of elements in $s$ satisfying $p$.
\begin{coqcode}
Definition count T (p : pred T) : list T -> nat :=
  foldr (fun x n => p x + n) 0.
\end{coqcode}
\end{definition}

\begin{definition}
\label{def:coq:size}
Given a list $s$, $\lvert s \rvert$ is the length of $s$.
This is a \Coq equivalent of \caml{List.length} in \OCaml.
\begin{coqcode}
Definition size T : list T -> nat := foldr (fun _ n => S n) 0.
\end{coqcode}
\end{definition}

\begin{lemma}
\label{lemma:coq:count_predT}
For any $s$ of type $\mathrm{list} \, T$, $\lvert s \rvert_T = \lvert s \rvert$ holds.
\begin{coqcode}
Lemma count_predT T (s : list T) : count (fun _ => true) s = size s.
\end{coqcode}
\end{lemma}

\begin{lemma}
\label{lemma:coq:count_cat}
For any $p \subseteq T$ and $s_1$ and $s_2$ of type $\mathrm{list} \, T$,
$\lvert s_1 \concat s_2 \rvert_p = \lvert s_1 \rvert_p + \lvert s_2 \rvert_p$ holds.
\begin{coqcode}
Lemma count_cat T (p : pred T) (s1 s2 : list T) :
  count p (s1 ++ s2) = count p s1 + count p s2.
\end{coqcode}
\end{lemma}

\begin{lemma}
\label{lemma:coq:count_rev}
For any $p \subseteq T$ and $s$ of type $\mathrm{list} \, T$,
$\lvert \texttt{rev} \, s \rvert_p = \lvert s \rvert_p$ holds.
\begin{coqcode}
Lemma count_rev T (p : pred T) (s : list T) : count a (rev s) = count a s.
\end{coqcode}
\end{lemma}

\subsection{Predicates on lists}

\begin{definition}
\label{def:coq:all}
Given a predicate $p \subseteq T$ and a list $s$ of type $\mathrm{list} \, T$,
$\texttt{all}_p \, s$ holds if all elements of $s$ satisfy $p$.
\begin{coqcode}
Definition all T (p : pred T) : list T -> bool :=
  foldr (fun x b => p x && b) true.
\end{coqcode}
\end{definition}

\begin{lemma}
\label{lemma:coq:all_count}
For any $p \subseteq T$ and $s$ of type $\mathrm{list} \, T$, $\texttt{all}_p \, s$ holds iff
$\lvert s \rvert_p = \lvert s \rvert$.
\begin{coqcode}
Lemma all_count T (p : pred T) (s : list T) : all p s = (count p s == size s).
\end{coqcode}
\end{lemma}




\begin{lemma}
\label{lemma:coq:sub_all}
For any $p_1 \subseteq p_2 \subseteq T$ and $s$ of type $\mathrm{list} \, T$,
$\texttt{all}_{p_2} \, s$ holds whenever $\texttt{all}_{p_1} \, s$ holds.
\begin{coqcode}
Lemma sub_all T (p1 p2 : pred T) :
  (forall x : T, p1 x -> p2 x) -> forall s : list T, all p1 s -> all p2 s.
\end{coqcode}
\end{lemma}


\begin{definition}
\label{def:coq:allrel}
Given a relation $R \subseteq T \times S$ and two lists $s_1$ and $s_2$ of respectively types
$\mathrm{list} \, T$ and $\mathrm{list} \, S$, $\texttt{allrel}_R \, s_1 \, s_2$ holds if any
elements $x$ of $s_1$ and $y$ of $s_2$ satisfy $x \mathrel{R} y$.
\begin{coqcode}
Definition allrel T S (r : T -> S -> bool) (xs : list T) (ys : list S) : bool :=
  all (fun x => all (r x) ys) xs.
\end{coqcode}
\end{definition}

\begin{lemma}
\label{lemma:coq:allrelC}
For any $R \subseteq T \times S$ and $s_1$ and $s_2$ of respectively types $\mathrm{list} \, T$ and
$\mathrm{list} \, S$, $\texttt{allrel}_R \, s_1 \, s_2$ holds iff
$\texttt{allrel}_{R^\mathrm{C}} \, s_2 \, s_1$ holds, where $R^\mathrm{C} \subseteq S \times T$ is
the converse relation of $R$, \ie, $y \mathrel{R^\mathrm{C}} x \coloneq x \mathrel{R} y$.
\begin{coqcode}
Lemma allrelC T S (r : T -> S -> bool) (xs : list T) (ys : list S) :
  allrel r xs ys = allrel (fun x y => r y x) ys xs.
\end{coqcode}
\end{lemma}

\begin{lemma}
\label{lemma:coq:allrel_rev2}
For any
relation 
$R \subseteq T \times S$ and $s_1$ and $s_2$ of respectively types $\mathrm{list} \, T$ and
$\mathrm{list} \, S$, $\texttt{allrel}_R \, (\texttt{rev} \, s_1) \, (\texttt{rev} \, s_2)$ holds
iff $\texttt{allrel}_R \, s_1 \, s_2$ holds.
\begin{coqcode}
Lemma allrel_rev2 T S (r : T -> S -> bool) (s1 : list T) (s2 : list S) :
  allrel r (rev s1) (rev s2) = allrel r s1 s2.
\end{coqcode}
\end{lemma}

\begin{definition}
\label{def:coq:uniq}
Given a list $s$, $\texttt{uniq} \, s$ holds if all the elements of $s$ are pairwise different, \ie,
$s$ is duplication free.
\begin{coqcode}
Definition uniq (T : eqType) : list T -> bool :=
  fix uniq (s : list T) {struct s} : bool :=
    if s is x :: s' then (x \notin s') && uniq s' else true.
\end{coqcode}
where \coq{x \notin s'} means that \coq{x} is not an element of \coq{s'}, which requires \coq{T} to
be an \coq{eqType}, a type with a decidable equality.
\end{definition}

\subsection{Map and filter}

\begin{definition}[Map]
\label{def:coq:map}
Given a function $f$ of type $T_1 \to T_2$ and a list $s := [x_1, x_2, \dots, x_n]$ of type
$\mathrm{list} \, T_1$, $\texttt{map}_f \, s$ is the list of type $\mathrm{list} \, T_2$ mapping $f$ to
the elements of $s$, \ie, $[f \, x_1, f \, x_2, \dots, f \, x_n]$.
This is a \Coq equivalent of \caml{List.map} in \OCaml.
\begin{coqcode}
Definition map T1 T2 (f : T1 -> T2) : list T1 -> list T2 :=
  foldr (fun x s => f x :: s) [::].
\end{coqcode}
\end{definition}

\begin{definition}[Filter]
\label{def:coq:filter}
Given a predicate $p \subseteq T$ and a list $s$ of type $\mathrm{list} \, T$,
$\texttt{filter}_p \, s$ is the list collecting all the elements of $s$ that satisfy $p$, and
preserves the order of the elements in the input.
This is a \Coq equivalent of \caml{List.filter} in \OCaml.
\begin{coqcode}
Definition filter T (p : pred T) : list T -> list T :=
  foldr (fun x s => if p x then x :: s else s) [::].
\end{coqcode}
\end{definition}


\begin{lemma}[The naturality of \texttt{filter}]
\label{lemma:coq:filter_map}
For any $f$ of type $T_1 \to T_2$, $p \subseteq T_2$, and $s$ of type $\mathrm{list} \, T_1$,
$\texttt{filter}_p \, (\texttt{map}_f \, s) = \texttt{map}_f \, (\texttt{filter}_{p \circ f} \, s)$
holds.
\begin{coqcode}
Lemma filter_map T1 T2 (f : T1 -> T2) (p : pred T2) (s : list T1) :
  filter p (map f s) = map f (filter (fun x => p (f x)) s).
\end{coqcode}
\end{lemma}



\begin{lemma}
\label{lemma:coq:mem_filter}
For any $p \subseteq T$, $x$ of type $T$, and $s$ of type $\mathrm{list} \, T$, $x$ is an element of
$\texttt{filter}_p \, s$ iff $x$ is an element of $s$ that satisfies $p$.
\begin{coqcode}
Lemma mem_filter (T : eqType) (p : pred T) (x : T) (s : list T) :
  (x \in filter p s) = p x && (x \in s).
\end{coqcode}
\end{lemma}

\subsection{Subsequences}

\begin{definition}[Mask]
\label{def:coq:mask}
Given two lists $m$ and $s$ of respectively types $\mathrm{list} \, \mathrm{bool}$ and
$\mathrm{list} \, T$, $\texttt{mask}_m \, s$ is the subsequence of $s$ selected by $m$; that is, the
$i^\mathrm{th}$ element of $s$ is selected if the $i^\mathrm{th}$ element of $m$ is
\texttt{true}.
\begin{coqcode}
Definition mask T : list bool -> list T -> list T :=
  fix mask (m : list bool) (s : list T) {struct m} : list T :=
    match m, s with
    | [::], _ | _, [::] => [::]
    | b :: m', x :: s' => if b then x :: mask m' s' else mask m' s'
    end.
\end{coqcode}
\end{definition}


\begin{lemma}[The naturality of \texttt{mask}]
\label{lemma:coq:map_mask}
For any $f$ of type $T_1 \to T_2$, and $m$ and $s$ of respectively types
$\mathrm{list} \, \mathrm{bool}$ and $\mathrm{list} \, T_1$,
$\texttt{map}_f \, (\texttt{mask}_m \, s) = \texttt{mask}_m \, (\texttt{map}_f \, s)$ holds.
\begin{coqcode}
Lemma map_mask T1 T2 (f : T1 -> T2) (m : list bool) (s : list T1) :
  map f (mask m s) = mask m (map f s).
\end{coqcode}
\end{lemma}

\begin{lemma}
\label{lemma:coq:filter_mask}
For any $p \subseteq T$ and $s$ of type $\mathrm{list} \, T$,
$\texttt{filter}_p \, s = \texttt{mask}_{m'} \, s$ where $m' \coloneq \texttt{map}_p \, s$ holds.
\begin{coqcode}
Lemma filter_mask T (p : pred T) (s : list T) : filter p s = mask (map p s) s.
\end{coqcode}
\end{lemma}

\begin{lemma}
\label{lemma:coq:mask_filter}
For any $s$ and $m$ of respectively types $\mathrm{list} \, T$ and $\mathrm{list} \, \mathrm{bool}$,
$\texttt{mask}_m \, s = \texttt{filter}_p \, s$ where $p \, x \coloneq x \in \texttt{mask}_m \, s$
holds whenever $s$ is duplication free.
\begin{coqcode}
Lemma mask_filter (T : eqType) (s : list T) (m : list bool) :
  uniq s -> mask m s = filter [in mask m s] s.
\end{coqcode}
\end{lemma}

\begin{definition}[Subsequence relation]
\label{def:coq:subseq}
Given two lists $s_1$ and $s_2$, $\texttt{subseq} \, s_1 \, s_2$ means that $s_1$ is a subsequence of
$s_2$.
\begin{coqcode}
Definition subseq (T : eqType) : list T -> list T -> bool :=
  fix subseq (s1 s2 : list T) {struct s2} : bool :=
    match s2, s1 with
      | _, [::] => true
      | [::], _ :: _ => false
      | y :: s2', x :: s1' => subseq (if x == y then s1' else s1) s2'
    end.
\end{coqcode}
\end{definition}

\begin{lemma}
\label{lemma:coq:subseqP}
For any lists $s_1$ and $s_2$, $s_1$ is a subsequence of $s_2$ iff there exists a list $m$ of type
$\mathrm{list} \, \mathrm{bool}$ such that $\lvert m \rvert = \lvert s_2 \rvert$ and
$s_1 = \texttt{mask}_m \, s_2$.
\begin{coqcode}
Lemma subseqP (T : eqType) (s1 s2 : list T) :
  reflect
    (exists2 m : list bool, size m = size s2 & s1 = mask m s2) (subseq s1 s2).
\end{coqcode}
\end{lemma}

\begin{lemma}
\label{lemma:coq:mask_subseq}
For any $m$ and $s$ of respectively types $\mathrm{list} \, \mathrm{bool}$ and $\mathrm{list} \, T$,
$\texttt{mask}_m \, s$ is a subsequence of $s$.
\begin{coqcode}
Lemma mask_subseq (T : eqType) (m : list bool) (s : list T) :
  subseq (mask m s) s.
\end{coqcode}
\end{lemma}

\begin{definition}
\label{def:coq:index}
For any $x$ of type $T$ and $s$ of type $\mathrm{list} \, T$, $\texttt{index}_x \, s$ is the index
of the first occurrence of $x$ in $s$.
\begin{coqcode}
Definition index (T : eqType) (x : T) : list T -> nat :=
  fix find (s : list T) {struct s} : nat :=
    if s is y :: s' then
      if y == x then 0 else S (find s')
    else
      0.
\end{coqcode}
\end{definition}

\begin{definition}
\label{def:coq:mem2}
For any $s$ of type $\mathrm{list} \, T$, and $x$ and $y$ of type $T$, we say that $x$ and $y$
occur in order in $s$, or write $\texttt{mem2} \, s \, x \, y$, when $y$ occurs in $s$
(non-strictly) after the first occurrence of $x$.
Here, ``non-strictly'' means that $\texttt{mem2} \, s \, x \ x$ holds even if $x$ occurs in $s$ only
once.
\begin{coqcode}
Definition mem2 (T : eqType) (s : list T) (x y : T) : bool :=
  y \in drop (index x s) s.
\end{coqcode}
\end{definition}

\begin{lemma}
\label{lemma:coq:mem2E}
For any $s$ of type $\mathrm{list} \, T$, and $x$ and $y$ of type $T$,
$x$ and $y$ occur in order in $s$ iff the following $\mathit{xy}$ is a subsequence of $s$:
\[
 \mathit{xy} \coloneq
 \begin{cases}
  [x]     & (x = y) \\
  [x, y]. & (\text{otherwise})
 \end{cases}
\]
\begin{coqcode}
Lemma mem2E (T : eqType) (s : list T) (x y : T) :
  mem2 s x y = subseq (if x == y then [:: x] else [:: x; y]) s.
\end{coqcode}
\end{lemma}

\subsection{Permutation}

\begin{definition}[Permutation relation]
\label{def:coq:perm_eq}
\label{def:coq:perm_eql}
Two lists $s_1$ and $s_2$ are \emph{permutation} of each other iff
$\lvert s_1 \rvert_{\{x\}} = \lvert s_2 \rvert_{\{x\}}$ for any element $x$ of $s_1$ or $s_2$, then
we write $s_1 \permeq s_2$.
\begin{coqcode}
Definition perm_eq (T : eqType) (s1 s2 : list T) : bool :=
  all [pred x | count (pred1 x) s1 == count (pred1 x) s2] (s1 ++ s2).
\end{coqcode}
In order to use the fact $s_1 \permeq s_2$ to rewrite a goal of the form $s_1 \permeq s_3$ to
$s_2 \permeq s_3$, some lemmas, \eg, \cref{lemma:coq:perm_catC,lemma:coq:perm_rev}, are stated in
the form of $\forall s_3, (s_1 \permeq s_3) = (s_2 \permeq s_3)$, using the following notation (see
also \cref{lemma:coq:permPl}).
\begin{coqcode}
Notation perm_eql s1 s2 := (perm_eq s1 =1 perm_eq s2).
\end{coqcode}
\end{definition}

\begin{lemma}
\label{lemma:coq:permP}
For any lists $s_1$ and $s_2$ of type $\mathrm{list} \, T$, $s_1 \permeq s_2$ holds iff
$\lvert s_1 \rvert_p = \lvert s_2 \rvert_p$ holds for any predicate $p \subseteq T$.
\begin{coqcode}
Lemma permP (T : eqType) (s1 s2 : list T) :
  reflect (forall p : pred T, count p s1 = count p s2) (perm_eq s1 s2).
\end{coqcode}
\end{lemma}

\begin{lemma}
\label{lemma:coq:permPl}
For any lists $s_1$ and $s_2$, $s_1 \permeq s_2$ holds iff
$(s_1 \permeq s_3) = (s_2 \permeq s_3)$ holds for any list $s_3$.
\begin{coqcode}
Lemma permPl (T : eqType) (s1 s2 : list T) :
  reflect (perm_eql s1 s2) (perm_eq s1 s2).
\end{coqcode}
\end{lemma}

\begin{lemma}
\label{lemma:coq:perm_mem}
For any lists $s_1$ and $s_2$ of type $\mathrm{list} \, T$ such that $s_1 \permeq s_2$ and $x$ of
type $T$, $x \in s_1$ iff $x \in s_2$.
\begin{coqcode}
Lemma perm_mem (T : eqType) (s1 s2 : list T) : perm_eq s1 s2 -> s1 =i s2.
\end{coqcode}
where \coq{s1 =i s2 := (forall x, x \in s1 = x \in s2)} means that \coq{s1} and \coq{s2} have the
same set of elements.
\end{lemma}

\begin{lemma}
\label{lemma:coq:perm_uniq}
For any lists $s_1$ and $s_2$ such that $s_1 \permeq s_2$, $s_1$ is duplication free iff $s_2$ is
duplication free.
\begin{coqcode}
Lemma perm_uniq (T : eqType) (s1 s2 : list T) :
  perm_eq s1 s2 -> uniq s1 = uniq s2.
\end{coqcode}
\end{lemma}

\begin{lemma}
\label{lemma:coq:perm_cat}
The permutation relation $\permeq$ is a congruence relation with respect to concatenation $\concat$;
that is, $s_1 \permeq s_2$ and $t_1 \permeq t_2$ imply $s_1 \concat t_1 \permeq s_2 \concat t_2$
for any lists $s_1$, $s_2$, $t_1$, and $t_2$.
\begin{coqcode}
Lemma perm_cat (T : eqType) (s1 s2 t1 t2 : list T) :
  perm_eq s1 s2 -> perm_eq t1 t2 -> perm_eq (s1 ++ t1) (s2 ++ t2).
\end{coqcode}
\end{lemma}

\begin{lemma}
\label{lemma:coq:perm_catC}
For any lists $s_1$ and $s_2$, $s_1 \concat s_2$ is a permutation of $s_2 \concat s_1$.
\begin{coqcode}
Lemma perm_catC (T : eqType) (s1 s2 : list T) : perm_eql (s1 ++ s2) (s2 ++ s1).
\end{coqcode}
\end{lemma}

\begin{lemma}
\label{lemma:coq:perm_rev}
For any list $s$, $\texttt{rev} \, s$ is a permutation of $s$.
\begin{coqcode}
Lemma perm_rev (T : eqType) (s : list T) : perm_eql (rev s) s.
\end{coqcode}
\end{lemma}

\subsection{Sortedness}

\begin{definition}[Pairwise sortedness, \cref{def:sortedness}]
\label{definition:coq:pairwise}
Given a relation $R \subseteq T \times T$, a list $s \coloneq [x_0, \dots, x_n]$ of type
$\mathrm{list} \, T$ is said to be \emph{pairwise sorted} \wrt $R$ if the relation $R$ holds any $x_i$
and $x_j$ such that $i < j \leq n$, \ie,
\[
 x_0 \mathrel{R} x_1 \land \dots \land x_0 \mathrel{R} x_n \land
 x_1 \mathrel{R} x_2 \land \dots \land x_1 \mathrel{R} x_n \land \dots \land
 x_{n - 1} \mathrel{R} x_n.
\]
\begin{coqcode}
Definition pairwise T (r : T -> T -> bool) : list T -> bool :=
  fix pairwise (xs : list T) {struct xs} : bool :=
    if xs is x :: xs0 then all (r x) xs0 && pairwise xs0 else true.
\end{coqcode}
\end{definition}

\begin{lemma}
\label{lemma:coq:pairwise_cat}
For any relation $R \subseteq T \times T$ and lists $s_1$ and $s_2$ of type $\mathrm{list} \, T$,
$s_1 \concat s_2$ is pairwise sorted \wrt $R$ iff the following conjunction holds:
\begin{itemize}
 \item $x \mathrel{R} y$ holds for any $x \in s_1$ and $y \in s_2$, and
 \item both $s_1$ and $s_2$ are pairwise sorted \wrt $R$.
\end{itemize}
\begin{coqcode}
Lemma pairwise_cat T (r : T -> T -> bool) (s1 s2 : list T) :
  pairwise r (s1 ++ s2) = allrel r s1 s2 && (pairwise r s1 && pairwise r s2).
\end{coqcode}
\end{lemma}

\begin{definition}[Path and sortedness, \cref{def:sortedness}]
\label{definition:coq:path}
\label{definition:coq:sorted}
Given a relation $R \subseteq T \times T$, a list $s \coloneq [x_0, \dots, x_n]$ of type
$\mathrm{list} \, T$ is said to be \emph{sorted} if $R$ holds for each adjacent pair in $s$, \ie,
$x_0 \mathrel{R} x_1 \land \dots \land x_{n - 1} \mathrel{R} x_n$.
A cons cell $x :: s$ is said to be an $R$-path if it is sorted \wrt $R$.

In \Coq, the former is defined using the latter:
\begin{coqcode}
Definition path T (e : rel T) : T -> list T -> bool :=
  fix path (x : T) (p : list T) {struct p} : bool :=
    if p is y :: p' then e x y && path y p' else true.

Definition sorted T (e : rel T) (s : list T) : bool :=
  if s is x :: s' then path e x s' else true.
\end{coqcode}
\end{definition}


\begin{lemma}
\label{lemma:coq:sorted_pairwise}
The sortedness and the pairwise sortedness are equivalent for transitive relations.
\begin{coqcode}
Lemma sorted_pairwise T (r : rel T) :
  transitive r -> forall s : list T, sorted r s = pairwise r s.
\end{coqcode}
\end{lemma}


\begin{lemma}
\label{lemma:coq:rev_sorted}
For any relation $R \subseteq T \times T$ and list $s$ of type $\mathrm{list} \, T$,
$\texttt{rev} \, s$ is sorted \wrt $R$ iff $s$ is sorted \wrt its converse relation $R^\mathrm{C}$.
\begin{coqcode}
Lemma rev_sorted T (r : rel T) (s : list T) :
  sorted r (rev s) = sorted (fun x y : T => r y x) s.
\end{coqcode}
\end{lemma}



\begin{lemma}
\label{lemma:coq:sorted_filter}
For any transitive relation $R \subseteq T \times T$, predicate $p \subseteq T$, and list $s$ of
type $\mathrm{list} \, T$, $\texttt{filter}_p \, s$ is sorted \wrt $R$ whenever $s$ is sorted \wrt $R$.
\begin{coqcode}
Lemma sorted_filter T (r : rel T) : transitive r ->
  forall (p : pred T) (s : list T), sorted r s -> sorted r (filter p s).
\end{coqcode}
\end{lemma}

\begin{lemma}
\label{lemma:coq:sorted_eq}
For any transitive and antisymmetric relation ${\leq} \subseteq T \times T$ and lists $s_1$ and
$s_2$ of type $\mathrm{list} \, T$, $s_1 = s_2$ holds whenever $s_1 \permeq s_2$ and both $s_1$ and
$s_2$ are sorted \wrt $\leq$.
\begin{coqcode}
Lemma sorted_eq (T : eqType) (leT : rel T) :
  transitive leT -> antisymmetric leT ->
  forall s1 s2 : list T,
    sorted leT s1 -> sorted leT s2 -> perm_eq s1 s2 -> s1 = s2.
\end{coqcode}
\end{lemma}

\begin{lemma}
\label{lemma:coq:irr_sorted_eq}
For any strict preorder ${\leq} \subseteq T \times T$ and lists $s_1$ and $s_2$ of type
$\mathrm{list} \, T$, $s_1 = s_2$ holds whenever $s_1$ and $s_2$ contain the same set of elements and
both of them are sorted \wrt $\leq$.
\begin{coqcode}
Lemma irr_sorted_eq (T : eqType) (leT : rel T) :
  transitive leT -> irreflexive leT ->
  forall s1 s2 : list T, sorted leT s1 -> sorted leT s2 -> s1 =i s2 -> s1 = s2.
\end{coqcode}
\end{lemma}

\subsection{Indexing}

\begin{definition}
\label{def:coq:iota}
Given two natural numbers $m$ and $n$, $\texttt{iota} \, m \, n$ is the list of natural numbers
$[m, m + 1, \dots, m + n - 1]$.
\begin{coqcode}
Fixpoint iota (m n : nat) {struct n} : list nat :=
  if n is S n' then m :: iota (S m) n' else [::].
\end{coqcode}
\end{definition}

\begin{definition}
\label{def:coq:nth}
Given $x_0$ of type $T$, a list $s$ of type $\mathrm{list} \, T$, a natural number $n$,
$\texttt{nth} \, x_0 \, s \, n$ is the $i^\mathrm{th}$ element of $s$, except that it is $x_0$ when
$n \leq \lvert s \rvert$.
\begin{coqcode}
Definition nth T (x0 : T) : list T -> nat -> T :=
  fix nth (s : list T) (n : nat) {struct n} : T :=
    match s, n with
    | [::], _ => x0
    | x :: _, 0 => x
    | _ :: s', S n' => nth s' n'
    end.
\end{coqcode}
\end{definition}

\begin{lemma}
\label{lemma:coq:mkseq_nth}
For any $x_0$ of type $T$ and $s$ of type $\mathrm{list} \, T$, $s$ is equal to
$[\texttt{nth} \, x_0 \, s \, i \mid i \leftarrow \texttt{iota} \, 0 \, \lvert s \rvert]$.
\begin{coqcode}
Lemma mkseq_nth T (x0 : T) (s : list T) : map (nth x0 s) (iota 0 (size s)) = s.
\end{coqcode}
\end{lemma}

\begin{lemma}
\label{lemma:coq:iota_uniq}
Any \texttt{iota} sequence is duplication free.
\begin{coqcode}
Lemma iota_uniq (m n : nat) : uniq (iota m n).
\end{coqcode}
\end{lemma}

\begin{lemma}
\label{lemma:coq:iota_ltn_sorted}
Any \texttt{iota} sequence is sorted \wrt the strict less than relation of natural numbers.
\begin{coqcode}
Lemma iota_ltn_sorted (i n : nat) : sorted ltn (iota i n).
\end{coqcode}
\end{lemma}

\subsection{Merge}

\begin{definition}[Merge]
\label{def:coq:merge}
Given a relation ${\leq} \subseteq T \times T$ and two lists $s_1$ and $s_2$ of type
$\mathrm{list} \, T$, their merge $s_1 \merge_\leq s_2$ is a list of type $\mathrm{list} \, T$
defined as follows:
\begin{align*}
 [] \merge_\leq s_2 &\coloneq s_2 \\
 s_1 \merge_\leq [] &\coloneq s_1 \\
 (x :: s_1) \merge_\leq (y :: s_2) &\coloneq
 \begin{cases}
  x :: (s_1 \merge_\leq (y :: s_2)) & (x \leq y) \\
  y :: ((x :: s_1) \merge_\leq s_2). & (\text{otherwise})
 \end{cases}
\end{align*}
\begin{coqcode}
Definition merge T (leT : rel T) :=
  fix merge (s1 : list T) {struct s1} : list T -> list T :=
    match s1 with
    | [::] => id
    | x1 :: s1' =>
      fix merge_s1 (s2 : list T) {struct s2} : list T :=
        match s2 with
        | [::] => s1
        | x2 :: s2' =>
          if leT x1 x2 then x1 :: merge s1' s2 else x2 :: merge_s1 s2'
        end
    end.
\end{coqcode}
\end{definition}

\begin{lemma}
\label{lemma:coq:count_merge}
For any ${\leq} \subseteq T \times T$, $p \subseteq T$, and $s_1$ and $s_2$ of type
$\mathrm{list} \, T$, $\lvert s_1 \merge_\leq s_2 \rvert_p$ is equal to
$\lvert s_1 \concat s_2 \rvert_p$.
\begin{coqcode}
Lemma count_merge T (leT : rel T) (p : pred T) (s1 s2 : list T) :
  count p (merge leT s1 s2) = count p (s1 ++ s2).
\end{coqcode}
\end{lemma}


\begin{lemma}
\label{lemma:coq:perm_merge}
For any ${\leq} \subseteq T \times T$, and $s_1$ and $s_2$ of type $\mathrm{list} \, T$,
$s_1 \merge_\leq s_2$ is a permutation of $s_1 \concat s_2$.
\begin{coqcode}
Lemma perm_merge (T : eqType) (r : rel T) (s1 s2 : list T) :
  perm_eql (merge r s1 s2) (s1 ++ s2).
\end{coqcode}
\end{lemma}

\begin{lemma}
\label{lemma:coq:allrel_merge}
For any ${\leq} \subseteq T \times T$, and $s_1$ and $s_2$ of type $\mathrm{list} \, T$,
$s_1 \merge_\leq s_2$ is equal to $s_1 \concat s_2$ whenever $x \leq y$ holds for any $x \in s_1$
and $y \in s_2$.
\begin{coqcode}
Lemma allrel_merge T (leT : rel T) (s1 s2 : list T) :
  allrel leT s1 s2 -> merge leT s1 s2 = s1 ++ s2.
\end{coqcode}
\end{lemma}

\begin{lemma}
\label{lemma:coq:merge_stable_sorted}
For any ${\leq_1}, {\leq_2} \subseteq T \times T$, and $s_1$ and $s_2$ of type $\mathrm{list} \, T$,
$s_1 \merge_{\leq_1} s_2$ is sorted \wrt the lexicographic order
${\leq_\mathrm{lex}} \coloneq {\leq_{(1, 2)}}$ whenever
$\leq_1$ is total,
$x \leq_2 y$ holds for any $x \in s_1$ and $y \in s_2$,
and both $s_1$ and $s_2$ are sorted \wrt $\leq_\mathrm{lex}$.
\begin{coqcode}
Lemma merge_stable_sorted T (r r' : rel T) :
  total r -> forall s1 s2 : list T,
  allrel r' s1 s2 -> sorted (lexord r r') s1 -> sorted (lexord r r') s2 ->
  sorted (lexord r r') (merge r s1 s2).
\end{coqcode}
\end{lemma}

\begin{lemma}
\label{lemma:coq:mergeA}
For any total preorder ${\leq} \subseteq T \times T$, $\merge_\leq$ is associative.
\begin{coqcode}
Lemma mergeA T (leT : rel T) :
  total leT -> transitive leT -> associative (merge leT).
\end{coqcode}
\end{lemma}

\subsection{Sigma types}

\begin{lemma}
\label{lemma:coq:in2_sig}
\label{lemma:coq:in3_sig}
Suppose $T_1$ and $T_2$ are types, $D_1 \subset T_1$ and $D_2 \subseteq T_2$ are predicates on them,
$P \subseteq T_1 \times T_2$ is a relation, and $\texttt{sig} \, D_1$ and $\texttt{sig} \, D_1$ are
subtypes of $T_1$ and $T_2$ collecting inhabitants satisfying $D_1$ and $D_2$, respectively.
Then, $(\texttt{val} \, x, \texttt{val} \, y) \in P$ holds for any $x \in \texttt{sig} \, D_1$ and
$y \in \texttt{sig} \, D_2$, whenever $(x, y) \in P$ holds for any $x \in T_1$ and $y \in T_2$ such
that $x \in D_1$ and $y \in D_2$.

Its ternary version also holds.
\begin{coqcode}
Lemma in2_sig T1 T2 (D1 : pred T1) (D2 : pred T2) (P2 : T1 -> T2 -> Prop) :
  {in D1 & D2, forall (x : T1) (y : T2), P2 x y} ->
  forall (x : sig D1) (y : sig D2), P2 (val x) (val y).

Lemma in3_sig
    T1 T2 T3 (D1 : pred T1) (D2 : pred T2) (D3 : pred T3)
    (P3 : T1 -> T2 -> T3 -> Prop) :
  {in D1 & D2 & D3, forall (x : T1) (y : T2) (z : T3), P3 x y z} ->
  forall (x : sig D1) (y : sig D2) (z : sig D3), P3 (val x) (val y) (val z).
\end{coqcode}
\end{lemma}

\begin{lemma}
\label{lemma:coq:all_sigP}
For any $p \subseteq T$ and a list $s$ of type $\mathrm{list} \, T$, there exists a list $s'$ of type
$\mathrm{list} \, (\texttt{sig} \, p)$ such that $s$ is equal to $\texttt{map}_\texttt{val} \, s'$,
whenever $\texttt{all}_p \, s$ holds.
\begin{coqcode}
Lemma all_sigP T (p : pred T) (s : list T) :
  all p s -> {s' : list (sig p) | s = map val s'}.
\end{coqcode}
\end{lemma}

\section{The theory of stable sort functions}
\label{appx:stablesort-theory}

This appendix provides the list of all lemmas and their proofs about stable sort functions we stated
and proved using the \coq{stableSort} structure (\cref{sec:the-interface}).
Each informal statement is followed by the corresponding formal statements in \Coq, which include
corollaries such as ones delimiting the domain by a predicate, \eg, \cref{corollary:filter_sort_in}.

\begin{lemma}[An induction principle over traces of \sort{}, \cref{lemma:sort_ind'}]
\label{lemma:coq:sort_ind}
Suppose $\leq$ and $\sim$ are binary relations on $T$ and $\mathrm{list} \, T$, respectively, and
$\mathit{xs}$ is a list of type $\mathrm{list} \, T$.
Then, $\mathit{xs} \sim \sort{}_\leq \, \mathit{xs}$ holds whenever the following four
induction cases hold:
\begin{itemize}
 \item for any lists $\mathit{xs}$, $\mathit{xs'}$, $\mathit{ys}$, and $\mathit{ys'}$ of type
       $\mathrm{list} \, T$,
       $(\mathit{xs} \concat \mathit{ys}) \sim (\mathit{xs'} \merge_\leq \mathit{ys'})$ holds
       whenever $\mathit{xs} \sim \mathit{xs'}$ and $\mathit{ys} \sim \mathit{ys'}$ hold,
 \item for any lists $\mathit{xs}$, $\mathit{xs'}$, $\mathit{ys}$, and $\mathit{ys'}$ of type
       $\mathrm{list} \, T$,
       $(\mathit{xs} \concat \mathit{ys}) \sim
	\texttt{rev} \, (\texttt{rev} \, \mathit{ys'} \merge_\geq \texttt{rev} \, \mathit{xs'})$
       holds whenever $\mathit{xs} \sim \mathit{xs'}$ and $\mathit{ys} \sim \mathit{ys'}$ hold,
 \item for any $x$ of type $T$, $[x] \sim [x]$ holds, and
 \item $[] \sim []$ holds.
\end{itemize}
\begin{coqcode}
Lemma sort_ind sort T (leT : rel T) (R : list T -> list T -> Prop) :
  (forall xs xs' : list T, R xs xs' -> forall ys ys' : list T, R ys ys' ->
    R (xs ++ ys) (merge leT xs' ys')) ->
  (forall xs xs' : list T, R xs xs' -> forall ys ys' : list T, R ys ys' ->
    R (xs ++ ys) (rev (merge (fun x y => leT y x) (rev ys') (rev xs')))) ->
  (forall x : T, R [:: x] [:: x]) ->
  R nil nil ->
  forall s : list T, R s (sort T leT s).
\end{coqcode}
\end{lemma}

\begin{proof}
 See \cref{lemma:sort_ind,lemma:sort_ind'}.
\end{proof}

\begin{lemma}[The naturality of \sort{}, \cref{lemma:sort_map}]
\label{lemma:coq:map_sort}
\label{lemma:coq:sort_map}
Suppose $\leq_T$ is a relation on type $T$, $f$ is a function from $T'$ to $T$, and $s$ is a list of
type $\mathrm{list} \, T'$. Then, the following equation holds:
\[
 \sort{}_{\leq_T} \, [f \, x \mid x \leftarrow s] =
 [f \, x \mid x \leftarrow \sort{}_{\leq_{T'}} \, s]
\]
where $x \leq_{T'} y \coloneq f \, x \leq_T f \, y$.
\begin{coqcode}
Lemma map_sort sort T T' (f : T' -> T) (leT' : rel T') (leT : rel T) :
  (forall x y : T', leT (f x) (f y) = leT' x y) ->
  forall s : list T', map f (sort T' leT' s) = sort T leT (map f s).

Lemma sort_map sort T T' (f : T' -> T) (leT : rel T) (s : list T') :
  sort T leT (map f s) = map f (sort T' (relpre f leT) s).
\end{coqcode}
\end{lemma}

\begin{proof}
 See \cref{lemma:sort_map}.
\end{proof}

\begin{lemma}[\cref{lemma:pairwise_sort}]
\label{lemma:coq:pairwise_sort}
For any ${\leq} \subseteq T \times T$ and $s$ of type $\mathrm{list} \, T$ pairwise sorted \wrt $\leq$,
$\sort{}_\leq \, s$ is equal to $s$.
\begin{coqcode}
Lemma pairwise_sort sort T (leT : rel T) (s : list T) :
  pairwise leT s -> sort T leT s = s.
\end{coqcode}
\end{lemma}

\begin{proof}
 We prove it by induction on $\sort{}_\leq \, s$ (\cref{lemma:coq:sort_ind}).
 For the first case, we show that $s'_1 \merge_\leq s'_2 = s_1 \concat s_2$ holds whenever
 $s_1 \concat s_2$ is pairwise sorted \wrt $\leq$, which is equivalent to the following conjunction
 (\cref{lemma:coq:pairwise_cat}):
 \begin{enumerate}[label=(\roman*)]
  \item \label{item:s1_s2_pairwise} both $s_1$ and $s_2$ are pairwise sorted \wrt $\leq$, and
  \item \label{item:s1_le_s2} $x \leq y$ holds for any $x \in s_1$ and $y \in s_2$.
 \end{enumerate}
 Then,
 \begin{align*}
  s'_1 \merge_\leq s'_2
  &= s_1 \merge_\leq s_2 & (\text{\ref{item:s1_s2_pairwise} and I.H.}) \\
  &= s_1 \concat s_2.    & (\text{\ref{item:s1_le_s2} and \cref{lemma:coq:allrel_merge}})
 \end{align*}

 For the second case, we show that
 $\texttt{rev} \, (\texttt{rev} \, s'_2 \merge_\geq \texttt{rev} \, s'_1) = s_1 \concat s_2$ holds
 whenever $s_1 \concat s_2$ is pairwise sorted \wrt $\leq$, which is equivalent to the following
 conjunction (\cref{lemma:coq:pairwise_cat}):
 \begin{enumerate}[label=(\roman*), resume]
  \item \label{item:s1_s2_pairwise'} both $s_1$ and $s_2$ are pairwise sorted \wrt $\leq$, and
  \item \label{item:s1_le_s2'} $x \leq y$ holds for any $x \in s_1$ and $y \in s_2$,
        or equivalently	(\cref{lemma:coq:allrelC,lemma:coq:allrel_rev2}),
        $y \geq x$ holds for any $y \in \texttt{rev} \, s_2$ and $x \in \texttt{rev} \, s_1$.
 \end{enumerate}
 Then,
 \begin{align*}
  \texttt{rev} \, (\texttt{rev} \, s'_2 \merge_\geq \texttt{rev} \, s'_1)
  &= \texttt{rev} \, (\texttt{rev} \, s_2 \merge_\geq \texttt{rev} \, s_1)
  & (\text{\ref{item:s1_s2_pairwise'} and I.H.}) \\
  &= \texttt{rev} \, (\texttt{rev} \, s_2 \concat \texttt{rev} \, s_1)
  & (\text{\ref{item:s1_le_s2'} and \cref{lemma:coq:allrel_merge}}) \\
  &= s_1 \concat s_2.
  & (\text{\cref{lemma:coq:rev_cat,lemma:coq:revK}})
 \end{align*}

 The last two cases are trivial.
\end{proof}

\begin{lemma}
\label{lemma:coq:sorted_sort}
\label{lemma:coq:sorted_sort_in}
For any transitive relation ${\leq} \subseteq T \times T$ and $s$ of type $\mathrm{list} \, T$
sorted \wrt $\leq$, $\sort{}_\leq \, s$ is equal to $s$.
\begin{coqcode}
Lemma sorted_sort sort T (leT : rel T) :
  transitive leT -> forall s : list T, sorted leT s -> sort T leT s = s.

Lemma sorted_sort_in sort T (P : pred T) (leT : rel T) :
  {in P & &, transitive leT} ->
  forall s : list T, all P s -> sorted leT s -> sort T leT s = s.
\end{coqcode}
\end{lemma}

\begin{proof}
 This follows from \cref{lemma:coq:sorted_pairwise,lemma:coq:pairwise_sort}.
\end{proof}

\begin{theorem}[\Cref{lemma:sort_pairwise_stable}]
\label{lemma:coq:sort_pairwise_stable}
\label{lemma:coq:sort_pairwise_stable_in}
For any ${\leq_1}, {\leq_2} \subseteq T \times T$ and $\mathit{xs}$ of type $\mathrm{list} \, T$,
$\sort{}_{\leq_1} \, \mathit{xs}$ is sorted \wrt the lexicographic order
${\leq_\mathrm{lex}} \coloneq {\leq_{(1, 2)}}$ whenever $\leq_1$ is total and $\mathit{xs}$ is
pairwise sorted \wrt $\leq_2$.
\begin{coqcode}
Lemma sort_pairwise_stable sort T (leT leT' : rel T) : total leT ->
  forall s : list T, pairwise leT' s -> sorted (lexord leT leT') (sort T leT s).

Lemma sort_pairwise_stable_in sort T (P : pred T) (leT leT' : rel T) :
  {in P &, total leT} -> forall s : list T, all P s -> pairwise leT' s ->
  sorted (lexord leT leT') (sort T leT s).
\end{coqcode}
\end{theorem}

\begin{proof}
 We prove it by induction on $\sort{}_\leq \, \mathit{xs}$ (\cref{lemma:coq:sort_ind}).
 See \cref{sec:induction-over-traces,sec:new-characterization-to-correctness} for details.
\end{proof}

\begin{lemma}[\Cref{corollary:sort_stable}]
\label{lemma:coq:sort_stable}
\label{lemma:coq:sort_stable_in}
For any ${\leq_1}, {\leq_2} \subseteq T \times T$ and $\mathit{xs}$ of type $\mathrm{list} \, T$,
$\sort{}_{\leq_1} \, \mathit{xs}$ is sorted \wrt the lexicographic order
${\leq_\mathrm{lex}} \coloneq {\leq_{(1, 2)}}$ whenever $\leq_1$ is total, $\leq_2$ is transitive,
and $\mathit{xs}$ is sorted \wrt $\leq_2$.
\begin{coqcode}
Lemma sort_stable sort T (leT leT' : rel T) : total leT -> transitive leT' ->
  forall s : list T, sorted leT' s -> sorted (lexord leT leT') (sort T leT s).

Lemma sort_stable_in sort T (P : pred T) (leT leT' : rel T) :
  {in P &, total leT} -> {in P & &, transitive leT'} ->
  forall s : list T, all P s -> sorted leT' s ->
  sorted (lexord leT leT') (sort T leT s).
\end{coqcode}
\end{lemma}

\begin{proof}
 This follows from \cref{lemma:coq:sorted_pairwise,lemma:coq:sort_pairwise_stable}.
\end{proof}

\begin{lemma}
\label{lemma:coq:count_sort}
For any ${\leq} \subseteq T \times T$, $p \subseteq T$, and $s$ of type $\mathrm{list} \, T$,
the numbers of the elements of $\sort{}_\leq \, s$ and $s$ satisfying $p$ are the same, \ie,
$\lvert \sort{}_\leq \, s \rvert_p = \lvert s \rvert_p$.
\begin{coqcode}
Lemma count_sort sort T (leT : rel T) (p : pred T) (s : list T) :
  count p (sort T leT s) = count p s.
\end{coqcode}
\end{lemma}

\begin{proof}
 We prove it by induction on $\sort{}_\leq \, s$ (\cref{lemma:coq:sort_ind}).
 For the first case,
 \begin{align*}
  \lvert s'_1 \merge_\leq s'_2 \rvert_p
  &= \lvert s'_1 \rvert_p + \lvert s'_2 \rvert_p
  & (\text{\cref{lemma:coq:count_merge,lemma:coq:count_cat}}) \\
  &= \lvert s_1 \rvert_p + \lvert s_2 \rvert_p   & (\text{I.H.}) \\
  &= \lvert s_1 \concat s_2 \rvert_p.            & (\text{\cref{lemma:coq:count_cat}})
  \intertext{For the second case,}
  \lvert \texttt{rev} \, (\texttt{rev} \, s'_2 \merge_\geq \texttt{rev} \, s'_1) \rvert_p
  &= \lvert s'_2 \rvert_p + \lvert s'_1 \rvert_p
  & (\text{\cref{lemma:coq:count_merge,lemma:coq:count_cat,lemma:coq:count_rev}}) \\
  &= \lvert s'_1 \rvert_p + \lvert s'_2 \rvert_p \\
  &= \lvert s_1 \rvert_p + \lvert s_2 \rvert_p   & (\text{I.H.}) \\
  &= \lvert s_1 \concat s_2 \rvert_p.            & (\text{\cref{lemma:coq:count_cat}})
 \end{align*}
 The last two cases are trivial.
\end{proof}

\begin{lemma}
\label{lemma:coq:size_sort}
For any ${\leq} \subseteq T \times T$ and $s$ of type $\mathrm{list} \, T$,
the lengths of $\sort{}_\leq \, s$ and $s$ are equal, \ie,
$\lvert \sort{}_\leq \, s \rvert = \lvert s \rvert$.
\begin{coqcode}
Lemma size_sort sort T (leT : rel T) (s : seq T) : size (sort T leT s) = size s.
\end{coqcode}
\end{lemma}

\begin{proof}
 This follows from \cref{lemma:coq:count_predT,lemma:coq:count_sort}.
\end{proof}

\begin{lemma}
\label{lemma:coq:sort_nil}
Sorting the empty list gives the empty list.
\begin{coqcode}
Lemma sort_nil sort T (leT : rel T) : sort T leT [::] = [::].
\end{coqcode}
\end{lemma}

\begin{proof}
 The length of $\sort{}_\leq \, []$ is $0$ (\cref{lemma:coq:size_sort}), and thus, it must be
 the empty list.
\end{proof}

\begin{lemma}
\label{lemma:coq:all_sort}
For any $p \subseteq T$, ${\leq} \subseteq T \times T$, and $s$ of type $\mathrm{list} \, T$,
all elements of $\sort{}_\leq \, s$ satisfy $p$ iff all elements of $s$ satisfy $p$.
\begin{coqcode}
Lemma all_sort sort T (p : pred T) (leT : rel T) (s : list T) :
  all p (sort T leT s) = all p s.
\end{coqcode}
\end{lemma}

\begin{proof}
 This follows from \cref{lemma:coq:all_count,lemma:coq:count_sort,lemma:coq:size_sort}.
\end{proof}

\begin{lemma}[\Cref{lemma:perm_sort}]
\label{lemma:coq:perm_sort}
For any ${\leq} \subseteq T \times T$ and $s$ of type $\mathrm{list} \, T$,
$\sort{}_\leq \, s$ is a permutation of $s$.
\begin{coqcode}
Lemma perm_sort sort (T : eqType) (leT : rel T) (s : list T) :
  perm_eql (sort T leT s) s.
\end{coqcode}
\end{lemma}

\begin{proof}
 This follows from \cref{lemma:coq:permP,lemma:coq:count_sort}.
\end{proof}

\begin{lemma}[\Cref{corollary:mem_sort}]
\label{lemma:coq:mem_sort}
For any ${\leq} \subseteq T \times T$ and $s$ of type $\mathrm{list} \, T$,
$\sort{}_\leq \, s$ has the same set of elements as $s$.
\begin{coqcode}
Lemma mem_sort sort (T : eqType) (leT : rel T) (s : list T) : sort T leT s =i s.
\end{coqcode}
\end{lemma}

\begin{proof}
 This follows from \cref{lemma:coq:perm_mem,lemma:coq:perm_sort}.
\end{proof}

\begin{lemma}
\label{lemma:coq:sort_uniq}
For any ${\leq} \subseteq T \times T$ and $s$ of type $\mathrm{list} \, T$,
$\sort{}_\leq \, s$ is duplication free iff $s$ is duplication free.
\begin{coqcode}
Lemma sort_uniq sort (T : eqType) (leT : rel T) (s : list T) :
  uniq (sort T leT s) = uniq s.
\end{coqcode}
\end{lemma}

\begin{proof}
 This follows from \cref{lemma:coq:perm_uniq,lemma:coq:perm_sort}.
\end{proof}

\begin{theorem}[\Cref{lemma:filter_sort,corollary:filter_sort_in}]
\label{lemma:coq:filter_sort}
\label{lemma:coq:filter_sort_in}
For any total preorder ${\leq} \subseteq T \times T$ and predicate $p \subseteq T$,
$\texttt{filter}_p$ commutes with $\sort{}_\leq$ under function composition; that is, the
following equation holds for any $s$ of type $\mathrm{list} \, T$:
\[
 \texttt{filter}_p \, (\sort{}_\leq \, s) = \sort{}_\leq \, (\texttt{filter}_p \, s).
\]
\begin{coqcode}
Lemma filter_sort sort T (leT : rel T) : total leT -> transitive leT ->
  forall (p : pred T) (s : seq T), filter p (sort T leT s) = sort T leT (filter p s).

Lemma filter_sort_in sort T (P : pred T) (leT : rel T) :
  {in P &, total leT} -> {in P & &, transitive leT} ->
  forall (p : pred T) (s : seq T),
  all P s -> filter p (sort T leT s) = sort T leT (filter p s).
\end{coqcode}
\end{theorem}

\begin{proof}
 This follows mainly from
 \cref{lemma:coq:mkseq_nth,lemma:coq:sort_map,lemma:coq:filter_map,lemma:coq:irr_sorted_eq,lemma:coq:sort_stable}.
 See the proof of \cref{lemma:filter_sort} for details.
\end{proof}

\begin{lemma}
\label{lemma:coq:sorted_filter_sort}
\label{lemma:coq:sorted_filter_sort_in}
For any total preorder ${\leq} \subseteq T \times T$, predicate $p \subseteq T$, and $s$ of type
$\mathrm{list} \, T$,
\[
 \texttt{filter}_p \, (\sort{}_\leq \, s) = \texttt{filter}_p \, s
\]
holds whenever $\texttt{filter}_p \, s$ sorted \wrt $\leq$.
\begin{coqcode}
Lemma sorted_filter_sort sort T (leT : rel T) :
  total leT -> transitive leT ->
  forall (p : pred T) (s : list T),
  sorted leT (filter p s) -> filter p (sort _ leT s) = filter p s.

Lemma sorted_filter_sort_in sort T (P : {pred T}) (leT : rel T) :
  {in P &, total leT} -> {in P & &, transitive leT} ->
  forall (p : pred T) (s : list T),
  all P s -> sorted leT (filter p s) -> filter p (sort _ leT s) = filter p s.
\end{coqcode}
\end{lemma}

\begin{proof}
 This follows from \cref{lemma:coq:sorted_sort,lemma:coq:filter_sort}.
\end{proof}

\begin{lemma}
\label{lemma:coq:sort_sort}
\label{lemma:coq:sort_sort_in}
For any total preorders ${\leq_1}, {\leq_2} \subseteq T \times T$,
sorting a list with $\leq_2$ and then $\leq_1$ gives the same result as sorting the same list with
the lexicographic order ${\leq_\mathrm{lex}} \coloneq {\leq_{(1, 2)}}$;
that is, the following equation holds for any $s$ of type $\mathrm{list} \, T$:
\[
 \sort{}_{\leq_1} \, (\sort{}_{\leq_2} \, s) = \sort{}_{\leq_\mathrm{lex}} \, s.
\]
\begin{coqcode}
Lemma sort_sort sort T (leT leT' : rel T) :
  total leT -> transitive leT -> total leT' -> transitive leT' ->
  forall s : list T, sort T leT (sort T leT' s) = sort T (lexord leT leT') s.

Lemma sort_sort_in sort T (P : pred T) (leT leT' : rel T) :
  {in P &, total leT} -> {in P & &, transitive leT} ->
  {in P &, total leT'} -> {in P & &, transitive leT'} ->
  forall s : list T,
  all P s -> sort T leT (sort T leT' s) = sort T (lexord leT leT') s.
\end{coqcode}
\end{lemma}

\begin{proof}
 As in the proof of \cref{lemma:coq:filter_sort}, we replace $s$ everywhere with
 $\texttt{map}_{\texttt{nth} \, x_0 \, s} \, [0, \dots, \lvert s \rvert - 1]$
 (\cref{lemma:coq:mkseq_nth}), use the naturality of \sort{} (\cref{lemma:coq:sort_map}), and
 apply the congruence rule with respect to \texttt{map}.
 It remains to prove
 \[
  \sort{}_{\leq'_1} \, (\sort{}_{\leq'_2} \, [0, \dots, \lvert s \rvert - 1])
  = \sort{}_{\leq'_\mathrm{lex}} \, [0, \dots, \lvert s \rvert - 1]
 \]
 where $\leq'_1$, $\leq'_2$, and $\leq'_\mathrm{lex}$ are the preimages of
 $\leq_1$, $\leq_2$, and $\leq_\mathrm{lex}$ under $\texttt{nth} \, x_0 \, s$, respectively.
 Thanks to \cref{lemma:coq:sort_stable}, each side of the above equation is respectively sorted \wrt
 \begin{itemize}
  \item the (right-associative) lexicographic composition of $\leq'_1$, $\leq'_2$, and
	$<_\mathbb{N}$, and
  \item the lexicographic composition of $\leq'_\mathrm{lex}$ and $<_\mathbb{N}$, which is by
	definition equal to the (left-associative) lexicographic composition of $\leq'_1$,
	$\leq'_2$, and $<_\mathbb{N}$
 \end{itemize}
 since
 \begin{itemize}
  \item $\leq'_1$, $\leq'_2$, and their lexicographic composition $\leq'_\mathrm{lex}$ are total
	(\cref{lemma:coq:lexord_total}),
  \item $\leq'_2$, $<_\mathbb{N}$, and their lexicographic composition are transitive
	(\cref{lemma:coq:lexord_trans}), and
  \item $[0, \dots, \lvert s \rvert - 1]$ is sorted \wrt $<_\mathbb{N}$
	(\cref{lemma:coq:iota_ltn_sorted}).
 \end{itemize}
 Two lexicographic compositions of $\leq'_1$, $\leq'_2$, and $<_\mathbb{N}$ are equivalent
 regardless of the associativity (\cref{lemma:coq:lexordA}), and transitive
 (\cref{lemma:coq:lexord_trans}) and irreflexive.
 Both sides of the equation have the same set of elements (\cref{lemma:coq:mem_sort}).
 Therefore, these two lists are equal (\cref{lemma:coq:irr_sorted_eq}).
\end{proof}

\begin{lemma}
\label{lemma:coq:sort_sorted}
\label{lemma:coq:sort_sorted_in}
For any ${\leq} \subseteq T \times T$ and $s$ of type $\mathrm{list} \, T$,
$\sort{}_\leq \, s$ is sorted \wrt $\leq$ whenever $\leq$ is total.
\begin{coqcode}
Lemma sort_sorted sort T (leT : rel T) :
  total leT -> forall s : list T, sorted leT (sort T leT s).

Lemma sort_sorted_in sort T (P : pred T) (leT : rel T) :
  {in P &, total leT} ->
  forall s : list T, all P s -> sorted leT (sort T leT s).
\end{coqcode}
\end{lemma}

\begin{proof}
 Since $s$ is sorted \wrt the trivial relation $\mathrel{R}$ such that $x \mathrel{R} y$ holds for any
 $x, y \in T$, \cref{corollary:sort_stable} implies that $\sort{}_\leq \, s$ is sorted \wrt the
 lexicographic order of $\leq$ and $\mathrel{R}$, which is equivalent to $\leq$.
\end{proof}

\begin{lemma}
\label{lemma:coq:perm_sortP}
\label{lemma:coq:perm_sort_inP}
For any total order ${\leq} \subseteq T \times T$, and $s_1$ and $s_2$ of type $\mathrel{list} \, T$,
$\sort{}_\leq \, s_1 = \sort{}_\leq \, s_2$ holds iff $s_1$ is a permutation of $s_2$.
\begin{coqcode}
Lemma perm_sortP sort (T : eqType) (leT : rel T) :
  total leT -> transitive leT -> antisymmetric leT ->
  forall s1 s2 : list T,
  reflect (sort T leT s1 = sort T leT s2) (perm_eq s1 s2).

Lemma perm_sort_inP sort (T : eqType) (leT : rel T) (s1 s2 : list T) :
  {in s1 &, total leT} -> {in s1 & &, transitive leT} ->
  {in s1 &, antisymmetric leT} ->
  reflect (sort T leT s1 = sort T leT s2) (perm_eq s1 s2).
\end{coqcode}
\end{lemma}

\begin{proof}
 Since $\sort{}_\leq \, s_1$ and $\sort{}_\leq \, s_2$ are respectively permutations of
 $s_1$ and $s_2$, $\sort{}_\leq \, s_1 \permeq \sort{}_\leq \, s_2$ holds iff
 $s_1 \permeq s_2$ holds. Therefore, ($\Rightarrow$) is trivial.

 ($\Leftarrow$) $\sort{}_\leq \, s_1$ and $\sort{}_\leq \, s_2$ are sorted \wrt $\leq$
 (\cref{lemma:coq:sort_sorted}) and permutation of each other.
 Thanks to \cref{lemma:coq:sorted_eq}, these two sorted lists are equal.
\end{proof}

\begin{lemma}[{\citet[Theorem 2.9 (Uniqueness of sorting)]{functional_algorithms_verified}}]
\label{lemma:coq:eq_sort}
\label{lemma:coq:eq_in_sort}
For any two stable sort functions $\sort{}$ and $\sort{}'$ and total preorder
${\leq} \subseteq T \times T$, $\sort{}_\leq$ and $\sort{}'_\leq$ are extensionally
equal; that is, $\sort{}_\leq \, s = \sort{}'_\leq \, s$ holds for any $s$ of type
$\mathrm{list} \, T$.
\begin{coqcode}
Lemma eq_sort sort1 sort2 T (leT : rel T) :
  total leT -> transitive leT -> sort1 T leT =1 sort2 T leT.

Lemma eq_in_sort sort1 sort2 T (P : pred T) (leT : rel T) :
  {in P &, total leT} -> {in P & &, transitive leT} ->
  forall s : list T, all P s -> sort1 T leT s = sort2 T leT s.
\end{coqcode}
\end{lemma}

\begin{proof}
 We replace $s$ everywhere with
 $\texttt{map}_{\texttt{nth} \, x_0 \, s} \, [0, \dots, \lvert s \rvert - 1]$
 (\cref{lemma:coq:mkseq_nth}), use the naturality of $\sort{}$ and $\sort{}'$
 (\cref{lemma:coq:sort_map}), and apply the congruence rule with respect to \texttt{map}.
 It remains to prove
 \[
  \sort{}_{\leq_I} \, [0, \dots, \lvert s \rvert - 1]
  = \sort{}'_{\leq_I} \, [0, \dots, \lvert s \rvert - 1]
 \]
 where $x \leq_I y \coloneq \texttt{nth} \, x_0 \, s \, x \leq \texttt{nth} \, x_0 \, s \, y$.

 Since $[0, \dots, \lvert s \rvert - 1]$ is sorted \wrt $<_\mathbb{N}$
 (\cref{lemma:coq:iota_ltn_sorted}), both sides of the above equation are sorted \wrt $<_I$, the
 lexicographic composition of $\leq_I$ and $<_\mathbb{N}$ (\cref{lemma:coq:sort_stable}).
 These two lists have the same set of elements (\cref{lemma:coq:mem_sort}).
 Therefore, they are equal (\cref{lemma:coq:irr_sorted_eq}).
\end{proof}

\begin{lemma}
\label{lemma:coq:eq_sort_insert}
\label{lemma:coq:eq_in_sort_insert}
For any total preorder ${\leq} \subseteq T \times T$ and $s$ of type $\mathrm{list} \, T$,
$\sort{}_\leq \, s = \texttt{isort}_\leq \, s$ holds for the insertion sort \texttt{isort}
defined as follows.
\begin{align*}
 \texttt{isort}_\leq \, [] &\coloneq [] \\
 \texttt{isort}_\leq \, (x :: s) &\coloneq [x] \merge_\leq \texttt{isort}_\leq \, s.
\end{align*}
\begin{coqcode}
Lemma eq_sort_insert sort T (leT : rel T) : total leT -> transitive leT ->
  sort T leT =1 foldr (fun x : T => merge leT [:: x]) [::].

Lemma eq_in_sort_insert sort T (P : pred T) (leT : rel T) :
  {in P &, total leT} -> {in P & &, transitive leT} ->
  forall s : list T, all P s ->
  sort T leT s = foldr (fun x : T => merge leT [:: x]) [::] s.
\end{coqcode}
\end{lemma}

\begin{proof}
 This follows from \cref{lemma:coq:eq_sort} and the fact that \texttt{isort} satisfies the
 characteristic property (\cref{lemma:insertion-sort-stable}).
\end{proof}

\begin{lemma}
\label{lemma:coq:sort_cat}
\label{lemma:coq:sort_cat_in}
For any total preorder ${\leq} \subseteq T \times T$, and $s_1$ and $s_2$ of type
$\mathrm{list} \, T$,
$\sort{}_\leq \, (s_1 \concat s_2) =
 \sort{}_\leq \, s_1 \merge_\leq \sort{}_\leq \, s_2$
holds.
\begin{coqcode}
Lemma sort_cat sort T (leT : rel T) : total leT -> transitive leT ->
  forall s1 s2 : seq T,
  sort T leT (s1 ++ s2) = merge leT (sort T leT s1) (sort T leT s2).

Lemma sort_cat_in sort T (P : pred T) (leT : rel T) :
  {in P &, total leT} -> {in P & &, transitive leT} ->
  forall s1 s2 : seq T, all P s1 -> all P s2 ->
  sort T leT (s1 ++ s2) = merge leT (sort T leT s1) (sort T leT s2).
\end{coqcode}
\end{lemma}

\begin{proof}
 We replace all the occurrences of \sort{} with the insertion sort \texttt{isort}
 (\cref{lemma:coq:eq_sort_insert}), and prove the following equation by induction on $s_1$:
 \[
 \texttt{isort}_\leq \, (s_1 \concat s_2) =
 \texttt{isort}_\leq \, s_1 \merge_\leq \texttt{isort}_\leq \, s_2.
 \]
 If $s_1 = []$, the equation holds by definition.
 Otherwise, $s_1 = x :: s'_1$ and
 \begin{align*}
  \texttt{isort}_\leq \, (s_1 \concat s_2)
  &= [x] \merge_\leq \texttt{isort}_\leq \, (s'_1 \concat s_2)
  & (\text{Definition}) \\
  &= [x] \merge_\leq (\texttt{isort}_\leq \, s'_1 \merge_\leq \texttt{isort}_\leq \, s_2)
  & (\text{I.H.}) \\
  &= ([x] \merge_\leq \texttt{isort}_\leq \, s'_1) \merge_\leq \texttt{isort}_\leq \, s_2
  & (\text{\cref{lemma:coq:mergeA}}) \\
  &= \texttt{isort}_\leq \, s_1 \merge_\leq \texttt{isort}_\leq \, s_2.
  & (\text{Definition}) & \qedhere
 \end{align*}
\end{proof}

\begin{lemma}
\label{lemma:coq:mask_sort}
\label{lemma:coq:mask_sort_in}
For any total preorder ${\leq} \subseteq T \times T$, $s$ of type $\mathrm{list} \, T$, and $m$ of
type $\mathrm{list} \, \mathrm{bool}$, there exists $m'$ of type $\mathrm{list} \, \mathrm{bool}$
that satisfies
$\texttt{mask}_{m'} \, (\sort{}_\leq \, s) = \sort{}_\leq \, (\texttt{mask}_m \, s)$.
\begin{coqcode}
Lemma mask_sort sort T (leT : rel T) : total leT -> transitive leT ->
  forall (s : list T) (m : list bool),
  {m_s : list bool | mask m_s (sort T leT s) = sort T leT (mask m s)}.

Lemma mask_sort_in sort T (P : pred T) (leT : rel T) :
  {in P &, total leT} -> {in P & &, transitive leT} ->
  forall (s : list T) (m : list bool), all P s ->
  {m_s : list bool | mask m_s (sort T leT s) = sort T leT (mask m s)}.
\end{coqcode}
\end{lemma}

\begin{proof}
 We replace $s$ everywhere with
 $\texttt{map}_{\texttt{nth} \, x_0 \, s} \, [0, \dots, \lvert s \rvert - 1]$
 (\cref{lemma:coq:mkseq_nth}), use the naturality of \sort{} (\cref{lemma:coq:sort_map}) and
 \texttt{mask} (\cref{lemma:coq:map_mask}), and apply the congruence rule with respect to
 \texttt{map}.
 It suffices to find $m'$ such that
 \[
 \texttt{mask}_{m'} \, (\sort{}_{\leq_I} \, [0, \dots, \lvert s \rvert - 1])
 = \sort{}_{\leq_I} \, (\texttt{mask}_m \, [0, \dots, \lvert s \rvert - 1])
 \]
 where $i \leq_I j \coloneq \texttt{nth} \, x_0 \, s \, i \leq \texttt{nth} \, x_0 \, s \, j$.
 We show that
 $m' \coloneq \texttt{map}_p \, (\sort{}_{\leq_I} \, [0, \dots, \lvert s \rvert - 1])$
 where $p \, i \coloneq i \in \texttt{mask}_m \, [0, \dots, \lvert s \rvert - 1]$ satisfy the above
 equation:
 \begin{align*}
      & \texttt{mask}_{m'} \, (\sort{}_{\leq_I} \, [0, \dots, \lvert s \rvert - 1]) \\
  ={} & \texttt{filter}_p \, (\sort{}_{\leq_I} \, [0, \dots, \lvert s \rvert - 1])
      & (\text{\cref{lemma:coq:filter_mask}}) \\
  ={} & \sort{}_{\leq_I} \, (\texttt{filter}_p \, [0, \dots, \lvert s \rvert - 1])
      & (\text{\cref{lemma:coq:filter_sort}}) \\
  ={} & \sort{}_{\leq_I} \, (\texttt{mask}_m \, [0, \dots, \lvert s \rvert - 1]).
      & (\text{\cref{lemma:coq:mask_filter,lemma:coq:iota_uniq}}) & \qedhere
 \end{align*}
\end{proof}

\begin{lemma}
\label{lemma:coq:sorted_mask_sort}
\label{lemma:coq:sorted_mask_sort_in}
For any total preorder ${\leq} \subseteq T \times T$, $s$ of type $\mathrm{list} \, T$, and $m$ of
type $\mathrm{list} \, \mathrm{bool}$, there exists $m'$ of type $\mathrm{list} \, \mathrm{bool}$
that satisfies $\texttt{mask}_{m'} \, (\sort{}_\leq \, s) = \texttt{mask}_m \, s$, whenever
$\texttt{mask}_m \, s$ is sorted \wrt $\leq$.
\begin{coqcode}
Lemma sorted_mask_sort sort T (leT : rel T) : total leT -> transitive leT ->
  forall (s : list T) (m : list bool), sorted leT (mask m s) ->
  {m_s : list bool | mask m_s (sort T leT s) = mask m s}.

Lemma sorted_mask_sort_in sort T (P : pred T) (leT : rel T) :
  {in P &, total leT} -> {in P & &, transitive leT} ->
  forall (s : list T) (m : list bool), all P s -> sorted leT (mask m s) ->
  {m_s : list bool | mask m_s (sort T leT s) = mask m s}.
\end{coqcode}
\end{lemma}

\begin{proof}
 This follows from \cref{lemma:coq:sorted_sort,lemma:coq:mask_sort}.
\end{proof}

\begin{lemma}
\label{lemma:coq:subseq_sort}
\label{lemma:coq:subseq_sort_in}
For any total preorder ${\leq} \subseteq T \times T$, $\sort{}_\leq$ preserves the subsequence
relation; that is, for any $t$ and $s$ of type $\mathrm{list} \, T$,
$\sort{}_\leq \, t$ is a subsequence of $\sort{}_\leq \, s$ whenever $t$ is a subsequence
of $s$.
\begin{coqcode}
Lemma subseq_sort sort (T : eqType) (leT : rel T) :
  total leT -> transitive leT ->
  forall t s : list T, subseq t s -> subseq (sort T leT t) (sort T leT s).

Lemma subseq_sort_in sort (T : eqType) (leT : rel T) (t s : list T) :
  {in s &, total leT} -> {in s & &, transitive leT} ->
  subseq t s -> subseq (sort T leT t) (sort T leT s).
\end{coqcode}
\end{lemma}

\begin{proof}
 Using \cref{lemma:coq:subseqP} that $t$ is a subsequence of $s$ iff there exists $m$ of the same
 size as~$s$ such that $t = \texttt{mask}_m~s$, we are left to show that
 $\sort{}_{\leq} \, (\texttt{mask}_m~s)$ is a subsequence of $\sort{}_{\leq}~s$.
 Now, by \cref{lemma:coq:mask_sort}, there is a mask $m'$ such that
 $\sort{}_{\leq} \, (\texttt{mask}_m~s) = \texttt{mask}_{m'} \, (\sort{}_{\leq}~s)$
 which is indeed a subsequence of $\sort{}_{\leq}~s$ (\cref{lemma:coq:mask_subseq}).
\end{proof}

\begin{lemma}
\label{lemma:coq:sorted_subseq_sort}
\label{lemma:coq:sorted_subseq_sort_in}
For any total preorder ${\leq} \subseteq T \times T$, and $t$ and $s$ of type $\mathrm{list} \, T$,
$t$ is a subsequence of $\sort{}_\leq \, s$ whenever $t$ is a subsequence of $s$ and sorted \wrt
$\leq$.
\begin{coqcode}
Lemma sorted_subseq_sort sort (T : eqType) (leT : rel T) :
  total leT -> transitive leT ->
  forall t s : list T, subseq t s -> sorted leT t -> subseq t (sort T leT s).

Lemma sorted_subseq_sort_in sort (T : eqType) (leT : rel T) (t s : list T) :
  {in s &, total leT} -> {in s & &, transitive leT} ->
  subseq t s -> sorted leT t -> subseq t (sort T leT s).
\end{coqcode}
\end{lemma}

\begin{proof}
This follows from \cref{lemma:coq:sorted_sort,lemma:coq:subseq_sort}.
\end{proof}

\begin{lemma}
\label{lemma:coq:mem2_sort}
\label{lemma:coq:mem2_sort_in}
For any total preorder ${\leq} \subseteq T \times T$, $s$ of type $\mathrm{list} \, T$, and $x$ and
$y$ of type $T$, $x$ and $y$ occur in order in $\sort{}_\leq \, s$ whenever $x \leq y$ holds
and $x$ and $y$ occur in order in $s$.
\begin{coqcode}
Lemma mem2_sort sort (T : eqType) (leT : rel T) :
  total leT -> transitive leT ->
  forall (s : list T) (x y : T),
  leT x y -> mem2 s x y -> mem2 (sort T leT s) x y.

Lemma mem2_sort_in sort (T : eqType) (leT : rel T) (s : list T) :
  {in s &, total leT} -> {in s & &, transitive leT} ->
  forall x y : T, leT x y -> mem2 s x y -> mem2 (sort T leT s) x y.
\end{coqcode}
\end{lemma}

\begin{proof}
If $x = y$, the statement means that all the elements of $s$ are in
$\sort{}_{\leq}~s$, which is a consequence of \cref{lemma:coq:mem_sort}.
Otherwise we remark that $x$ and $y$ occur in order in $s$ iff the sequence
$[x, y]$ is a subsequence of $s$ (\cref{lemma:coq:mem2E}).
Using \cref{lemma:coq:subseq_sort}, we thus have
that $\sort{}_{\leq} \, [x, y]$ is a subsequence of $\sort{}_{\leq}~s$.
Since $x \leq y$, we have that $\sort{}_{\leq}~[x, y] = [x, y]$ (\eg, by
computation through \cref{lemma:coq:eq_sort_insert}), hence $[x, y]$ a
subsequence of $\sort{}_{\leq}~s$.
\end{proof}

\section{Comparison of formulations of the stability}
\label{appx:stability-statements}

In this appendix, we compare our stability results
(\cref{lemma:coq:sort_pairwise_stable,lemma:coq:sort_stable,lemma:coq:filter_sort,lemma:coq:sorted_filter_sort})
with the standard formulation (\cref{corollary:sort_standard_stable}) and its versions formally
proved in related work~\cite{Leroy:mergesort, DBLP:journals/jar/Sternagel13,
  DBLP:journals/tocl/LeinoL15}.

\Citet{Leroy:mergesort} proved the stability of a mergesort function in the following form.

\begin{lemma}[\citet{Leroy:mergesort}, \cref{corollary:sort_standard_stable}]
\label{lemma:coq:sort_stable_leroy}
For any total preorder $\leq$ on $T$, the equivalent elements always appear in the same order in the
input and output of sorting; that is, the following equation holds for any $x$ of type $T$ and $s$
of type $\mathrm{list} \, T$:
\[
 [y \leftarrow \sort{}_\leq \, s \mid x \equiv y] = [y \leftarrow s \mid x \equiv y].
\]
\begin{coqcode}
Lemma sort_stable_leroy sort T (leT : rel T) :
  total leT -> transitive leT ->
  forall (x : T) (s : list T),
    [seq y <- sort T leT s | leT x y && leT y x] =
      [seq y <- s | leT x y && leT y x].
\end{coqcode}
\end{lemma}

\begin{proof}
 This follows from \cref{lemma:coq:filter_sort} or \cref{lemma:coq:sorted_filter_sort}.
 See \cref{corollary:sort_standard_stable}.
\end{proof}

\Citet{DBLP:journals/jar/Sternagel13, DBLP:journals/tocl/LeinoL15} proved the stability of
\GHC's mergesort in the following form, where the total preorder on $T$ is decomposed into a totally
ordered type $T'$ and a keying function $k$ of type $T \to T'$.

\begin{lemma}[\citet{DBLP:journals/jar/Sternagel13, DBLP:journals/tocl/LeinoL15}]
\label{lemma:coq:sort_stable_sternagel}
Let $k : T \to T'$ be a keying function whose codomain $T'$ is totally ordered by $\leq_{T'}$,
and ${\leq_T} \subseteq T \times T$ be the total preorder induced by $k$, \ie,
$x \leq_T y \coloneq k \, x \leq_{T'} k \, y$.
For any $s$ of type $\mathrm{list} \, T$, the relative order of the elements of $s$ having the same
key is preserved by $\sort{}_{\leq_T} \, s$; that is, the following equation holds for any
$x$ of type $T$:
\[
 [y \leftarrow \sort{}_{\leq_T} \, s \mid k \, x = k \, y] = [y \leftarrow s \mid k \, x = k \, y].
\]
\begin{coqcode}
Lemma sort_stable_sternagel
    sort T (d : unit) (T' : orderType d) (key : T -> T') (x : T) (s : list T) :
  [seq y <- sort T (relpre key <=%O) s | key x == key y] =
    [seq y <- s | key x == key y].
\end{coqcode}
where \coq{orderType} is the interface of totally ordered types and \coq{<=
to a totally ordered type.
\end{lemma}

\begin{proof}
 Since $\leq_{T'}$ is total order hence antisymmetric,
 $x$ and $y$ of type $T$ have the same key ($k \, x = k \, y$) iff they have the same order
 ($x \leq_T y \land y \leq_T x$).
 Since $\leq_T$ is total preorder, this follows from \cref{lemma:coq:sort_stable_leroy}.
\end{proof}

As we argued in \cref{sec:sort_standard_stable}, both
\cref{lemma:coq:sort_stable_sternagel,lemma:coq:sort_stable_leroy} are easy consequences of
\cref{lemma:coq:filter_sort} or \cref{lemma:coq:sorted_filter_sort}, which follow from
\cref{lemma:coq:sort_pairwise_stable} or \cref{lemma:coq:sort_stable}.
We formally proved its converse: \cref{lemma:coq:sort_stable_leroy} implies
\cref{lemma:coq:sort_pairwise_stable,lemma:coq:sort_stable,lemma:coq:filter_sort}, assuming some
extra conditions (mostly on $\sort{}$) detailed below.
However, their proofs are not as straightforward as the other direction, and
\cref{lemma:coq:sort_pairwise_stable,lemma:coq:sort_stable} proved in this way
(\cref{lemma:usual_stable_sort_pairwise_stable,lemma:usual_stable_sort_stable} below) require the
extra condition that $\leq_1$ is transitive.

\begin{lemma}[\cref{lemma:coq:sort_stable_leroy} implies \cref{lemma:coq:sort_pairwise_stable}]
 \label{lemma:usual_stable_sort_pairwise_stable}
 Let \sort{} be a function of type
 $\forall T, (T \to T \to \mathrm{bool}) \to \mathrm{list} \, T \to \mathrm{list} \, T$ such that
 \begin{itemize}
  \item $\sort{}_\leq \, [] = []$ for any $\leq$,
  \item $\sort{}_\leq \, s$ is sorted \wrt $\leq$ whenever $\leq$ is total, and
  \item $[y \leftarrow \sort{}_\leq \, s \mid x \equiv y] = [y \leftarrow s \mid x \equiv y]$
	for any type $T$, total preorder $\leq$ on $T$, $s$ of type $\mathrm{list} \, T$, and
	$x \in T$.
 \end{itemize}
 Then, for any type $T$, relations ${\leq_1}, {\leq_2} \subseteq T \times T$, and $s$ of type
 $\mathrm{list} \, T$, $\sort{}_{\leq_1} \, s$ is sorted \wrt the lexicographic order
 $\leq_{(1, 2)}$ whenever $\leq_1$ is total preorder and $s$ is pairwise sorted \wrt $\leq_2$.
\end{lemma}

\begin{lemma}[\cref{lemma:coq:sort_stable_leroy} implies \cref{lemma:coq:sort_stable}]
 \label{lemma:usual_stable_sort_stable}
 Let \sort{} be a function of type
 $\forall T, (T \to T \to \mathrm{bool}) \to \mathrm{list} \, T \to \mathrm{list} \, T$
 satisfying the same conditions as \cref{lemma:usual_stable_sort_pairwise_stable}.
 Then, for any type $T$, relations ${\leq_1}, {\leq_2} \subseteq T \times T$, and $s$ of type
 $\mathrm{list} \, T$, $\sort{}_{\leq_1} \, s$ is sorted \wrt the lexicographic order
 $\leq_{(1, 2)}$ whenever $\leq_1$ is total preorder, $\leq_2$ is transitive, and $s$ is sorted \wrt
 $\leq_2$.
\end{lemma}

\begin{lemma}[\cref{lemma:coq:sort_stable_leroy} implies \cref{lemma:coq:filter_sort}]
 \label{lemma:usual_stable_filter_sort}
 Let \sort{} be a function of type
 $\forall T, (T \to T \to \mathrm{bool}) \to \mathrm{list} \, T \to \mathrm{list} \, T$
 that satisfies the following conditions in addition to ones in
 \cref{lemma:usual_stable_sort_pairwise_stable}:
 \begin{itemize}
  \item $\sort{}_\leq \, s$ has the same set of elements as $s$ for any
	${\leq} \subseteq T \times T$ and $s$ of type $\mathrm{list} \, T$, and
  \item $\sort{}$ is natural, that is,
	$\sort{}_{\leq_T} \, [f \, x \mid x \leftarrow s] =
	 [f \, x \mid x \leftarrow \sort{}_{\leq_{T'}} \, s]$
	for any types $T$ and $T'$, function $f$ from $T'$ to $T$, relation
	${\leq_T} \subseteq T \times T$, list $s$ of type $\mathrm{list} \, T'$,
	where $x \leq_{T'} y \coloneq f \, x \leq_T f \, y$,
 \end{itemize}
 Then, for any type $T$, total preorder ${\leq} \subseteq T \times T$, and predicate
 $p \subseteq T$, $\texttt{filter}_p$ commutes with $\sort{}_\leq$ under function composition;
 that is, the following equation holds for any $s$ of type $\mathrm{list} \, T$:
 \[
  \texttt{filter}_p \, (\sort{}_\leq \, s) =
  \sort{}_\leq \, (\texttt{filter}_p \, s).
 \]
\end{lemma}

\section{Definitions of mergesorts in \Coq}
\label{appx:mergesort-in-coq}

In \cref{sec:termination}, we presented structurally-recursive non-tail-recursive and tail-recursive
mergesorts in \OCaml.
In this appendix, we present their \Coq reimplementations, including some optimized ones such as
smooth variants (\cref{sec:smooth-mergesort}).
The formal correctness proofs of these implementations are omitted in this appendix and available
only in the supplementary material.

\subsection{Non-tail-recursive mergesorts in \Coq}
\label{appx:nontailrec-mergesort-in-coq}

\begin{figure}[t]
\begin{coqcode}
Module Type MergeSig.
Parameter merge : forall (T : Type) (leT : rel T), list T -> list T -> list T.
Parameter mergeE : forall (T : Type) (leT : rel T), merge leT =2 path.merge leT.
End MergeSig.

Module Merge <: MergeSig.

Fixpoint merge (T : Type) (leT : rel T) (xs ys : list T) : list T :=
  if xs is x :: xs' then
    (fix merge' (ys : list T) : list T :=
       if ys is y :: ys' then
         if leT x y then x :: merge leT xs' ys else y :: merge' ys'
       else xs) ys
  else ys.

Lemma mergeE (T : Type) (leT : rel T) : merge leT =2 path.merge leT.

End Merge.
\end{coqcode}
\caption{Structurally-recursive non-tail-recursive merge in \Coq, \ie, the \Coq counterpart of
\caml{merge} in \cref{fig:struct-nontailrec-mergesort}.}
\label{fig:struct-nontailrec-merge-coq}
\end{figure}
\begin{figure}[t]
\begin{coqcode}
Module CBN_ (M : MergeSig).
Section CBN.
Context (T : Type) (leT : rel T).

Fixpoint push (xs : list T) (stack : list (list T)) {struct stack} :
    list (list T) :=
  match stack with
  | [::] :: stack | [::] as stack => xs :: stack
  | ys :: stack => [::] :: push (M.merge leT ys xs) stack
  end.

Fixpoint pop (xs : list T) (stack : list (list T)) {struct stack} : list T :=
  if stack is ys :: stack then pop (M.merge leT ys xs) stack else xs.

Fixpoint sort1rec (stack : list (list T)) (xs : list T) {struct xs} : list T :=
  if xs is x :: xs then sort1rec (push [:: x] stack) xs else pop [::] stack.

Definition sort1 : list T -> list T := sort1rec [::].

Fixpoint sort2rec (stack : list (list T)) (xs : list T) {struct xs} : list T :=
  if xs is x1 :: x2 :: xs then
    let t := if leT x1 x2 then [:: x1; x2] else [:: x2; x1] in
    sort2rec (push t stack) xs
  else pop xs stack.

Definition sort2 : list T -> list T := sort2rec [::].

[..]

End CBN.
End CBN_.
\end{coqcode}
\caption{Structurally-recursive non-tail-recursive mergesorts in \Coq.
The \coq{push}, \coq{pop}, \coq{sort1rec}, and \coq{sort1} functions in this figure are the \Coq
counterparts of \caml{push}, \caml{pop}, \caml{sort_rec}, and \caml{sort} in
\cref{fig:struct-nontailrec-mergesort}.
Some optimized implementations are omitted here as marked by \coq{[..]} and shown in
\cref{fig:struct-nontailrec-mergesort3-coq,fig:struct-nontailrec-mergesortN-coq}.}
\label{fig:struct-nontailrec-mergesort-coq}
\end{figure}

In \cref{fig:struct-nontailrec-merge-coq}, we define a module type~\cite{rocqrefman:modules}
\coq{MergeSig} abstracting out a non-tail-recursive merge function, and then, define one of its
instances \coq{Merge} using the nested fixpoint technique (\cref{sec:guard-condition}).
The purpose of this module type is to allow replacement of the implementation of merge function.
Therefore, the non-tail-recursive mergesort functions in \cref{fig:struct-nontailrec-mergesort-coq}
are provided as a functor module \coq{CBN_} that takes an instance \coq{M} of \coq{MergeSig} as a
parameter.
In order to prove the mergesort functions in the \coq{CBN_} module correct, \coq{MergeSig} and
\coq{Merge} are equipped with the proof \coq{mergeE} that the merge function implemented in the
module is extensionally equal to the one in the \MC library (\coq{path.merge} in
\cref{fig:struct-nontailrec-merge-coq}, or \cref{def:coq:merge}).

Inside the \coq{CBN_} module, we introduce a type \coq{T} and a relation \coq{leT} on \coq{T} as
section variables~\cite{rocqrefman:section}.
After closing the \coq{CBN} section, these variables become inaccessible and declarations in the
section will be parameterized by them.
The functions \coq{push}, \coq{pop}, \coq{sort1rec}, and \coq{sort1} inside this module correspond
to \caml{push}, \caml{pop}, \caml{sort_rec}, and \caml{sort} in
\cref{fig:struct-nontailrec-mergesort}, respectively.

\begin{figure}[p]
\begin{coqcode}
Fixpoint sort3rec (stack : list (list T)) (xs : list T) {struct xs} : list T :=
  match xs with
  | x1 :: x2 :: x3 :: xs =>
    let t :=
        if leT x1 x2 then
          if leT x2 x3 then [:: x1; x2; x3]
          else if leT x1 x3 then [:: x1; x3; x2] else [:: x3; x1; x2]
        else
          if leT x1 x3 then [:: x2; x1; x3]
          else if leT x2 x3 then [:: x2; x3; x1] else [:: x3; x2; x1]
    in
    sort3rec (push t stack) xs
  | [:: x1; x2] => pop (if leT x1 x2 then xs else [:: x2; x1]) stack
  | _ => pop xs stack
  end.

Definition sort3 : list T -> list T := sort3rec [::].
\end{coqcode}
\caption{A structurally-recursive non-tail-recursive mergesorts in \Coq, that takes three elements
from the input at a time. This implementation is omitted in the place marked by \coq{[..]} in
\cref{fig:struct-nontailrec-mergesort-coq}.}
\label{fig:struct-nontailrec-mergesort3-coq}
\begin{coqcode}
Fixpoint sortNrec (stack : list (list T)) (x : T) (xs : list T) {struct xs} :
    list T :=
  if xs is y :: xs then
    if leT x y then incr stack y xs [:: x] else decr stack y xs [:: x]
  else
    pop [:: x] stack
with incr (stack : list (list T)) (x : T) (xs accu : list T) {struct xs} :
    list T :=
  if xs is y :: xs then
    if leT x y then
      incr stack y xs (x :: accu)
    else
      sortNrec (push (catrev accu [:: x]) stack) y xs
  else
    pop (catrev accu [:: x]) stack
with decr (stack : list (list T)) (x : T) (xs accu : list T) {struct xs} :
    list T :=
  if xs is y :: xs then
    if leT x y then
      sortNrec (push (x :: accu) stack) y xs
    else
      decr stack y xs (x :: accu)
  else
    pop (x :: accu) stack.

Definition sortN (xs : list T) : list T :=
  if xs is x :: xs then sortNrec [::] x xs else [::].
\end{coqcode}
\caption{A smooth structurally-recursive non-tail-recursive mergesort in \Coq. This implementation
is omitted in the place marked by \coq{[..]} in \cref{fig:struct-nontailrec-mergesort-coq}.}
\label{fig:struct-nontailrec-mergesortN-coq}
\end{figure}

The \coq{sort2} function is an optimized variants of \coq{sort1} that perform sorting by
repetitively taking taking two elements from the input, locally sorting them, and pushing them to
the stack.
Similarly, the \coq{sort3} function in \cref{fig:struct-nontailrec-mergesort3-coq} takes three
elements from the input at a time which is analogous to the optimization technique employed by
\caml{List.stable_sort} of \OCaml (see the end of \cref{sec:tailrec-mergesort}).

The \coq{sortN} function in \cref{fig:struct-nontailrec-mergesortN-coq} is a smooth
structurally-recursive non-tail-recursive mergesort that takes the longest sorted prefix from the
input instead of a fixed number of elements.
The \coq{sortNrec} function determines whether the prefix is increasing or decreasing one, and calls
either \coq{incr} or \coq{decr} depending on that.
The \coq{incr} and \coq{decr} respectively takes the longest weakly increasing and strictly
decreasing slice from the input, push it to the stack, and call \coq{sortNrec} to process the next
sorted slice.
These three functions are made mutually recursive to pass the termination checker of \Coq.

\subsection{Tail-recursive mergesorts in \Coq}
\label{appx:tailrec-mergesort-in-coq}

\begin{figure}[t]
\begin{coqcode}
Module Type RevmergeSig.
Parameter revmerge :
  forall (T : Type) (leT : rel T), list T -> list T -> list T.
Parameter revmergeE : forall (T : Type) (leT : rel T) (xs ys : list T),
    revmerge leT xs ys = rev (path.merge leT xs ys).
End RevmergeSig.

Module Revmerge <: RevmergeSig.

Fixpoint merge_rec (T : Type) (leT : rel T) (xs ys accu : list T) {struct xs} :
    list T :=
  if xs is x :: xs' then
    (fix merge_rec' (ys accu : list T) {struct ys} :=
       if ys is y :: ys' then
         if leT x y then
           merge_rec leT xs' ys (x :: accu) else merge_rec' ys' (y :: accu)
       else
         catrev xs accu) ys accu
  else catrev ys accu.

Definition revmerge (T : Type) (leT : rel T) (xs ys : list T) : list T :=
  merge_rec leT xs ys [::].

Lemma revmergeE (T : Type) (leT : rel T) (xs ys : list T) :
  revmerge leT xs ys = rev (path.merge leT xs ys).

End Revmerge.
\end{coqcode}
\caption{Structurally-recursive tail-recursive merge in \Coq. The \coq{merge_rec} function in this
figure is the \Coq counterpart of \caml{revmerge} in \cref{fig:struct-tailrec-mergesort}, and
\coq{revmerge} in this figure instantiate it with \coq{[]} as the accumulator.}
\label{fig:struct-tailrec-merge-coq}
\end{figure}

\begin{figure}[p]
\begin{coqcode}
Module CBV_ (M : RevmergeSig).
Section CBV.
Context (T : Type) (leT : rel T).
Let geT x y := leT y x.

Fixpoint push (xs : list T) (stack : list (list T)) {struct stack} :
    list (list T) :=
  match stack with
  | [::] :: stack | [::] as stack => xs :: stack
  | ys :: [::] :: stack | ys :: ([::] as stack) =>
    [::] :: M.revmerge leT ys xs :: stack
  | ys :: zs :: stack =>
    [::] :: [::] :: push (M.revmerge geT (M.revmerge leT ys xs) zs) stack
  end.

Fixpoint pop (mode : bool) (xs : list T) (stack : list (list T)) {struct stack}
    : list T :=
  match stack, mode with
  | [::], true => rev xs
  | [::], false => xs
  | [::] :: [::] :: stack, _ => pop mode xs stack
  | [::] :: stack, _ => pop (~~ mode) (rev xs) stack
  | ys :: stack, true => pop false (M.revmerge geT xs ys) stack
  | ys :: stack, false => pop true (M.revmerge leT ys xs) stack
  end.

Fixpoint sort1rec (stack : list (list T)) (xs : list T) {struct xs} : list T :=
  if xs is x :: xs then
    sort1rec (push [:: x] stack) xs
  else
    pop false [::] stack.

Definition sort1 : list T -> list T := sort1rec [::].

Fixpoint sort2rec (stack : list (list T)) (xs : list T) {struct xs} : list T :=
  if xs is x1 :: x2 :: xs then
    let t := if leT x1 x2 then [:: x1; x2] else [:: x2; x1] in
    sort2rec (push t stack) xs
  else pop false xs stack.

Definition sort2 : list T -> list T := sort2rec [::].

[..]

End CBV.
End CBV_.
\end{coqcode}
\caption{Structurally-recursive tail-recursive mergesorts in \Coq.
The \coq{push}, \coq{pop}, \coq{sort1rec}, and \coq{sort1} functions in this figure are the \Coq
counterparts of \caml{push}, \caml{pop}, \caml{sort_rec}, and \caml{sort} in
\cref{fig:struct-tailrec-mergesort}.
Some optimized implementations are omitted here as marked by \coq{[..]} and shown in
\cref{fig:struct-tailrec-mergesort3-coq,fig:struct-tailrec-mergesortN-coq}.}
\label{fig:struct-tailrec-mergesort-coq}
\end{figure}

Similarly to the non-tail-recursive counterpart
(\cref{fig:struct-nontailrec-merge-coq,fig:struct-nontailrec-mergesort-coq} in
\cref{appx:nontailrec-mergesort-in-coq}), we define a module type \coq{RevmergeSig} and its instance
\coq{Revmerge} implementing a tail-recursive merge function (\cref{fig:struct-tailrec-merge-coq}),
and the tail-recursive mergesort functions are provided as a functor module \coq{CBV_} that takes an
instance \coq{M} of \coq{RevmergeSig} as a parameter (\cref{fig:struct-tailrec-mergesort-coq}).

\begin{figure}[p]
\begin{coqcode}
Fixpoint sort3rec (stack : list (list T)) (xs : list T) {struct xs} : list T :=
  match xs with
  | x1 :: x2 :: x3 :: xs =>
    let t :=
        if leT x1 x2 then
          if leT x2 x3 then [:: x1; x2; x3]
          else if leT x1 x3 then [:: x1; x3; x2] else [:: x3; x1; x2]
        else
          if leT x1 x3 then [:: x2; x1; x3]
          else if leT x2 x3 then [:: x2; x3; x1] else [:: x3; x2; x1]
    in
    sort3rec (push t stack) xs
  | [:: x1; x2] => pop false (if leT x1 x2 then xs else [:: x2; x1]) stack
  | _ => pop false xs stack
  end.

Definition sort3 : list T -> list T := sort3rec [::].
\end{coqcode}
\caption{A structurally-recursive tail-recursive mergesorts in \Coq, that takes three elements from
the input at a time. This implementation is omitted in the place marked by \coq{[..]} in
\cref{fig:struct-tailrec-mergesort-coq}.}
\label{fig:struct-tailrec-mergesort3-coq}
\begin{coqcode}
Fixpoint sortNrec (stack : list (list T)) (x : T) (xs : list T) {struct xs} :
    list T :=
  if xs is y :: xs then
    if leT x y then incr stack y xs [:: x] else decr stack y xs [:: x]
  else
    pop false [:: x] stack
with incr (stack : list (list T)) (x : T) (xs accu : list T) {struct xs} :
    list T :=
  if xs is y :: xs then
    if leT x y then
      incr stack y xs (x :: accu)
    else
      sortNrec (push (catrev accu [:: x]) stack) y xs
  else
    pop false (catrev accu [:: x]) stack
with decr (stack : list (list T)) (x : T) (xs accu : list T) {struct xs} :
    list T :=
  if xs is y :: xs then
    if leT x y then
      sortNrec (push (x :: accu) stack) y xs
    else
      decr stack y xs (x :: accu)
  else
    pop false (x :: accu) stack.

Definition sortN (xs : list T) : list T :=
  if xs is x :: xs then sortNrec [::] x xs else [::].
\end{coqcode}
\caption{A smooth structurally-recursive tail-recursive mergesort in \Coq. This implementation is
 omitted in the place marked by \coq{[..]} in \cref{fig:struct-tailrec-mergesort-coq}.}
\label{fig:struct-tailrec-mergesortN-coq}
\end{figure}

The \coq{sort2} and \coq{sort3} functions in
\cref{fig:struct-tailrec-mergesort-coq,fig:struct-tailrec-mergesort3-coq} take two and three
elements from the input at a time, respectively.
Again, the latter is analogous to the optimization technique employed by \caml{List.stable_sort} of
\OCaml (\cref{sec:tailrec-mergesort}).
The \coq{sortN} function in \cref{fig:struct-tailrec-mergesortN-coq} is a smooth
structurally-recursive tail-recursive mergesort, implemented in the same way as the
non-tail-recursive counterpart (\cref{fig:struct-nontailrec-mergesortN-coq}).

We stress that the recursive functions defined and used in
\cref{fig:struct-tailrec-merge-coq,fig:struct-tailrec-mergesort-coq,fig:struct-tailrec-mergesort3-coq,fig:struct-tailrec-mergesortN-coq}
are tail recursive except for \coq{push} in \cref{fig:struct-tailrec-mergesort-coq}, whose depth of
recursive calls is logarithmic in the length of the input.
Therefore, a linear consumption of stack space does not occur in any of the mergesort functions
presented in this section.

\end{document}